\documentclass[a4paper, 12pt]{article}

\usepackage[margin=2.5cm]{geometry}
\usepackage{setspace}
\usepackage{amsmath, amssymb, amsthm, bm, mathrsfs, empheq, xfrac}

\usepackage{gensymb}
\usepackage{bbm}
\usepackage{multirow,array}

\usepackage{graphicx}
\usepackage{float}
\usepackage{tikz}
\usepackage{pgfplots}
\pgfplotsset{width=10cm,compat=1.18}

\usepackage{dirtytalk}
\usepackage{color}

\usepackage{caption}
\usepackage{subcaption}

\usepackage{enumitem}

\usepackage[most]{tcolorbox}

\usepackage[T1]{fontenc}
\usepackage{ebgaramond}
\usepackage{newpxmath}

\usepackage[bookmarksnumbered=true]{hyperref}
\hypersetup{
    colorlinks = true,
    linkcolor = red,
    anchorcolor = blue,
    citecolor = blue,
    filecolor = blue,
    urlcolor = blue
}
\usepackage{natbib}

\newtheorem{definition}{Definition}
\newtheorem{proposition}{Proposition}

\newtheorem{lemma}{Lemma}
\newtheorem{theorem}{Theorem}

\newtheorem{example}{Example}

\begin{document}

\title{Reduced Forms: Feasibility, Extremality, Optimality%
\thanks{We thank Gabriel Carroll, Hector Chade, Daniele Condorelli, Francesc Dilm\'{e}, Mikhail Golosov, Marina Halac, Navin Kartik, Deniz Kattwinkel, Andreas Kleiner, Natalia Kovrijnykh, Elliot Lipnowski, Alejandro Manelli, Francesco Nava, Axel Niemeyer, Juan Ortner, Ludovic Renou, Fedor Sandomirskiy, Rakesh Vohra, Frank Yang, Kai Hao Yang, and participants at ASU, Bonn University, Caltech, HSE University, UCL, Warwick University, Yale University, EWMES 2025, and 4th IWGTEA for helpful comments. Krasikov thanks the NSF for support (award number 2416505).}%
}
\author{
\setcounter{footnote}{1}
    Pasha Andreyanov\thanks{
        HSE University, Moscow.  
        Email: \href{mailto:pandreyanov@gmail.com}{pandreyanov@gmail.com}.
    }
    \and Ilia Krasikov\thanks{
        Arizona State University.  
        Email: \href{mailto:krasikovis.main@gmail.com}{krasikovis.main@gmail.com}.
    }
    \and Alex Suzdaltsev\thanks{
        HSE University, Saint Petersburg.  
        Email: \href{mailto:asuzdaltsev@gmail.com}{asuzdaltsev@gmail.com}.
    }
}
\date{September 2026}
\maketitle

\begin{abstract}
We study independent private values auction settings with heterogeneous bidders, revenue nonlinear in interim winning probabilities, and place no ad hoc restrictions on auctions that a seller can use.
We show that feasibility of interim winning probabilities is testable along a unidimensional \emph{principal curve}, constructed in closed form.
We then explicitly characterize the extreme points of the feasible set and the allocations that induce them.
Combining these results, we solve for the optimal auction under a natural regularity condition.
The optimal auction equalizes bidders' marginal revenues along its principal curve and awards the good to the bidder whose type the curve reaches first.
Revenue curvature determines whether and how low types are excluded---at the intensive margin, scaling winning probabilities down, or the extensive margin, cutting them off.
Applications include the linear model, endogenous valuations from ex ante investments, and non-expected utility preferences, where previous results were largely limited to symmetric environments and allocations or two bidders.

\end{abstract}

\section{Introduction}
\label{sec: intro}

Auctions are among the most widely used market institutions.
Governments use them to allocate spectrum licenses and procurement contracts, firms employ them to acquire or liquidate inputs, and millions of individuals participate daily on online platforms.
A central question in auction theory is how to design auctions that maximize an auctioneer's objective such as expected revenue.

The benchmark case of independent private values (IPV) with linear utilities has been well understood since at least \citet{myerson1981optimal}.
Myerson's celebrated analysis shows that, in this environment, the optimal auction awards the good to the bidder with the highest positive virtual value under mild regularity conditions.
A key feature that makes this analysis tractable is that expected revenue is linear in bidders' interim winning probabilities.
This linearity allows the auctioneer's problem to be solved pointwise over feasible allocations in regular environments when ironing is not needed.

Much less is known once the auctioneer's objective departs from linearity.
Nonlinear dependence on interim winning probabilities arises naturally in economically important settings.
For instance, it arises when bidders undertake ex ante investments that affect valuations, as in models of endogenous valuations \citep{gershkov2021theory}, or when bidders have non-expected-utility preferences such as constant risk aversion (CRA) \citep{gershkov2022optimal}.
In these environments, because the auctioneer's revenue depends nonlinearly on interim winning probabilities, the convenient pointwise maximization technique of the linear benchmark is no longer available.
As a result, classical methods break down, and the analysis of optimal auctions becomes substantially more complex.

This paper develops a unified framework for studying optimal auctions when the auctioneer's expected revenue can be written in reduced form as an integral of bidder-specific revenue terms.
These terms may be nonlinear in interim winning probabilities.
We allow for full asymmetry across bidders and do not impose any ad hoc restrictions on the auctions available to the auctioneer.
Our framework accommodates a wide range of economic environments, including the linear IPV model, settings with endogenous valuations, and settings with CRA preferences cited above.

The core difficulty in nonlinear environments is the loss of separability: the auctioneer's problem no longer decomposes across types, and allocations cannot be chosen type by type.
Making progress therefore requires resolving three intertwined challenges.

\textit{Feasibility.}
Bidders' interim winning probabilities must be consistent with some ex post allocation, and characterizing this consistency requires checking Border's condition.
This condition provides a set of necessary and sufficient inequalities, but this set is high-dimensional---one inequality for every possible cutoff vector of types.
While tractable for symmetric auctions or with two bidders, it quickly becomes unmanageable as we move to asymmetric auctions with a larger number of bidders.
This leaves little hope of solving the fully constrained problem directly.

\textit{Extremality.}
The set of feasible reduced forms---tuples of bidders' interim winning probabilities---is convex.
Consequently, when the auctioneer's objective is convex---as in many environments with endogenous valuations---an optimum is attained at an extreme point of the feasible set.
A tractable description of extremal reduced forms could therefore sharply restrict the set of candidates and reduce the subsequent maximization problem to a much lower-dimensional search.
Yet, existing characterizations are largely confined to symmetric settings or to two bidders, leaving the structure of extreme points poorly understood in richer settings.

\textit{Optimality.}
Ultimately, the core obstacle is that nonlinearities destroy the pointwise structure of the linear benchmark.
Feasibility and extremality help identify the relevant constraint set and restrict attention to plausible candidates, but the remaining task is still to optimize a nonlinear objective over that set.
Doing so requires a methodology that both solves the auctioneer's problem and delivers structural insights about the optimal auction.

In this paper, we address these challenges by developing new characterizations for feasibility and extremality and then leveraging them to solve the auctioneer's problem.

Our first contribution is a new representation of feasibility that collapses Border's high-dimensional family of constraints to a unidimensional test.
Its building block is a novel transformation of CDFs, the \emph{$\delta$-transform}, which reparametrizes each bidder's interim winning probability as a curve in the type-probability square along which the product of type and probability decays at a constant rate.
Applying it coordinatewise to a reduced form yields a path inside the unit cube---its \emph{principal curve}---along which all bidders' type-probability products are equalized.
We show that Border's inequalities hold everywhere if and only if they hold along this single curve, so that feasibility reduces to checking the sign of one scalar function of one variable.
This makes feasibility straightforward to verify even in fully asymmetric environments with many bidders.

Our second contribution is a sharp characterization of extremality.
Building on our feasibility representation, we show that a reduced form is extremal precisely when Border's constraints bind along its principal curve.
This criterion dramatically simplifies the auctioneer's problem in environments where the optimum is known to be extremal.
It also clarifies how to induce reduced forms through ex post allocations.
Specifically, we show that every extremal reduced form can be implemented by a simple \emph{score allocation}.
In a score allocation, each bidder's type is mapped into a unidimensional score, and the good is awarded to the bidder with the highest positive score, with ties occurring with probability zero.
We prove that score allocations are exactly the allocations that induce extremal reduced forms.
Moreover, the allocation inducing a given extremal reduced form is unique and is obtained from its \emph{canonical scores}---each type times its winning probability---so that extremality can be tested without any search over scores.
This also yields a sharper form of the classical equivalence between Bayesian and dominant-strategy implementation originally established by \citet{manelli2010bayesian}.
In particular, we explain that any feasible reduced form is implementable as a mixture of deterministic, dominant-strategy incentive compatible, and nonbossy allocations, in the sense of \citet{satterthwaite1981strategy}.

Our third contribution concerns optimality in nonlinear environments.
We first show that, in regular environments where it is optimal to always sell the good, the optimum is a score allocation whose scores are constructed to equalize bidders' marginal revenues along the principal curve.
Unlike in linear settings, where the marginal revenue equalization directly yields the solution (\citet{myerson1981optimal} and \citet{bulow1989simple}), it is a fixed-point problem in nonlinear settings.
This is because each bidder's marginal revenue depends on own interim winning probability, and bidders' interim winning probabilities are interconnected through feasibility.
The principal curve untangles this fixed point: expressed in terms of $\delta$-transforms, the equalization becomes a system of equations that can be solved at each point of the curve separately, directly from the primitives.
To certify optimality, we develop a new methodology that recasts the auctioneer's problem as a continuous-time optimal control problem.
In this problem, the principal curve is built over ``time'' moving from high to low types, the auctioneer chooses how quickly to allocate winning probability to each bidder, and Border's condition enters as a dynamic feasibility constraint along the curve.
The system describing the marginal revenue equalization reappears there as the first-order condition, and regularity---which extends Myerson's regularity to nonlinear settings so that no ironing is needed---guarantees its sufficiency.

We show how this characterization applies to the endogenous-valuation model, where bidders choose ex ante investments that affect their valuations.
In this environment, interim winning probabilities matter not only for rent extraction but also because they determine the return to investing.
Any bidder with a positive interim winning probability has an incentive to invest, even though only one ultimately wins, so investment costs are largely duplicated.
The optimal auction departs from Myerson's auction by shifting winning probability toward the strongest bidder and away from the weakest one until bidders' marginal revenues are equalized along the optimal principal curve.
Relative to Myerson's auction, this reallocation reduces excessive duplication of investment costs, thereby increasing the auctioneer's expected revenue.

When exclusion is optimal, its form depends on the curvature of bidders' revenue functions in interim winning probabilities.
In concave environments, as in settings with risk-averse bidders, exclusion operates at the intensive margin.
At low types, the marginal revenue from awarding the full unit can be negative, and, rather than excluding such types from trade altogether, the auctioneer awards them the fraction at which marginal revenue is exactly zero.
We illustrate this case in the CRA model with two groups of bidders, risk-neutral and risk-averse.
Letting risk-averse bidders trade at scaled-down probabilities at the bottom of the type distribution is profitable, since it lowers information rents at a small efficiency cost, and the optimal auction favors the risk-averse bidders at low types and the risk-neutral bidders at high types.
This echoes Myerson's logic that the optimal auction should tilt toward ex ante ``weak'' bidders to intensify competition, with the twist that, in the nonlinear CRA model, who is ``weak'' depends on interim winning probabilities.
At high interim winning probabilities, risk-averse bidders are effectively stronger, since they are willing to pay a premium to secure the good with near certainty; at low interim winning probabilities, the ranking is reversed.

In convex environments, as in the endogenous-valuation model, exclusion operates at the extensive margin: the auctioneer either awards the good with certainty or cuts the type off.
Where to cut is determined by a trade-off that is absent in the linear model.
Excluding a bidder's low types relaxes feasibility for the other bidders, whose winning probabilities can then be raised.
As a result, bidders are excluded at positive and bidder-specific marginal revenues, as opposed to the linear model, where every bidder is excluded at the common threshold of zero.

\subsection{Relation to literature}

A substantial body of literature has studied feasibility, extremality, and optimality of reduced forms.
Below, we describe the key existing contributions and relate our results to them.

\textit{Feasibility.}
In a classic result, \citet{border1991implementation}, building on \citet{maskin1984optimal} and \citet{matthews1984implementability}, characterized the set of feasible symmetric reduced forms---that is, symmetric interim winning probabilities inducible by some ex post allocation---via a unidimensional family of linear inequalities.
\citet{border2007reduced} and \citet{mierendorff2011asymmetric} generalized this result to asymmetric reduced forms, where feasibility is described by a family of inequalities whose dimension grows with the number of bidders.
\citet{che2013generalized}, in turn, extended these results to multiunit auction settings with capacity constraints.
\citet{hart2015implementation} related Border's conditions to second-order stochastic dominance, and \citet{he2021private} showed that Border-type inequalities can also characterize the set of feasible belief distributions generated by private-private information structures.

The closest paper to ours regarding feasibility in auction settings is \citet{cai2011constructive}, who showed that, in finite (potentially only partially ordered) type spaces, one can reduce the family of asymmetric Border inequalities to one dimension.
Our result is complementary to theirs: it relies on additional structure---continuous, unidimensional, linearly ordered types and monotone interim winning probabilities---but delivers a sharper characterization.
In this setting, we characterize feasibility via a unidimensional monotone curve, which we describe in closed form through an explicit transformation of reduced forms.

\textit{Extremality.}
Since a quasiconvex functional attains its maximum at an extreme point of a convex compact set,\footnote{This holds under continuity by a generalization of Bauer's maximum principle to quasiconvex functionals; see \citet{ball2023bauer}.} characterizing extreme points has become an increasingly powerful approach.
It allows one to solve many optimization problems of economic interest at once.
In a landmark contribution, \citet{kleiner2021extreme} characterized extreme points of the set of monotone functions that majorize, or are majorized by, a given monotone function.
Aside from applications to contests, optimal persuasion, and delegation, this yielded an immediate description of extreme points of the set of feasible symmetric reduced forms, since the symmetric Border constraint could be viewed as a majorization constraint.\footnote{\citet{kleiner2021extreme} also covered multiunit symmetric environments.}

Extreme points of asymmetric reduced forms in auction settings are still not fully understood.
A notable contribution is \citet{manelli2010bayesian}, who characterized some extremal asymmetric reduced forms---those that are step functions---and used this description to establish the BIC--DIC equivalence.
More recently, \citet{yang2025multidimensional} provided an equivalent characterization of extremal asymmetric reduced forms in the two-bidder case.
Their approach differs from ours: it characterizes extreme points of the set of rationalizable tuples of monotone functions.
While the connection between rationalizability and majorization is transparent with two bidders, it is unclear how to extend it to characterize extremal reduced forms when the number of bidders exceeds two.
In contrast, our approach delivers an explicit description of all extreme points irrespective of the number of bidders.

A few other papers used the extreme-points approach to solve problems of economic significance.
Among them, \citet{manelli2007multidimensional} employed it to study multidimensional screening.
\citet{nikzad2023constrained} and \citet{candogan2023optimal} characterized extreme points of the same set as in \citet{kleiner2021extreme} but with additional linear constraints and applied this approach to study optimal information disclosure.
\citet{arieli2023optimal} described extreme points of the set of mean-preserving contractions of a given prior.
\citet{kleiner2024extreme} studied extreme points of the set of fusions, the multidimensional analog of mean-preserving contractions, and
\citet{lahr2024extreme} studied extreme points in multidimensional monopolistic screening settings.
Finally, \citet{yang2024monotone} characterized extreme points of the set of monotone functions bounded by two given monotone functions and applied this to political economy, persuasion, the psychology of judgment, and security design.

\textit{Optimality.}
As noted above, once the objective is nonlinear in interim winning probabilities, the optimization problem becomes hard because pointwise optimization is no longer available.
Despite these difficulties, several papers have made substantial progress.
\citet{gershkov2021theory} leveraged the Fan--Lorentz inequality to solve for the optimal symmetric allocation with many goods when valuations are endogenous and provided conditions under which the globally optimal reduced form is symmetric in two-bidder environments;
\citet{scoringandfavoritism} provided analysis for a related setting with contractible actions.
\citet{zhang2017auctions} characterized the (asymmetric) optimum in the symmetric two-bidder endogenous-valuations model with quadratic costs, whereas \citet{scoringcompanion} solved for the optimum in general symmetric two-bidder settings with endogenous valuations.
Finally, \citet{gershkov2022optimal} characterized the optimal auction in symmetric environments with non-expected-utility bidders exhibiting constant risk aversion.
Our approach covers these applications and unifies and extends the analysis to general asymmetric settings.

\subsection{Structure of the paper}
\label{sec: structure}

The paper contains several distinct results, and, to help the reader navigate, we indicate key theorems and propositions in brackets throughout this outline.
Because the main objects and transformations we introduce are nonstandard, and later arguments rely on this shared language, we recommend reading the paper linearly.

Section \ref{sec: setting} introduces the model and explains how several leading auction environments fit into our framework.
Section \ref{sec: example} presents an illustrative example that previews the economics of the optimal auction, and Section \ref{sec: preliminaries} collects the simplifications used throughout.

Sections \ref{sec: feasibility}--\ref{sec: optimality} develop the main methodological results.
Specifically, Section \ref{sec: feasibility} reduces Border's high-dimensional feasibility constraints to a unidimensional test along the principal curve (Theorem \ref{thm: feasibility}).
Section \ref{sec: extremality} characterizes extremal reduced forms and links them to implementation via score allocations (Theorems \ref{thm: extremality} and \ref{thm: score allocations}).
Section \ref{sec: optimality} characterizes the optimal auction in regular environments when it is optimal to always sell the good (Theorem \ref{thm: optimal no exclusion}) and applies it to discrimination in the endogenous-valuation model.
It also outlines the methodology for certifying optimality via optimal control (Proposition \ref{prop: control reformulation}).

Section \ref{sec: optimal exclusion} extends the characterization to environments in which exclusion is optimal, distinguishing between concave and convex environments (Theorems \ref{thm: optimal concave} and \ref{thm: optimal convex}), and applies the concave case to discrimination in the CRA-preferences model.

Section \ref{sec: conclusion} concludes, and Appendix \ref{sec: appendix} and Supplementary Appendix \ref{sec: supplementary appendix} collect proofs and supplementary results.

\section{Setting}
\label{sec: setting}

We consider an auction environment with $n$ bidders, indexed by $i = 1, \ldots, n$, and a single auctioneer who sells one indivisible unit of a good.
Each bidder $i$ has a private unidimensional type $\theta_i \in \Theta_i \subset \mathbf{R}_+$, where $\Theta_i$ is a compact interval.
The type $\theta_i$ is drawn from a distribution $F_i$ on $\Theta_i$ that admits a strictly positive and continuous density denoted by $f_i$.
Types are independently distributed across bidders, although their marginal distributions may differ.
We use boldface for vectors of individual variables, e.g.,
$\bm\theta = (\theta_1,\ldots,\theta_n)$ stands for a type profile.

Let $\bm\Theta = \Theta_1 \times \cdots \times \Theta_n$ denote the set of all possible type profiles.
An \textbf{allocation} is a function $\bm z = (z_1,\ldots,z_n): \bm\Theta \to [0,1]^n$ satisfying
\begin{align}
\label{feasibility z}
\sum_{i=1}^n z_i(\bm\theta) \leq 1
\end{align}
for all type profiles $\bm\theta \in \bm\Theta$.
Here, $z_i(\bm\theta)$ represents the probability that bidder $i$ receives the good when the type profile is $\bm\theta$, and the constraint simply means that the auctioneer can sell at most one unit.
For bidder $i$, the \textbf{interim winning probability} given their type $\theta_i$ is defined by
\begin{align}
\label{x from z}
x_i(\theta_i) = \mathbf{E}[z_i \mid \theta_i],
\end{align}
where the expectation is taken with respect to the independent draws of the other bidders' types.

The setting described thus far corresponds to the standard IPV auction model with linear utilities.
It is well known that, in this setting, the auctioneer's expected revenue can be expressed as the following function of interim winning probabilities:
\[
\sum_{i=1}^n \int_{\Theta_i} \Bigg(\underbrace{\theta_i - \frac{1-F_i(\theta_i)}{f_i(\theta_i)}}_{=:\mathrm{MVV}_i(\theta_i)}\Bigg) x_i(\theta_i) dF_i(\theta_i),
\]
where the term in parentheses is often referred to as \textbf{Myerson's virtual value} (henceforth \textbf{MVV}).

In contrast to the standard model, we remain agnostic about bidders' preferences and instead represent the auctioneer's implied revenue in reduced form.
Specifically, let $J_i(x_i,\theta_i)$ denote the revenue obtained from selling the good to type $\theta_i$ with probability $x_i$, so that the expected revenue can be written as
\begin{align}
\label{revenue}
\sum_{i=1}^n \int_{\Theta_i} J_i(x_i(\theta_i),\theta_i) dF_i(\theta_i).
\end{align}
This specification subsumes the standard model studied by Myerson, in which bidder $i$'s \textbf{marginal revenue} $\frac{\partial J_i}{\partial x_i}$ (henceforth \textbf{MR}) coincides with that bidder's MVV, and it is far more general.
For example, this framework accommodates environments with endogenous valuations and bidders with CRA preferences discussed below.

We are interested in finding the revenue-maximizing auction when the auctioneer's revenue is given by \eqref{revenue} subject to bidders' interim winning probabilities, defined by \eqref{x from z}, being weakly increasing.
The monotonicity requirement means that the auctioneer can implement any allocation in which bidders' interim winning probabilities are weakly increasing in their types, and is restricted to search among such allocations.
This condition is the only substantive assumption we impose, and, as we explain below, it is satisfied in the standard model as well as in settings with endogenous valuations and non-expected-utility CRA preferences.
In all these cases, monotonicity is sufficient for implementability, and in the standard model and with endogenous valuations it is also necessary, as a direct consequence of bidders' incentive compatibility.
To streamline the exposition, we assume that $J_i$ is twice continuously differentiable, and we normalize $J_i(0,\theta_i) = 0$ for all $\theta_i$.
This normalization is without loss of generality, since the term $\int_{\Theta_i}J_i(0,\theta_i)dF_i(\theta_i)$ does not depend on the allocation.
In the rest of this section, we show how our general framework encompasses several important auction environments studied in the literature.

\subsection{IPV setting with linear utilities}
\label{sec: linear}

This is the classical and most extensively studied case, already mentioned above.
Bidder $i$'s utility takes the form $\theta_i z_i - t_i$, where $t_i$ denotes bidder $i$'s payment to the auctioneer, and the auctioneer seeks to maximize the expected sum of bidders' payments.
It is well known that only allocations in which bidders' interim winning probabilities are monotone can be implemented in Bayesian Nash equilibrium.
In this environment, the auctioneer's revenue coincides with \eqref{revenue} with
\[
J_i(x_i,\theta_i) = \mathrm{MVV}_i(\theta_i) x_i
\]
up to a constant.
This constant is pinned down by the participation constraint of the lowest types, which binds at zero, so the auctioneer's problem is exactly of the form described above.

The solution has been well understood since \citet{myerson1981optimal}.
In the regular case---where bidders' MVVs are strictly increasing---the good is allocated to the bidder with the highest positive MVV.\footnote{More generally, the optimal allocation is determined by ironed MVVs and may involve randomization among tied bidders.}
To see this, note that, due to linearity in interim winning probabilities, we can rewrite the objective in \eqref{revenue} as
\begin{align}
\label{revenue z}
\sum_{i=1}^n \int_{\Theta_i} J_i(x_i(\theta_i),\theta_i) dF_i(\theta_i) = \mathbf{E}\left[\sum_{i=1}^n \mathrm{MVV}_i \cdot z_i\right],
\end{align}
which can be easily maximized pointwise in the space of allocations satisfying \eqref{feasibility z}.

\subsection{IPV setting with CRA preferences}
\label{sec: CRA}

In an important contribution, \citet{gershkov2022optimal} (henceforth GMSZ-CRA) extended optimal auction theory to environments in which bidders' preferences over interim winning probabilities depart from expected utility.\footnote{Their standing assumption is quasilinearity in money.}
Under constant risk aversion, they showed that the auctioneer's revenue can still be expressed as in \eqref{revenue}, with
\[
J_i(x_i,\theta_i) = \theta_i x_i - \frac{1-F_i(\theta_i)}{f_i(\theta_i)} g_i(x_i)
\]
up to a constant.
Here, $g_i(x_i)$ is the certainty equivalent of a lottery paying \$1 with probability $x_i$ and \$0 otherwise.
If $g_i'' \equiv 0$, bidder $i$ has expected-utility preferences; if $g_i'' \geq 0$, bidder $i$ is risk-averse.

Monotonicity of $\bm x$ remains sufficient for implementability.%
\footnote{With general CRA preferences, monotonicity need not be necessary; see footnote 3 in GMSZ-CRA, who impose enough structure for the optimum to be monotone.}
Thus, the auctioneer's problem retains the same structure as before, but it is considerably harder because $J_i$ is generally nonlinear in $x_i$.
As a result, one cannot express the auctioneer's expected revenue as a function of the allocation similarly to \eqref{revenue z} and maximize it pointwise.

To gain tractability, GMSZ-CRA focused on symmetric settings in which $J_i$ and $F_i$ are identical across bidders.
Under risk aversion, $J_i$ is concave in $x_i$, which was their main case of interest.
They showed that symmetry of the environment, combined with concavity of $J_i$, implies that the optimal allocation is symmetric, and they derived a closed-form characterization of the corresponding interim winning probabilities.\footnote{GMSZ-CRA also described an approach to the optimal symmetric auction when bidders are symmetric but not necessarily risk-averse.}
Their analysis, however, does not cover asymmetric environments, which we take up in this manuscript.

\subsection{IPV setting with endogenous valuations}
\label{sec: EVs}

In another influential paper, \citet{gershkov2021theory} (henceforth GMSZ-EV) analyzed models with ex ante investments and showed that they naturally generate preferences that are convex in interim winning probabilities.
In their framework, before the auction, each bidder may take a private action $a_i$ that increases their valuation from winning.
Specifically, bidder $i$'s utility is
\[
v_i(a_i,\theta_i) z_i - C_i(a_i) - t_i,
\]
where $C_i$ is a convex cost incurred regardless of the auction outcome.

GMSZ-EV demonstrated that, when bidders are expected-utility maximizers, the auctioneer's expected revenue again takes the form in \eqref{revenue}, with
\[
J_i(x_i,\theta_i) = \omega_i(x_i,\theta_i) - \frac{1-F_i(\theta_i)}{f_i(\theta_i)} \frac{\partial}{\partial\theta_i} \omega_i(x_i,\theta_i),
\]
up to a constant, where
\[
\omega_i(x_i,\theta_i) := \max_{a_i} \, v_i(a_i,\theta_i)x_i - C_i(a_i)
\]
is the maximal expected utility of type $\theta_i$ given their interim winning probability $x_i$.
Since $\omega_i$ is a pointwise maximum of functions linear in $x_i$, it is convex in $x_i$, and typically $J_i$ inherits this convexity.%
\footnote{It may happen that $J_i$ itself is nonconvex in $x_i$, whereas its positive part $\max\{0,J_i\}$ is convex.
As explained in GMSZ-EV, the difference is inconsequential for the structure of the optimum under suitable regularity; convexity of the positive part is also the condition under which we characterize the optimum with exclusion in Section \ref{sec: optimal convex}.}

Finding the optimal auction therefore requires maximizing the convex functional in \eqref{revenue} over monotone interim winning probabilities $\bm x$ that are consistent with an allocation, i.e., that satisfy \eqref{feasibility z} and \eqref{x from z}.
Very little is known about problems of this form.
Notable exceptions include \citet{zhang2017auctions}, who found the optimum in an additively separable example with $n=2$; GMSZ-EV, who analyzed optimal symmetric allocations (and their global optimality) in symmetric environments; and \citet{scoringcompanion}, who characterized the asymmetric optimum with two bidders.
To the best of our knowledge, no results exist for asymmetric environments with more than two bidders.
To address this gap, we develop a novel approach to the study of optimal auctions in general nonlinear settings without imposing symmetry or restricting the auctioneer to symmetric auctions.

\subsection{Other applications}
\label{sec: other applications}

We have cast the problem as revenue maximization by an auctioneer and keep this language throughout, but it is a matter of convenience only.
Nothing in our analysis uses the interpretation of $J_i$ as revenue, and the results in the paper directly apply to a wide range of environments.
In auction settings, this covers expected welfare and other redistributive objectives with type-dependent welfare weights, all of which can be brought to the form \eqref{revenue} by adjusting $J_i$, e.g., see \citet{dworczak2021redistribution}.
We now describe several settings outside of auctions that are covered as well.

\emph{Procurement with pre-auction quality.}
Consider a buyer who awards a single contract to one of $n$ suppliers, as in \citet{scoringandfavoritism}.
Supplier $i$ privately knows their type $\theta_i$ and, after learning it but before bidding, chooses a contractible quality $q_i$ at an investment cost $C_i^{\rm I}(q_i,\theta_i)$ that is sunk regardless of the outcome.
If the supplier wins, they incur a production cost $C_i^{\rm P}(q_i,\theta_i)$ and the buyer obtains the value $V_i(q_i)$.
Replacing the costs with their virtual counterparts to account for information rents, as in Section \ref{sec: EVs}, and optimizing over quality given supplier $i$'s interim winning probability $x_i$ yields
\[
J_i(x_i,\theta_i) = \max_{q_i} \;(V_i(q_i) - \tilde C_i^{\rm P}(q_i,\theta_i))x_i - \tilde C_i^{\rm I}(q_i,\theta_i),
\]
where $\tilde C_i^{\rm P}$ and $\tilde C_i^{\rm I}$ are the virtual production and investment costs.

\emph{Information design under privacy.}
\citet{he2021private} study information structures in which agents' signals about a common state are independent of one another---each agent learns about the state but nothing about what the others know---and observe that feasibility of such information structures is closely related to feasibility of interim allocations in mechanism design settings, with agents' types playing the role of signals and the outcome of the mechanism playing the role of the state.
Under this mapping, our setting corresponds to a designer who selects at most one of $n$ agents---a candidate to promote, a project to fund---and sends each agent an independent private signal about this choice.
Bidder $i$'s interim winning probability $x_i(\theta_i)$ is then agent $i$'s posterior probability of being selected, the feasibility constraint in \eqref{feasibility z}, \eqref{x from z} is precisely the condition for the profile of posterior distributions to be jointly feasible, and monotonicity of $x_i$ is a normalization rather than a constraint, since signals can be relabeled by the quantiles of the posteriors they induce.

\emph{Settings without money.}
Finally, transfers play no role in the framework beyond the derivation of $J_i$ in the examples above.
What is required is an objective of the form \eqref{revenue} and monotonicity of interim winning probabilities.
Both may arise in allocation problems without money in which the designer screens agents through non-monetary instruments, e.g., see \citet{benporath2014optimal} for a related model of costly verification.

\section{Illustrative example}
\label{sec: example}

Since our approach requires several steps to develop, it is useful to first preview the main ideas in a concrete economic example.
The purpose of this example is threefold.
First, it illustrates how nonlinearities in $J_i$ can be microfounded through bidders' ex ante investments.
Second, it explains why Myerson's auction can be suboptimal once valuations are endogenous.
Third, it previews the structure of the optimal auction that reallocates probability so as to equalize bidders' MRs.
Throughout, we work with $n$ bidders---one strong and $n-1$ symmetric weak ones.
Already for $n = 3$, this environment lies beyond the reach of existing results, which are confined either to symmetric settings or to two bidders.

Consider an endogenous-valuation model in which each bidder $i$ can take a private action $a_i$ at cost $a_i^2$, which increases their own valuation from winning the good to $v_i(a_i,\theta_i) = 2\sqrt{\theta_i}a_i$.
The maximal utility of bidder $i$ is then $\omega_i(x_i,\theta_i) = \theta_i x_i^2$.
Relative to the standard linear model, the key difference is the quadratic dependence of utility on the interim winning probability, which arises endogenously from bidders' optimal ex ante investment decisions.
It is easy to see that bidders' revenue functions inherit the quadratic dependence on their interim winning probability, i.e.,
\[
J_i(x_i,\theta_i) = \mathrm{MVV}_i(\theta_i) x_i^2.
\]

In this environment, Myerson's auction is generally suboptimal.
The reason is that such an allocation discriminates against stronger bidders so as to reduce their information rents (e.g., see \citet{carroll2019robustly}), thereby inducing excessive duplication of investment costs at the ex ante stage.
By reallocating winning probability away from weaker bidders, the auctioneer can mitigate over-investment and increase expected revenue.

To see this more concretely, suppose that bidders' type distributions satisfy%
\footnote{To obtain $F_i$, remark that its inverse $v_i$ satisfies $v_i(u)-(1-u)v_i'(u) = u^{1/\beta_i}$. Solving this equation yields $v_i(u)=\frac{1}{1-u}\int_{u}^1 s^{1/\beta_i}ds$.}
\[
\mathrm{MVV}_i(\theta_i) = (F_i(\theta_i))^{1/\beta_i}
\]
for some $\beta_i \in (0,1)$, where a larger value of $\beta_i$ is indicative of a higher bidder $i$'s strength in the sense of hazard rate ordering.%
\footnote{For two cumulative distribution functions, $F$ and $G$, with the same support, $F$ is stronger in the sense of hazard rate ordering than $G$ if $\frac{1-F}{f} \geq \frac{1-G}{g}$.
See Section \ref{sec: EV applications} for the general definition and discussion of this order.}
Specifically, suppose that bidder $1$ is strong with $\beta_1 = \frac{7}{8}$, whereas the remaining $n - 1$ weak bidders have $\beta_2 = \cdots = \beta_n = \frac{1}{2}$.
Then, bidders' interim winning probabilities in Myerson's auction $\bm x^{\rm M}$ are given by
\[
x_1^{\rm M}(\theta_1) = \mathbf{Pr}\Big(\mathrm{MVV}_1 \geq \max_{i \neq 1}\mathrm{MVV}_i|\theta_1\Big) = (F_1(\theta_1))^{\frac{4(n-1)}{7}},
\qquad
x_i^{\rm M}(\theta_i) = (F_i(\theta_i))^{\frac{4n-1}{4}} \;\; \forall i \neq 1.
\]
Direct calculations show that Myerson's auction yields revenue of
\[
\int_{\Theta_1} J_1(x_1^{\rm M}(\theta_1),\theta_1)dF_1(\theta_1) + \sum_{i \neq 1}\int_{\Theta_i} J_i(x_i^{\rm M}(\theta_i),\theta_i)dF_i(\theta_i) = \frac{7}{8n+7} + (n-1)\cdot\frac{2}{4n+5},
\]
which equals $\frac{243}{527} \approx 0.461$ for $n = 3$.
In contrast, trading exclusively with the strongest bidder---bidder $1$---gives
\[
\int_{\Theta_1} J_1(1,\theta_1)dF_1(\theta_1) + \sum_{i \neq 1}\int_{\Theta_i} J_i(0,\theta_i)dF_i(\theta_i) = \frac{7}{15} + (n-1) \cdot 0 = \frac{7}{15} \approx 0.467
\]
independently of $n$.
So, for $n = 3$, Myerson's auction is strictly dominated even by the simple policy of always trading with the strong bidder.%
\footnote{It can be verified that this domination occurs if and only if $n \leq 4$.
However, the optimal auction constructed below strictly dominates both Myerson's auction and trading exclusively with bidder $1$ for every $n$.}

Trading only with bidder $1$, however, is still suboptimal.
Excluding the weak bidders entirely forgoes the surplus from letting them win at the top of their type distributions, where their marginal contributions to the auctioneer's revenue can exceed bidder $1$'s.
Concretely, consider instead the allocation in which bidder $1$ wins when $(F_1(\theta_1))^{\frac{1}{7}} \geq \max_{i \neq 1} F_i(\theta_i)$, and otherwise the good goes to the weak bidder with the highest type.
The induced interim winning probabilities $\bm x^*$ are
\[
x_1^*(\theta_1) = (F_1(\theta_1))^{\frac{n-1}{7}},
\qquad
x_i^*(\theta_i) = (F_i(\theta_i))^{n+5} \;\;\forall i \neq 1,
\]
and the auctioneer's expected revenue is
\[
\int_{\Theta_1} J_1(x_1^*(\theta_1),\theta_1)dF_1(\theta_1) + \sum_{i \neq 1}\int_{\Theta_i} J_i(x_i^*(\theta_i),\theta_i)dF_i(\theta_i) = \frac{7}{2n+13} + (n-1)\cdot\frac{1}{2n+13},
\]
which exceeds the revenue from Myerson's auction and exclusive trade for every $n$, and equals $\frac{9}{19} \approx 0.474$ for $n = 3$.
Relative to Myerson's auction, this allocation shifts winning probability toward the strongest bidder---$x_1^* > x_1^{\rm M}$ and $x_i^* < x_i^{\rm M}$ for all $i \neq 1$ and every $n$---but stops short of exclusion, so that the weak bidders continue to win at high types.

The distinguishing feature of this auction is that the bidder with the highest MR wins the good with certainty, that is,
\begin{align}
\label{eq: fixed point x}
x_i^*(\theta_i) =
\mathbf{Pr}\left(\frac{\partial}{\partial x_i}J_i(x_i^*,\theta_i) \geq \frac{\partial}{\partial x_j} J_j(x_j^*,\theta_j) \;\;\forall j \neq i
|\theta_i\right).
\end{align}
In particular, bidders' MRs are equalized whenever they are tied, $\frac{\partial J_i}{\partial x_i} = \frac{\partial J_j}{\partial x_j}$.
This ranking principle is familiar from \citet{bulow1989simple} interpretation of Myerson’s auction in linear settings.
There, our $\frac{\partial J_i}{\partial x_i}$ coincides with $\mathrm{MVV}_i(\theta_i)$, which they interpret as marginal revenue in an associated monopoly problem. 
Because this expression does not depend on the interim winning probability, the ranking directly determines the optimal allocation.

In our nonlinear example, by contrast, bidders' MRs in \eqref{eq: fixed point x} depend on their interim winning probabilities, so this equation is a genuine fixed-point problem.
It is natural to conjecture that a fixed point $\bm x^*$ characterizes the revenue-maximizing auction.
As we show later, this conjecture is correct in regular nonlinear environments when it is optimal to always sell the good (as happens in this example).
The verification is nontrivial, however, since local perturbations do not certify global optimality due to convexity of $J_i$, and admissible perturbations must remain implementable by some ex post allocation.

Our approach overcomes these difficulties by changing variables.
Rather than solving directly for interim winning probabilities, we recast the auctioneer's problem in terms of a reparametrization of them.
This reparametrization builds implementability into the problem and separates the optimality conditions that are entangled in \eqref{eq: fixed point x}.
As a result, we can both verify the fixed-point characterization and solve it explicitly. In section~\ref{sec: main result}, we will show how our approach allows one to derive the reduced form $\bm x^*$ in the example above, establish its optimality and also identify the allocation that induces it.

\section{Preliminaries and simplifications}
\label{sec: preliminaries}

Before moving on, it is useful to introduce several simplifications and auxiliary observations that streamline exposition and reduce the dimensionality of the problem.

First, we can absorb all asymmetries across bidders into their revenue functions by working with quantiles of bidders' types.
Specifically, let $u_i = F_i(\theta_i)$ denote the quantile of $\theta_i$, which is uniformly distributed on $[0,1]$, and define
\[
H_i(x_i,u_i) := J_i(x_i,F_i^{-1}(u_i)).
\]
Then, our original problem is equivalent to one in which types, denoted by $\bm u$, are uniformly and identically distributed.
In this problem, the auctioneer's revenue from selling the good to type $u_i$ with probability $x_i$ is given by $H_i(x_i,u_i)$.
In this formulation, both allocations and interim winning probabilities can be viewed as functions of $\bm u$, so that \eqref{revenue} becomes%
\footnote{To obtain allocations in the original type space, we can simply substitute back $\bm u = (F_1(\theta_1),\ldots,F_n(\theta_n))$, e.g.,
$x_i(F_i(\theta_i))$ stands for bidder $i$'s interim winning probability.}
\begin{align}
\label{revenue u}
\sum_{i=1}^n \int_{0}^1 H_i(x_i(u),u)du.
\end{align}
Since $\frac{\partial }{\partial x_i}H_i(x_i,u_i) = \frac{\partial }{\partial x_i}J_i(x_i,F_i^{-1}(u_i))$, we continue to refer to $\frac{\partial H_i}{\partial x_i}$ as bidder $i$'s MR.

Second, observe that the expected revenue in \eqref{revenue u} depends only on interim winning probabilities.
It has been known since at least \citet{border1991implementation} and \citet{mierendorff2011asymmetric} that $\bm x$ is consistent with some allocation $\bm z$ satisfying analogues of \eqref{feasibility z} and \eqref{x from z} in the $\bm u$-space if and only if the following
\textbf{Border constraint} holds:
\begin{align}
\label{border}
B(\bm u) = \prod_{i=1}^n u_i + \sum_{i=1}^n \int_{u_i}^1 x_i(u)du \leq 1 \;\; \forall \bm u \in [0,1]^n.
\end{align}
in which case we say that $\bm z$ \textbf{induces} $\bm x$.
Intuitively, \eqref{border} requires that, for every cutoff vector $\bm u$, the total probability assigned to the bidders whose types lie in the upper tails $[u_i,1]$ does not exceed the probability that at least one such bidder is present.

Finally, it is without loss of generality to require each $x_i$ to be right-continuous with $x_i(1)=1$.
To see this, recall that $x_i$ is weakly increasing and thus has at most countably many jump discontinuities.
Redefining $x_i$ at those points (and possibly at $u=1$) affects neither the objective in \eqref{revenue u} nor the Border constraint in \eqref{border}.
By this observation, we can restrict attention to $\bm x$ that are cumulative distribution functions (CDFs) on the unit interval.
Let $\mathscr{X}$ denote the set of such CDFs.

We refer to a tuple $\bm x \in \mathscr{X}^n$, which captures interim winning probabilities, as a \textbf{reduced form}, and we use
\[
\underline u_i := \inf\{u_i \in [0,1]:x_i(u_i)>0\}
\]
to denote bidder $i$'s \textbf{exclusion type}, below which the bidder never wins.
A reduced form $\bm x$ is \textbf{feasible} if it satisfies the Border constraint.
For a feasible reduced form $\bm x$, the probability that the good is sold, $\mathbf{E}[\sum_{i=1}^n z_i]$, equals $B(\bm 0)$, and we say that the reduced form has \textbf{no exclusion} if this probability is equal to one.%
\footnote{If $\bm x$ has no exclusion, then necessarily $\prod_{i=1}^n\underline u_i = 0$; the converse, however, need not hold.}
Finally, a reduced form $\bm x$ is \textbf{optimal} if it solves the auctioneer's problem, which can be succinctly restated as
\begin{align}
\label{abstract problem}
\max_{\bm x \in \mathscr{X}^n} \; \sum_{i=1}^n \int_{0}^1 H_i(x_i(u),u)du
\;\; \text{s.t.}\;\; B(\bm u) \leq 1 \;\; \forall \bm u \in [0,1]^n.
\end{align}

While compact, the maximization problem in \eqref{abstract problem} remains challenging due to the nonlinearities in $H_i$, and since the feasible set is defined by a large and intricate family of inequalities.
To make progress, we introduce a novel characterization of feasibility that dramatically simplifies the characterization of the feasible set.
Furthermore, since that set is convex, additional tractability can be obtained by analyzing its extreme points.%
\footnote{Recall that a point in a convex set is extreme if it cannot be expressed as a convex combination of two distinct points in the set.
As explained in Section 2 of \citet{kleiner2021extreme}, by the Krein--Milman theorem, any convex and compact set in a locally convex space coincides with the closed convex hull of its extreme points.}
For example, if each revenue function $H_i$ is convex in its first argument---as often arises in settings with endogenous valuations---then an optimal reduced form can be found at an extreme point of the feasible set.
We call a feasible reduced form $\bm x$ \textbf{extremal} if it is an extreme point of the set of CDFs satisfying the Border constraint.


\section{Feasibility}
\label{sec: feasibility}

The Border constraint tells us exactly which reduced forms are feasible.
The drawback, however, is that it requires checking an entire $n$-dimensional family of inequalities---one for every cutoff vector $\bm u \in [0,1]^n$.
When $n=2$, this is still manageable, but, as soon as there are three or more bidders, the condition becomes challenging to verify directly.
Our goal in this section is to show that this system of constraints can, in fact, be reduced to a simple unidimensional test.

Our key insight is that feasibility can be checked along a single continuous curve inside the unit cube.
That is, for every reduced form $\bm x$, there exists a path
\[
t \in \mathbf{R}_+ \mapsto \bm\nu(t) \in [0,1]^n
\]
such that Border's inequality holds everywhere if and only if it holds along this path.
Instead of monitoring a continuum of constraints in $n$ dimensions, we can therefore follow just one curve.
We refer to such a curve associated with a given reduced form as its \textbf{principal curve}.

\subsection{$\delta$-transform}
\label{sec: delta transform}

To derive the principal curve, we first reparametrize CDFs as curves in the type-probability square.
To fix ideas, remark that the uniform CDF on $[0,1]$ can be thought of as a straight line in the type-probability square.
For each parameter $t \in \mathbf{R}_+$ (henceforth time), this line prescribes type $u=e^{-\sfrac{t}{2}}$ and probability $x=e^{-\sfrac{t}{2}}$.
This curve starts from the top-right corner and is such that the product of type and probability moves toward the origin at rate $1$.
Clearly, we can recover the uniform distribution from this parametric description as $x(u) = e^{-\sfrac{t}{2}}\big|_{t=-2\ln u}$.

We now seek to reparametrize an arbitrary CDF $x \in \mathscr{X}$ by a curve that starts at the top-right corner and whose coordinates multiply at each time $t$ to $e^{-t}$.
In general, such a curve can be expressed as%
\footnote{
In logarithmic coordinates $(-\ln u, -\ln x)$, the parametrization in \eqref{eq: delta parametrization} is the classical parametrization of monotone graphs by the sum of their coordinates, which goes back to \citet{minty1962monotone}; e.g., see Proposition 1.1 in \citet{alberti1999geometrical}.}
\begin{align}
\label{eq: delta parametrization}
t \in \mathbf{R}_+ \mapsto \left(e^{-\delta(t)},e^{\delta(t)-t}\right),
\end{align}
where $\delta$ belongs to the set $\mathscr{D}$ of absolutely continuous functions with $\delta(0) = 0$ and $\delta' \in [0,1]$.
Indeed, the coordinates multiply to $e^{-t}$, and the initial condition places the curve at the corner $(1,1)$.
The bounds on $\delta'$ make both coordinates weakly decreasing in $t$, so that the curve traces the graph of a weakly increasing function from top to bottom.

The representing function $\delta$ can be explicitly constructed from the CDF $x$ in two consecutive steps.
First, multiply $x$ by $u$ to obtain the product $ux$ and construct the inverse of that product, $(ux)^{-1}$.
Note that $ux$ itself is a CDF and is strictly increasing wherever $x$ is strictly positive.
Second, apply the log change of variables to obtain
\begin{align}
\label{eq: delta transform}
\delta(t) = -\ln (ux)^{-1}(e^{-t}).
\end{align}
Conversely, we can reverse these steps starting from any $\delta \in \mathscr{D}$ by finding the time at which the curve reaches type $u$
using the inverse of $\delta$ and evaluating the probability coordinate at that time, that is,
\begin{align}
\label{eq: delta transform inverse}
x(u) = e^{\delta(t)-t}\big|_{t=\delta^{-1}(-\ln u)}.
\end{align}
The following lemma records that nothing is lost in the process described above.

\begin{lemma}
\label{lm: delta transform}
The assignments in \eqref{eq: delta transform} and \eqref{eq: delta transform inverse} are mutually inverse and define a bijection between $\mathscr{X}$ and $\mathscr{D}$.%
\footnote{The inverses used in \eqref{eq: delta transform} and \eqref{eq: delta transform inverse} are generalized right- and left-continuous, respectively.
We provide their precise definitions in Appendix \ref{sec: appendix delta}.}
\end{lemma}
In view of this lemma, we can use CDFs and functions in $\mathscr{D}$ interchangeably.
For each $x \in \mathscr{X}$, we shall refer to the corresponding $\delta \in \mathscr{D}$ as the \textbf{$\delta$-transform of $x$}.
A more detailed discussion of the geometry of the $\delta$-transform can be found in Appendix \ref{sec: appendix geometry}.

\subsection{Feasibility through the principal curve}
\label{sec: feasibility curve}

The $\delta$-transform is the tool that allows us to collapse Border's $n$-dimensional condition into a unidimensional one.
Given a reduced form $\bm x$, which is not necessarily feasible, let $(\delta_1, \ldots, \delta_n)$ denote the $\delta$-transforms of its coordinates, and define the \textbf{principal curve} by applying \eqref{eq: delta parametrization} coordinatewise:
\begin{align}
\label{eq: principal curve}
t \in \mathbf{R}_+ \mapsto \bm\nu(t) := \left(e^{-\delta_1(t)}, \ldots, e^{-\delta_n(t)}\right).
\end{align}
The principal curve is continuous and weakly decreasing.
It starts at the top corner $(1, \ldots, 1)$ and descends toward the profile of bidders' exclusion types $(\underline u_1,\ldots,\underline u_n)$, below which the bidders stop receiving the good.
By construction, all bidders' types times their probabilities are equalized along the principal curve and equal $e^{-t}$.
The main result of this section is the following characterization of feasibility through the principal curve.

\begin{theorem}
\label{thm: feasibility}
A reduced form $\bm x$ is feasible if and only if
\[
B(\bm\nu(t)) \leq 1 \;\; \forall t \in \mathbf{R}_+.
\]
Moreover, along the principal curve, the value of $B$ can be expressed as the following function of bidders' $\delta$-transforms alone:
\begin{align}
\label{eq: border along curve}
B(\bm\nu(t)) = e^{-\sum_{i=1}^n \delta_i(t)} + \int_0^t e^{-\tau} \sum_{i=1}^n \delta_i'(\tau) d\tau.
\end{align}
\end{theorem}

Theorem \ref{thm: feasibility} shows that the Border constraint needs to be checked only along the principal curve.
Moreover, the test along this curve can be directly expressed in terms of the sum of bidders' $\delta$-transforms, so verifying feasibility requires no optimization---only checking the sign of a single scalar function of one variable.

\begin{example}
\label{ex: feasibility}
To illustrate, consider reduced forms of the power type $x_i(u) = u^{\sfrac{1}{\alpha_i} - 1}$, where $\alpha_i \in (0,1)$.
Direct computations reveal that $\delta_i(t) = \alpha_i t$, which gives
\begin{align}
\label{eq: example 2 border}
B(\bm\nu(t)) = e^{- t\sum_{i=1}^n \alpha_i} + (1 - e^{-t})\sum_{i=1}^n \alpha_i.
\end{align}
Applying the theorem, we conclude that $\bm x$ is feasible if and only if $\sum_{i=1}^n \alpha_i \leq 1$.
\end{example}

We now sketch the logic behind the proof of the theorem, leaving additional details to Appendix \ref{sec: appendix feasibility}.
For each level $s \in [0,1]$, consider the following auxiliary problem:
\begin{align}
\label{eq: auxiliary problem}
G(s) := \max_{\bm u \in [0,1]^n} s^n + \sum_{i=1}^n \int_{u_i}^1 x_i(u)du
\;\; \text{s.t.} \;\; \prod_{i=1}^n u_i \geq s^n.
\end{align}
Its value dominates the Border constraint: evaluating $G$ at $s^n = \prod_{i=1}^n u_i$ shows that $G(s) \geq B(\bm u)$ for every cutoff vector, whereas, at any maximizer, the objective is weakly below $B$.
Hence, $\bm x$ is feasible if and only if $G(s) \leq 1$ for all $s$.
The gain from replacing $\prod_{i=1}^n u_i$ with the constant $s^n$ in the objective is that the auxiliary problem becomes concave, since each tail integral is concave in $u_i$ and the constraint set is convex.
Using concavity and (generalized) first-order conditions, we show in Appendix \ref{sec: appendix feasibility} that \eqref{eq: auxiliary problem} always admits a maximizer at which the product $u_i x_i(u_i)$ is common across bidders.
But such points are precisely the points of the principal curve, which, by construction, equalizes bidders’ type-probability products. 
At jumps, the relevant type coordinate remains fixed while the probability coordinate varies, allowing the curve to cover every common value of these products.
Finally, evaluating $B$ at $\bm\nu(t)$ and changing the variable of integration along the curve, $\int_{\nu_i(t)}^1 x_i(u) du = \int_0^t e^{-\tau} \delta_i'(\tau) d\tau$, yields the expression in \eqref{eq: border along curve}.

To sum up, this unidimensional reduction provides a complete and tractable test for feasibility.
Just as important, it prepares the ground for our next steps, namely, the characterization of extremality and optimality that we will be covering in the subsequent sections.

\section{Extremality}
\label{sec: extremality}

In this section, we study the structure of extremal reduced forms and the structure of the corresponding allocations that induce them.

\subsection{Extremal reduced forms}
\label{sec: extremal reduced forms}

Whenever bidders' revenue functions $H_i$ are convex in $x_i$---as they are in many applications with endogenous valuations---the objective in \eqref{abstract problem} is convex.
Maximizing a continuous convex function over a convex compact set always yields at least one solution at an extreme point.
In such cases, the optimal auction can be implemented by an extremal reduced form.
Even when the objective is not convex, extremal reduced forms remain essential.
They serve as the ``building blocks'' of the feasible set, since any feasible reduced form can be expressed as a mixture of them.\footnote{For example, see \citet{kleiner2021extreme}, who explained that every feasible reduced form can be written as an integral over extremal ones due to the Choquet theorem. \label{fn: choquet}}

This perspective of extremality was first emphasized in the influential paper \citet{kleiner2021extreme}, which demonstrated the surprising tractability of extreme points of feasible symmetric reduced forms.
That paper also showed how they illuminate problems in auctions, delegation, and persuasion.
Here, we take the next natural step in this agenda by studying extreme points of (potentially asymmetric) feasible reduced forms.
We show that this broader set also admits a tractable description, one that can be directly exploited in economic applications.

At an intuitive level, extremal reduced forms correspond to ``corners'' of the feasible set, where Border's condition is maximally tight.
If the inequality were slack anywhere, then the reduced form could be perturbed in different directions and written as a convex combination of other feasible points.
Our main result in this section establishes that extremality arises precisely when the feasibility constraints along the principal curve bind at every point.

\begin{theorem}
\label{thm: extremality}
A feasible reduced form $\bm x$ is extremal if and only if
\[
B(\bm\nu(t)) = 1 \;\; \forall t \in \mathbf{R}_+,
\]
or equivalently:
\begin{align}
\label{eq: extremality delta}
\sum_{i=1}^n \delta_i(t) = \min\{t, T\} \;\; \forall t \in \mathbf{R}_+ \;\; \text{for some } T \in [0, \infty].
\end{align}
\end{theorem}

The theorem provides a unidimensional test analogous to the feasibility test developed in the previous section.
By construction, the principal curve starts at the northeast corner with $B(\bm\nu(0)) = 1$.
Then, feasibility requires that $B(\bm\nu(\cdot))$ never climb back above one, whereas extremality requires that it remain equal to one throughout.

The second part of the theorem builds on the representation of the Border constraint along the principal curve in \eqref{eq: border along curve} and delivers an explicit characterization of extremal reduced forms in terms of their $\delta$-transforms, that is, \eqref{eq: extremality delta}.
To see it, differentiate the expression for the Border constraint along the principal curve to obtain
\begin{align}
\label{eq: border derivative}
\frac{d}{dt}B(\bm\nu(t)) = \left(\sum_{i=1}^n \delta_i'(t)\right)\left(e^{-t} - e^{-\sum_{i=1}^n \delta_i(t)}\right).
\end{align}
Thus, $B(\bm\nu(\cdot))$ stays at one if and only if, at almost every $t$, the principal curve either does not move, $\sum_{i=1}^n \delta_i'(t) = 0$, or satisfies $\sum_{i=1}^n \delta_i(t) = t$.
Since a motionless curve can never rejoin the set $\sum_{i=1}^n \delta_i(t) = t$ once time has passed it by, the curve must ride this set up to some time $T$ and stand still forever after, which is exactly \eqref{eq: extremality delta}.

The stopping time $T$ has a direct economic interpretation in terms of exclusion.
At this time, the principal curve stops at the profile of bidders' exclusion types $\underline{\bm u} = \bm\nu(T)$.%
\footnote{Individual coordinates, however, may reach them earlier.
Specifically, $\delta_i$ becomes constant at the first time the principal curve hits bidder $i$'s exclusion type, which is $-\ln \big(\underline u_i x_i\big(\underline u_i\big)\big)$, and stays constant from then on, while at least one of the coordinates keeps changing before $T$.}
Since bidders' interim winning probabilities for types below the exclusion types are zero, the ex ante probability of selling the good satisfies
\[
B(\bm 0) = B\big(\underline{\bm u}\big) - \prod_{i=1}^n \underline u_i = 1 - e^{-T}.
\]
Consequently, an extremal reduced form has no exclusion precisely when its principal curve never terminates, i.e., when $T=\infty$.

\medskip
\noindent\textbf{Example \ref{ex: feasibility} (continued).}
Since $\sum_{i=1}^n \delta_i(t) = t \sum_{i=1}^n \alpha_i$, \eqref{eq: extremality delta} holds---that is, a reduced form of the power type is extremal---if and only if $\sum_{i=1}^n \alpha_i = 1$.
In that case, $T = \infty$, so the reduced form also has no exclusion.
\medskip

Although the extremality condition restricts only the sum of bidders' $\delta$-transforms, it pins down their interim winning probabilities along the principal curve.
By construction, for each bidder $i$, $\nu_i(t)x_i(\nu_i(t)) = e^{-t}$ except during intervals on which the type coordinate $\nu_i$ remains fixed,%
\footnote{During such an interval, the probability coordinate moves through a jump of $x_i$, and the identity holds with $x_i(\nu_i(t))$ replaced by $e^{\delta_i(t)-t}$.} whereas extremality states that $\prod_{j=1}^n \nu_j(t) = e^{-t}$.
Dividing the latter by the former yields the following expression for bidder $i$'s interim winning probability along the principal curve:
\begin{align}
\label{eq: winning against curve}
x_i(\nu_i(t)) = \prod_{j \neq i} \nu_j(t).
\end{align}
In other words, bidder $i$ with type $\nu_i(t)$ wins precisely against the type profiles in the box $\prod_{j \neq i} \big[0, \nu_j(t)\big)$ and loses otherwise.
That is, bidder $i$ beats every other bidder $j$ whose type $u_j$ lies below that bidder's position $\nu_j(t)$ on the principal curve.
As we explain in the next subsection, this ``rectangular'' structure of winning regions is the defining geometric feature of the allocations that induce extremal reduced forms.

Figure \ref{fig: extremality} illustrates this point with a simple numerical example.
In this figure, the black line corresponds to the principal curve with $u_1$ and $u_2$ on the horizontal and vertical axes, respectively, constructed from an extremal reduced form.
For example, bidder $1$ with type $u_1 = \frac{1}{2}$ wins with probability $\frac{1}{2}$, and with type $u_1 = \frac{3}{4}$ wins with probability $\frac{3}{4}$.
Using \eqref{eq: winning against curve}, we can recover the reduced form consistent with the principal curve shown in the figure:
\begin{align}
x_1(u_1) &= u_1 \cdot \bm 1_{[\sfrac{1}{4},\sfrac{1}{2}) \cup [\sfrac{3}{4},1]}(u_1) + \sfrac{1}{2}\cdot \bm 1_{[\sfrac{1}{2}, \sfrac{3}{4})}(u_1), \label{eq: example x1}\\
x_2(u_2) &= u_2 \cdot \bm 1_{[\sfrac{1}{4},\sfrac{1}{2})\cup [\sfrac{3}{4},1]}(u_2) + \sfrac{3}{4} \cdot \bm 1_{[\sfrac{1}{2},\sfrac{3}{4})}(u_2). \label{eq: example x2}
\end{align}
This representation makes clear that the allocation inducing this extremal reduced form is unique: bidder $1$ wins in the gray region, bidder $2$ wins in the red region, and otherwise the good remains unallocated.

\begin{figure}[!htb]
\centering
\begin{tikzpicture}[scale=1.15]

\draw (0,0) -- (0,4) -- (4,4) -- (4,0) -- (0,0);

\fill[red!15] (0,1) -- (4,1) -- (4,4) -- (0,4) -- cycle;
\fill[red!15] (1,1) -- (2,2) -- (2,4) -- (1,4) -- cycle;
\fill[red!15] (2,2) -- (3,2) -- (3,4) -- (2,4) -- cycle;
\fill[red!15] (3,4) -- (4,4) -- (3,3) -- cycle;

\fill[gray!20] (1,0) -- (2,0) -- (2,2) -- (1,1) -- cycle;
\fill[gray!20] (2,0) -- (3,0) -- (3,2) -- (2,2) -- cycle;
\fill[gray!20] (3,0) -- (4,0) -- (4,4) -- (3,3) -- cycle;

\draw[ultra thick, black] (1,1) -- (2,2);
\draw[ultra thick, black] (2,2) -- (3,2);
\draw[ultra thick, black] (3,2) -- (3,3);
\draw[ultra thick, black] (3,3) -- (4,4);

\draw[blue, thick, ->] (2,0) to[bend right=12] (2,1.88);
\draw[blue, thick, ->] (3,0) to[bend right=12] (3.1,2.9);
\draw[red, thick, ->] (2,2) to[bend right=12] (0,2);
\draw[red, thick, ->] (3,3) to[bend right=12] (0,3);
\filldraw[black] (2,2) circle (2pt) node[above, xshift=0.5pt, yshift=0.pt] {$t'$};
\filldraw[black] (3,3) circle (2pt) node[above, xshift=0.5pt, yshift=0.pt] {$t''$};

\draw[dotted, black] (1,1) -- (1,0) node[below]{$\sfrac{1}{4}$};
\draw[dotted, black] (2,3) -- (2,0) node[below]{$\sfrac{1}{2}$};
\draw[dotted, black] (3,3) -- (3,0) node[below]{$\sfrac{3}{4}$};
\draw[dotted, black] (4,4) -- (4,0) node[below]{$u_1$};

\draw[dotted, black] (1,1) -- (0,1) node[left]{$\sfrac{1}{4}$};
\draw[dotted, black] (2,2) -- (0,2) node[left]{$\sfrac{1}{2}$};
\draw[dotted, black] (3,3) -- (0,3) node[left]{$\sfrac{3}{4}$};
\draw[dotted, black] (0,4) -- (0,4) node[left]{$u_2$};

\end{tikzpicture}
\caption{
Principal curve (black) induced by an extremal reduced form.
The points $t',t''$ correspond to bidder $1$ types $u_1=\frac{1}{2},\frac{3}{4}$ and the associated bidder $2$'s cutoffs $\nu_2(t')=\frac{1}{2}$, $\nu_2(t'')=\frac{3}{4}$.}
\label{fig: extremality}
\end{figure}

Though the logic outlined in \eqref{eq: winning against curve} applies to any number of bidders, the two-bidder case admits a further geometric reading that can also be seen in Figure \ref{fig: extremality}.
With $n = 2$, extremality gives $\nu_2(t) = e^{-t}/\nu_1(t)$ for $t \leq T$, so that the principal curve is nothing but the completed graph of $x_1$ drawn in the $(u_1,u_2)$-square.
Whenever bidder $1$'s coordinate moves, $\nu_2(t) = x_1(\nu_1(t))$. 
When that coordinate remains fixed, the curve descends vertically, filling in the corresponding jump of $x_1$, as illustrated by the black line in the figure.
It follows that the reduced form satisfies two simple inverse relations: $x_1 = x_2^{-1}$ on $\big[\underline u_1,1\big]$ and $x_2 = x_1^{-1}$ on $\big[\underline u_2,1\big]$.
For the two-bidder case, \citet{yang2025multidimensional} established that this inverse-relation property is equivalent to extremality.
With more than two bidders, no such unidimensional inverse relationship exists, and \eqref{eq: winning against curve} can be viewed as its natural multidimensional generalization.

The proof of the ``only if'' direction of Theorem \ref{thm: extremality}---asserting that slack anywhere on the principal curve destroys extremality---is rather involved.
We show in Supplementary Appendix \ref{sec: appendix extremality} that a feasible reduced form $\bm x$ with a stretch of slack can be perturbed by a function $\bm\Delta \neq \bm 0$ so that both $\bm x + \varepsilon\bm\Delta$ and $\bm x - \varepsilon\bm\Delta$ are feasible reduced forms provided that $\varepsilon > 0$ is small enough.
The construction is delicate because the principal curve itself moves with the perturbation in a nontrivial fashion, since the $\delta$-transform depends nonlinearly on the reduced form.
The ``if'' direction follows from the analysis of the next subsection by a different route.
Specifically, as we explain there, reduced forms with $B(\bm\nu(t)) \equiv 1$ maximize suitable linear functionals uniquely, making them not only extreme but exposed points of the feasible set.

\subsection{Allocations inducing extremal reduced forms}
\label{sec: score allocations}

The analysis in the previous section showed that extremal reduced forms are entirely determined by their principal curves and that Border’s condition must bind along those curves. 
This indicates that such reduced forms can arise only from highly structured allocations. 
We now turn to these allocations and show that each extremal reduced form is induced by a unique allocation, which can be recovered directly from the reduced form itself. To formulate this characterization, we introduce the notion of a score allocation.

\begin{definition}
\label{def: score allocation}
An allocation $\bm z$ is a \textbf{score allocation} if there exists a tuple of weakly increasing, right-continuous functions $\bm q = (q_1, \ldots, q_n) : [0,1] \to \mathbf{R}^n$ such that, for a.e. type profile $\bm u$,
\begin{align}
\label{eq: winner q}
\#\left\{i:q_i(u_i) = \max_{j}q_j(u_j)>0\right\} \leq 1,
\end{align}
and
\begin{align}
\label{eq: score allocation}
z_i(\bm u) =
\begin{cases}
\bm 1_{(0,\infty)}(q_i(u_i)) & \text{if } q_i(u_i) \geq \max_{j \neq i} q_j(u_j),\\
0 & \text{otherwise}.
\end{cases}
\end{align}
\end{definition}

In this definition, $q_i(u_i)$ is the \textbf{score} assigned to bidder $i$ with type $u_i$, and the good is awarded to the bidder with the highest positive score.
Condition \eqref{eq: winner q} ensures that the event of two or more bidders tied at the highest positive score is null, so that the score allocation is feasible.%
\footnote{To be clear, ties can occur with positive probability among those bidders that are not winning the good.}
Score allocations are familiar in classical settings; e.g., Myerson's optimal auction is a score allocation whenever bidders' MVVs are strictly increasing in their types, with the MVVs themselves serving as scores.

Every score allocation induces an extremal reduced form.
To see it, take a tuple of scores $\bm q$ and remark that the associated score allocation $\bm z^{\bm q}$ is the a.e.-unique maximizer of
\[
\sum_{i=1}^n \int_{[0,1]^n} \Big[q_i(u_i)-\bm 1_{(-\infty,0]}(q_i(u_i))\Big] z_i(\bm u) d\bm u
\]
in the space of all allocations.
Here, uniqueness is a direct consequence of weak monotonicity of scores, \eqref{eq: winner q}, and the fact that the bracketed weights are negative wherever scores are nonpositive.
As a result, the reduced form $\bm x^{\bm q}$ induced by that score allocation is the unique maximizer of the corresponding linear functional in the space of all feasible reduced forms.
This shows that $\bm x^{\bm q}$ is not merely extremal but in fact exposed.%
\footnote{Recall that a point in a convex subset of a topological vector space is exposed if some continuous linear functional attains its strict maximum at that point, and every exposed point is extremal, e.g., see \citet{kleiner2021extreme}. Here, the objective is a continuous linear functional when the set of CDFs on the unit interval is viewed as a subset of Lebesgue integrable functions endowed with the norm topology.}
Strikingly, the converse is also true, as shown in the theorem below.

\begin{theorem}
\label{thm: score allocations}
A feasible reduced form $\bm x$ is extremal if and only if it is induced by a score allocation.
In that case, the allocation inducing $\bm x$ is unique, and it is the score allocation that can be obtained with \textbf{canonical scores} $q_i = u_i x_i$.
\end{theorem}
Theorem \ref{thm: score allocations} is a complete description of how extremal reduced forms are induced.
First, any tuple of weakly increasing scores satisfying \eqref{eq: winner q} induces an extreme point, and every extreme point arises this way.
Second, the ex post allocation behind an extremal reduced form is unique, even though the scores representing it are not; i.e., applying a common strictly increasing, sign-preserving transformation to all scores leaves the allocation unchanged.
Third, testing a candidate $\bm x$ requires no search over arbitrary scores, and it suffices to use its canonical scores $u_i x_i$, which are constructed directly from the reduced form itself.

Score allocations have appeared in the social choice literature, but their connection to extremality has not, to the best of our knowledge, been recognized.
\citet{mishra2014non} referred to them as simple utility maximizers and showed that they can be axiomatized by three basic properties:
\begin{enumerate}
\item \textbf{Deterministic.} Either some bidder gets the good with certainty or the good is not traded.
\item \textbf{Monotone.} Each bidder's allocation is weakly increasing in their own type.
\item \textbf{Nonbossy.} A bidder cannot change the allocation of others without also changing their own. Formally,
\[
z_i(u_i, \bm u_{-i}) = z_i(u_i', \bm u_{-i}) \implies \bm z(u_i, \bm u_{-i}) = \bm z(u_i', \bm u_{-i}) \;\; \forall i,\;\; \bm u_{-i},
\]
in the sense of \citet{satterthwaite1981strategy}.
If bidder $i$'s own outcome does not change when their type is varied, then the outcome for all other bidders must also remain unchanged.
This rules out ``bossy'' manipulations where the bidder could leave their own allocation fixed while reshuffling who among their rivals receives the good.
\end{enumerate}

Monotonicity yields a direct proof of the classical BIC--DIC equivalence, first established by \citet{manelli2010bayesian}.
By Theorem \ref{thm: score allocations}, every extremal reduced form is induced by a allocation that is monotone in each bidder's own type and hence can be implemented in dominant strategies.%
\footnote{See also \citet{gershkov2013equivalence} and \citet{goeree2023geometric}.}
Since any feasible reduced form can be expressed as an integral over extremal ones (see Footnote \ref{fn: choquet}), the equivalence follows by mixing their score allocations.
Our characterization further shows that the allocations in this mixture are not merely monotone but also deterministic and nonbossy.

Taken together, these three properties have a sharp geometric implication.
The region of the type space where the auctioneer keeps the good, $\{\bm u \in [0,1]^n : \bm z(\bm u) = \bm 0\}$, and the regions where each bidder $i$'s type $u_i$ wins, $\{\bm u_{-i} \in [0,1]^{n-1} : z_i(u_i, \bm u_{-i}) = 1\}$, are hyperrectangles aligned with the axes.
This rectangular structure is exactly the winning-region geometry identified in \eqref{eq: winning against curve}, and it is what makes score allocations induce extremal reduced forms.
Not every deterministic and monotone allocation admits this structure; e.g., suppose the auctioneer awards the good to bidder $1$ when $u_1 > \frac{u_2 + u_3}{2}$, and otherwise allocates it efficiently among bidders $2$ and $3$.
Clearly, the set of $(u_2, u_3)$ such that bidder $1$'s type $u_1$ wins is not a rectangle.

Turning to the proof, the ``if'' direction and uniqueness both follow from the argument preceding the theorem.
That argument shows that a score allocation is the a.e.-unique maximizer of a linear functional whose value depends on the allocation only through the induced reduced form.
Hence, the induced reduced form is exposed---in particular extremal.%
\footnote{With finitely many types, the set of feasible reduced forms is a polytope, its extreme points coincide with its exposed points, and they are known to be implemented by score-allocation-type mechanisms; e.g., see \citet{vohra2011mechanism}. With a continuum of types, exposed points are still implemented in this way; this is, in essence, the chain characterization of \citet{delbaen1974convex} for cores of convex games. The general theory, however, guarantees only that exposed points are dense in the extreme points, whose structure has otherwise remained unknown. Theorem \ref{thm: score allocations} shows that, in our setting, the two notions in fact coincide.}
The ``only if'' direction and the third point are direct consequences of \eqref{eq: winning against curve} and the construction of the $\delta$-transform in \eqref{eq: delta transform}.

\section{Optimality: the case of no exclusion}
\label{sec: optimality}

Having characterized feasibility and extremality, we now turn to the auctioneer's problem.
We begin with the case in which the optimal reduced form, $\bm x^*$, has no exclusion, since there our approach to optimality through the principal curve is most transparent.
Later, we provide conditions under which this case indeed describes the optimum and discuss how the analysis changes once exclusion arises.

Recall the conjecture of Section \ref{sec: example}, motivated by  \citet{bulow1989simple} interpretation of Myerson's auction, according to which the good should go to the bidder with the highest positive MR evaluated at the optimal interim winning probabilities themselves, that is,
\begin{align}
\label{eq: fixed point H}
x_i^*(u_i) =
\mathbf{Pr}\left(\frac{\partial}{\partial x_i}H_i(x_i^*,u_i) \geq \frac{\partial}{\partial x_j}H_j(x_j^*,u_j) \;\;\forall j \neq i
| u_i\right)
\end{align}
whenever bidder $i$'s MR is positive, and $x_i^*(u_i) = 0$ otherwise.
The difficulty is that this condition is self-referential, since the ranking of bidders depends on the very reduced form we are trying to determine.
The only case where this circularity disappears is Myerson's linear model, in which bidders' MRs
\begin{align}
\label{eq: Myerson H}
\frac{\partial }{\partial x_i}H_i(x_i,u_i) = \zeta_i(u_i),
\end{align}
where $\zeta_i$ is bidder $i$'s MVV in the quantile space, do not depend on interim winning probabilities.%
\footnote{As explained in Section \ref{sec: linear}, bidder $i$'s MVV as a function of their type $\theta_i$ equals $\theta_i - \frac{1-F_i}{f_i}$.
In the quantile space, this object can be expressed as $\zeta_i = v_i-(1-u_i)v_i'$, where $v_i$ is the inverse of $F_i$. \label{ftn: virtual value}}

The results of the preceding sections allow us to remove this circularity and express the conjecture directly in terms of the principal curve.
To streamline the exposition and highlight the main idea, we specialize to the case in which the optimal reduced form is continuous on $[0,1)$ and bidders' MRs at the optimum, $\frac{\partial}{\partial x_i}H_i(x_i^*(u_i),u_i)$, are weakly increasing in their own types.
We also require these MRs to be strictly increasing whenever positive.
The general characterization in Theorem \ref{thm: optimal no exclusion} shows that, in regular environments, the optimal reduced form indeed has the properties just listed, and that it satisfies the system we arrive at below.

Observe that \eqref{eq: fixed point H} describes the score allocation in which bidder $i$'s score is own MR.
By Theorem \ref{thm: score allocations}, such an allocation induces an extremal reduced form and thus can also be induced by the canonical scores, i.e., for almost every $\bm u$,%
\footnote{\label{fn: below exclusion}Equation \eqref{eq: ranking equivalence} also holds for $u_i$ below the exclusion type $\underline u_i$.
To see it, remark that, in the auction with no exclusion, at least one bidder $j$ is never excluded, i.e., $\underline u_j=0$.
By \eqref{eq: fixed point H}, that bidder's MR is strictly positive at all positive types, so each bidder $i$ with $u_i<\underline u_i$ loses to bidder $j$ on both sides of \eqref{eq: ranking equivalence}---their canonical score is zero and their MR is not maximal.}
\begin{align}
\label{eq: ranking equivalence}
u_i x_i^*(u_i) \geq u_j x_j^*(u_j) \;\;\forall j \neq i
\iff
\frac{\partial}{\partial x_i}H_i(x_i^*(u_i),u_i) \geq \frac{\partial}{\partial x_j}H_j(x_j^*(u_j),u_j) \;\;\forall j \neq i.
\end{align}
In other words, bidders' MRs induce the same ranking of winners as their canonical scores.
Since both sets of scores are continuous in types below the top, whenever canonical scores are tied, MRs are tied as well.
But the canonical scores are tied precisely along the principal curve, and, absent exclusion, the curve never stops descending from the northeast corner.
At each time $t$, every bidder whose coordinate has left the top sits at type $\nu_i^*(t)=e^{-\delta_i^*(t)}<1$ with winning probability $x_i^*(\nu_i^*(t))=e^{\delta_i^*(t)-t}$ and canonical score
\[
\nu_i^*(t)x_i^*(\nu_i^*(t)) = e^{-t}.
\]
By \eqref{eq: ranking equivalence}, these bidders' MRs must share a common value denoted by $p^*(t)$.

A bidder whose coordinate is still at the top, $\nu_i^*(t)=e^{-\delta_i^*(t)}=1$, is in a different position.
The jump of $x_i^*$ at $u_i = 1$ means that all of this bidder's types have canonical scores lower than $e^{-t}$, so, by \eqref{eq: ranking equivalence}, their MRs fall below $p^*(t)$.
This, however, is a statement about MRs at winning probabilities up to $x_i^*(1-)$.
Joining the principal curve at time $t$ would give the bidder's top type the winning probability $e^{-t}$, and the conjecture is silent about the MR there.
To complete the description, we require that no such bidder wants to join, that is, that the MR of its top type at the entry probability $e^{-t}$ be at most the common MR,
\[
\frac{\partial}{\partial x_i}H_i(e^{-t},1) \leq p^*(t).
\]
As we explain in the sequel, it is exactly the condition that the auctioneer's problem delivers in regular environments.

Recalling from Theorem \ref{thm: extremality} that extremality with no exclusion means the $\delta$-transforms sum to $t$, we are led to the following system:
\begin{align}
\label{eq: static system}
\frac{\partial}{\partial x_i}H_i\left(e^{\delta_i^*(t)-t}, e^{-\delta_i^*(t)}\right) \leq p^*(t) \;\;\text{with equality whenever } \delta_i^*(t) > 0 \;\; \forall i,
\qquad
\sum_{i=1}^n \delta_i^*(t) = t.
\end{align}
Thus far, \eqref{eq: static system} is a set of conditions we expect the optimum to satisfy, but it can also be read as a way to find it.
To this end, treat \eqref{eq: static system} as a system of $n+1$ equations in $n+1$ unknowns parametrized by time $t$.
The system involves only positions---no rates, nothing linking one time to another---so it can be solved at each time separately, directly from the primitives.
The main result of the section that follows is that, in regular environments where the optimum has no exclusion, the path satisfying \eqref{eq: static system} characterizes the optimal auction.
Section \ref{sec: optimal exclusion} studies environments in which it is optimal to exclude some bidders' types.

\subsection{Main result}
\label{sec: main result}

We begin with a notion of regularity that extends Myerson's regularity---MVVs are strictly increasing---to nonlinear environments so that no ironing is needed.
We say that an environment is \textbf{regular} if the following two conditions hold.
\begin{itemize}
\item[(A)] For each bidder $i$ and all $x_i,u_i \in (0,1)$,
\[
\frac{\partial}{\partial \ln x_i}\left(\frac{\partial }{\partial x_i}H_i(x_i,u_i)\right) < \frac{\partial}{\partial \ln u_i}\left(\frac{\partial}{\partial x_i}H_i(x_i,u_i)\right).
\]
\item[(B)] The solution $(\bm\delta^{\rm s}(t), p^{\rm s}(t))$ to \eqref{eq: static system}---unique under Condition (A) at each time $t$, with the superscript ``s'' for \eqref{eq: static system}---is such that $p^{\rm s}$ is strictly decreasing and each $\delta_i^{\rm s}$ is strictly increasing whenever positive.%
\footnote{At $t=0$, we use the convention $p^{\rm s}(0)=\max_i \frac{\partial}{\partial x_i}H_i(1,1)$ so that $p^{\rm s}$ is continuous.}
\end{itemize}

To understand the content of regularity, consider again Myerson's linear model in \eqref{eq: Myerson H}.
Since $\frac{\partial H_i}{\partial x_i}$ does not depend on the winning probability, the left-hand side of Condition (A) vanishes, and the condition reads $u_i \zeta_i' > 0$, that is, bidders' MVVs are strictly increasing---precisely Myerson's regularity.
Condition (B) then holds automatically, since bidders' MRs in \eqref{eq: static system} do not vary with time.%
\footnote{Using \eqref{eq: static system}, we may express $\delta_i^{\rm s}$ as a function of $p^{\rm s}$ alone and then find the value of $p^{\rm s}$ by solving $\prod_{i=1}^n e^{-\delta_i^{\rm s}} = e^{-t}$.
Since MVVs do not depend on time, raising $t$ only lowers $p^{\rm s}$.}
In nonlinear environments, bidders' MRs in that system vary with time through the winning probabilities.
Condition (A) then means that bidder $i$'s MR is strictly decreasing in $\delta_i$ at each time $t$, and Condition (B) becomes an additional restriction that rules out ironing.
Supplementary Appendix \ref{sec: appendix regularity} further discusses both conditions and unpacks Condition (B) in terms of primitives.

By construction, the path $\bm\delta^{\rm s}$ equalizes bidders' MRs at every time but never contemplates exclusion, since $\sum_{i=1}^n \delta_i^{\rm s}(t) = t$.
We refer to $\bm\delta^{\rm s}$ as the \textbf{equi-MR path} and to $p^{\rm s}$ as the \textbf{common MR}, and we say that bidder $i$ is \textbf{active} at time $t$ if $\delta_i^{\rm s}(t) > 0$.
Each coordinate of the equi-MR path is a weakly increasing and absolutely continuous function, and thus it can be inverted using \eqref{eq: delta transform inverse} to obtain the following CDF:
\begin{align}
\label{eq: optimal x no exclusion}
x_i^{\rm s}(u_i) = e^{\delta_i^{\rm s}(t_i)-t_i}\big|_{t_i=(\delta_i^{\rm s})^{-1}(-\ln u_i)}.
\end{align}
By Theorem \ref{thm: feasibility}, the tuple of CDFs $\bm x^{\rm s}$ is a feasible reduced form, which we term the \textbf{equi-MR reduced form}.
Our main result explains that the equi-MR reduced form is optimal provided that the common MR stays positive.

\begin{theorem}
\label{thm: optimal no exclusion}
Suppose the environment is regular and $p^{\rm s} \geq 0$.
Then, the equi-MR reduced form is the unique optimal reduced form.
Furthermore, it is extremal, has no exclusion, and satisfies \eqref{eq: fixed point H}.
\end{theorem}

Theorem \ref{thm: optimal no exclusion} explicitly constructs the optimal auction in regular nonlinear environments where the optimum has no exclusion.
In such an auction, each bidder's type is mapped to the time at which the equi-MR path reaches it, $t_i = (\delta_i^{\rm s})^{-1}(-\ln u_i)$, and the bidder's canonical score takes the particularly simple form $u_ix_i^*(u_i) = e^{-t_i}$.
The good then goes to the bidder whose type is reached by the path the earliest, so that the good is sold with certainty.
At every type $u_i$ in $\big[\underline u_i^*, 1\big)$, bidder $i$'s MR in the optimal auction equals $p^{\rm s}(t_i)$---the common MR at the time the path reaches their type.
Consequently, the ranking by canonical scores in Theorem \ref{thm: optimal no exclusion} is exactly the MR ranking conjectured in \eqref{eq: fixed point H}.

The role of the condition $p^{\rm s} \geq 0$ is to guarantee that no exclusion is optimal.
Intuitively, it says that, at every time, there is at least one bidder with a positive MR along the optimal principal curve.
If instead $p^{\rm s}$ turned negative at some time, the auctioneer would prefer to stop selling to the types reached after that time, and the optimum would involve exclusion---the case we take up in the sequel.
It is worth mentioning that, since $p^{\rm s}$ is defined explicitly by \eqref{eq: static system}, this condition is directly on primitives.
In Lemma \ref{lm: positive common MR} in Supplementary Appendix \ref{sec: appendix positive MR}, we unpack it directly in terms of $\frac{\partial H_i}{\partial x_i}$.

In the linear model, bidders' MRs are pinned down by their types alone, and the optimal auction reduces to Myerson's auction, which awards the good to the bidder with the highest positive MVV.
In nonlinear environments, MRs depend also on bidders' winning probabilities, so bidder $i$'s MR in the optimal auction---unlike MVV---incorporates information about other bidders.
For instance, whether bidder $i$'s type $u_i$ beats bidder $j$'s type $u_j$ depends on the entire environment, including the type distributions and revenue functions of the remaining bidders.
In Myerson's auction, by contrast, this comparison involves the two bidders alone.
Strikingly, despite the presence of nonlinearities that render standard approaches inapplicable, the characterization of the optimum coincides formally in the general and linear cases.
In both, the optimal reduced form is obtained by solving \eqref{eq: static system} at each time and inverting the resulting $\delta$-transforms via \eqref{eq: optimal x no exclusion}.

We defer the discussion of the proof to Section \ref{sec: methodology} below and first revisit our motivating example through the lens of Theorem \ref{thm: optimal no exclusion}.

\begin{example}[Illustrative example, continued]
\label{ex: ev revisited}
To see Theorem \ref{thm: optimal no exclusion} at work, consider again the parametric specification of Section \ref{sec: example}, in which bidder $i$'s revenue function in the quantile space is $H_i = u_i^{1/\beta_i} x_i^2$.
Condition (A) holds since
\[
\frac{\partial}{\partial \ln x_i}\left(\frac{\partial H_i}{\partial x_i}\right) = \frac{\partial H_i}{\partial x_i} < \frac{1}{\beta_i}\frac{\partial H_i}{\partial x_i} = \frac{\partial}{\partial \ln u_i}\left(\frac{\partial H_i}{\partial x_i}\right).
\]
Along the parametrization underlying the $\delta$-transform, bidder $i$'s MR equals
\[
\frac{\partial}{\partial x_i}H_i\left(e^{\delta_i-t}, e^{-\delta_i}\right) = 2e^{-(1/\beta_i-1)\delta_i}e^{-t}.
\]
Equalizing these MRs across bidders and requiring that bidders' $\delta$-transforms add up to $t$, we obtain
\[
\delta_i^{\rm s}(t) = \frac{\kappa}{1/\beta_i-1} t,
\qquad
p^{\rm s}(t) = 2e^{-(1+\kappa)t},
\]
where $\kappa$ is the normalization constant so that $\sum_{i=1}^n\delta_i^{\rm s}(t)=t$.
This shows that Condition (B) holds as well, and so the environment is regular with $p^{\rm s} \geq 0$.

Inverting $\bm\delta^{\rm s}$ via \eqref{eq: optimal x no exclusion} yields that the optimal reduced form is given by
\[
x_i^*(u_i) = u_i^{\sfrac{(1/\beta_i-1)}{\kappa}-1}.
\]
The canonical scores are $u_ix_i^*(u_i) = u_i^{\sfrac{(1/\beta_i-1)}{\kappa}}$, so bidder $i$ beats bidder $j$ if and only if $u_i^{1/\beta_i-1} \geq u_j^{1/\beta_j-1}$.
Returning to the original type space, this shows that the optimal auction awards the good to the bidder with the highest value of $(F_i(\theta_i))^{1/\beta_i-1}$, as claimed in Section \ref{sec: example}.
\end{example}

The example confirms the conjectured structure of the optimum, which, as explained in Section \ref{sec: example}, favors stronger bidders relative to Myerson's auction.
Intuitively, in the endogenous-valuation model, the auctioneer internalizes the effect of changing interim winning probabilities on bidders' investment incentives.
By equalizing MRs, the auctioneer shifts winning probability toward bidders whose investments generate the most revenue, reducing the duplication of investment costs.
In the next section, we show, using Theorem \ref{thm: optimal no exclusion}, that this phenomenon holds beyond that single example.

\subsection{Application to discrimination in the endogenous-valuation model}
\label{sec: EV applications}

Consider a setting with endogenous valuations, in which the reduced-form utility of every bidder $i$ is equal to $\omega_i(x_i,\theta_i)=\theta_i h(x_i)$, where $h',h''>0$.
As explained in Section \ref{sec: EVs}, bidder $i$'s revenue function $J_i$ takes the form
\[
J_i(x_i,\theta_i) = \left(\theta_i-\frac{1-F_i(\theta_i)}{f_i(\theta_i)}\right)h(x_i),
\]
where the term in parentheses is that bidder's MVV.
Bidders' MVVs are assumed to be strictly increasing to avoid dealing with ironing that does not add much economic insight.
This setting generalizes that of Example \ref{ex: ev revisited} and is chosen because Myerson's virtual values arise naturally as scaling factors of bidders' revenue functions, even though the utilities are nonlinear.\footnote{The primitive utility over action-type $(a_i,\theta_i)$ pairs that gives rise to such an $\omega_i$ can be recovered via $\min_x (\theta_i h(x)+a_i)/x$ when the investment costs are normalized to $C_i(a_i)=a_i$; see \citet{scoringcompanion}.}
As a result, we can naturally compare bidders' strengths within this framework in terms of their MVVs and contrast the optimal auction to Myerson's auction, which allocates the good to the bidder with the highest positive MVV.

One natural order that has been employed to rank bidders' strength (e.g., see \citet{carroll2019robustly}) is the \textbf{hazard rate order}.
A cumulative distribution function $F$ dominates another cumulative distribution function $G$ in this order, denoted $F \geq_{hr} G$, when their survival functions satisfy
\[
(1-F(\theta'))(1-G(\theta)) \geq (1-F(\theta))(1-G(\theta')) \;\;\forall \theta \leq \theta'.
\]
See Chapter 1 in \citet{shaked2007stochastic}.
As explained in this book, this ranking implies $\frac{1-F}{f} \geq \frac{1-G}{g}$ over the common support of these two distributions.
Intuitively, high types are more likely under $F$ than under $G$ conditional on any interval of the form $[\theta,\infty)$.
Hence, a bidder whose types are drawn from $F$ should receive a higher information rent than the one whose types are drawn from $G$.
As a result, if there are two distinct bidders $i,j$ with the same value $\theta_i=\theta_j$ and $F_i \geq_{hr} F_j$, then Myerson's auction allocates the good more often to the weaker bidder $j$.

As shown in the previous section, to determine the optimum, one needs to work in the quantile space of bidders' types.
The hazard rate order, however, is silent about bidders' strength in the quantile space, where their MVVs can be expressed using $v_i = F_i^{-1}$ as in Footnote \ref{ftn: virtual value}, that is,
\begin{align}
\label{eq: zeta}
\zeta_i(u) = v_i(u)\left(1-(1-u)\frac{v_i'(u)}{v_i(u)}\right).
\end{align}
This happens because the second term in \eqref{eq: zeta} can be lower for bidder $i$ than for bidder $j$ when $F_i$ is much more unequal than $F_j$.
This is so even if the two bidders have the same quantile $u_i=u_j$ and $F_i \geq_{hr} F_j$, so that $v_i \geq v_j$.
In our comparison of the optimal and Myerson's auctions, we use one more standard stochastic order allowing us to rank bidders' strengths in the quantile space as well.

Following Chapter 4 in \citet{shaked2007stochastic}, we say that a cumulative distribution function $F$ dominates another cumulative distribution function $G$ in the \textbf{star order}, denoted $F \geq_{*} G$, when their quantile functions satisfy
\[
F^{-1}(u')G^{-1}(u) \geq F^{-1}(u)G^{-1}(u')\;\;\forall u \leq u'.
\]
The star order means that $G$ is more equal relative to its mean than $F$.
There are many parametric families of distributions for which $F \geq_{hr} G$ automatically implies $G \geq_{*} F$.
For example, this is true for (1) Normal, Logistic, Gumbel distributions ordered by a location parameter; (2) Exponential distributions ordered by a scale parameter; (3) Power distributions ordered by a shape parameter; as well as (4) distributions of the form in Example \ref{ex: ev revisited}.
A notable family for which this is not true is Pareto distributions, where $F \geq_{hr} G$ implies $F \geq_{*} G$ due to a thicker tail.%
\footnote{The star order also has an informativeness interpretation. 
For scale experiments, the star order corresponds (with reversed ranking) to Lehmann's informativeness order; see \citet{lehmann1988comparing}.}

Going back to our auction framework, consider two bidders $i$ and $j$ with $F_i \geq_{hr} F_j$ and, in addition, $F_j \geq_{*} F_i$.
Then, not only are bidder $i$'s values, $v_i$, stochastically higher, but also per-unit information rents, $(1-u)\frac{v_i'}{v_i}$, are lower.
This means that the stronger bidder $i$ is allocated the good more often in the quantile space than bidder $j$.
Equipped with these two approaches to ranking bidders' strengths, we can now state our result formally.

\begin{proposition}
\label{prop: optimum vs myerson}
Consider the model as above, and assume that the environment is regular with $p^{\rm s}\geq 0$ and that
\[
F_1 \geq_{hr} F_i \geq_{hr} F_n, \;\;F_n \geq_{*} F_i \geq_{*} F_1 \;\;\forall i.
\]
Then, the optimal reduced form $\bm x^*$ and the reduced form induced by Myerson's auction $\bm x^{\rm M}$ satisfy:
\begin{enumerate}
\item
The strongest bidder $1$ is always favored by the optimum, $x_1^* \geq x_1^{\rm M}$;
\item
The weakest bidder $n$ is always discriminated against by the optimum, $x_n^* \leq x_n^{\rm M}$.
\end{enumerate}
\end{proposition}

Proposition \ref{prop: optimum vs myerson} compares the interim winning probabilities of the strongest and weakest bidders across two auctions---optimal and Myerson's.
It shows that the optimum allocates the good more often to the strongest bidder at the expense of the weakest one.
It is instructive to unpack the argument behind this proposition to see the economics at play as well as the role of its assumptions.
Let $\bm\delta^*$ and $\bm\delta^{\rm M}$ denote the $\delta$-transforms of the two reduced forms $\bm x^*$ and $\bm x^{\rm M}$.
Below, we assume that all bidders are active in both auctions in the sense that $\delta_i^*(t), \delta_i^{\rm M}(t) > 0$ for all $t > 0$.
In Supplementary Appendix \ref{sec: proof optimum vs myerson}, we complete the argument by handling bidders who are inactive in either auction.
We also show that, under the stated assumptions, Myerson's auction has no exclusion, so both principal curves indeed satisfy $\sum_{i=1}^n \delta_i(t) = t$ at all times.

As explained in the previous section, Myerson's auction equates bidders' MVVs along its principal curve, that is,
\begin{align*}
\zeta_1\left(e^{-\delta_1^{\rm M}(t)}\right) = \zeta_i\left(e^{-\delta_i^{\rm M}(t)}\right) \;\;\forall i \neq 1.
\end{align*}
Since the strongest bidder's distribution $F_1$ is more equal in the sense of the star order than any other bidder $i$'s distribution $F_i$, their MVVs satisfy $\zeta_1 \geq \zeta_i$ whenever $\zeta_i \geq 0$.
As a result, bidder $1$'s $\delta$-transform corresponding to Myerson's auction is the highest one.

By equating MVVs, Myerson's auction completely ignores the curvature of bidders' revenue functions---the channel through which interim winning probabilities affect ex ante investments---leaving bidder $1$ with the highest MR along its principal curve:
\begin{align}
\label{eq: optimum ev myerson}
\zeta_1\left(e^{-\delta_1^{\rm M}(t)}\right)h'\left(e^{\delta_1^{\rm M}(t)-t}\right) \geq \zeta_i\left(e^{-\delta_i^{\rm M}(t)}\right)h'\left(e^{\delta_i^{\rm M}(t)-t}\right) \;\;\forall i \neq 1.
\end{align}
This suggests that the auctioneer can raise revenue by reallocating more winning probability to bidder $1$.
Since the cumulative height along the principal curve is fixed at $t$, i.e., $\sum_{i=1}^n \delta_i^{\rm M}(t) = t$, the only such reallocation must be at the expense of other weaker bidders.
But this is precisely what the optimum does---it equates MRs along its principal curve, internalizing the curvature in $h$, so that
\begin{align}
\label{eq: optimum ev equalization}
\zeta_1\left(e^{-\delta_1^{*}(t)}\right)h'\left(e^{\delta_1^{*}(t)-t}\right) = \zeta_i\left(e^{-\delta_i^{*}(t)}\right)h'\left(e^{\delta_i^{*}(t)-t}\right) \;\;\forall i \neq 1,
\end{align}
and the cumulative height stays the same, $\sum_{i=1}^n \delta_i^{*}(t) = t$.
Since bidders' MRs are strictly decreasing in own $\delta$-transforms (due to regularity), we should have $\delta_1^{*}(t) \geq \delta_1^{\rm M}(t)$, which is equivalent to the first point in Proposition \ref{prop: optimum vs myerson}.

As can be seen from our argument, the only place where the star order is used is to conclude that bidder $1$ is favored in the quantile space in Myerson's auction, which is the same as $\delta_1^{\rm M}(t) \geq \delta_i^{\rm M}(t)$.
So, the conclusion of this proposition holds under much weaker conditions; in fact, $x_1^*(u) \geq x_1^{\rm M}(u)$ for all high enough $u$ without any extra structure.
On the other hand, not much can be said at this level of generality about the bidders other than the strongest and the weakest one.
The reason is that, relative to Myerson's auction, the optimum simultaneously reallocates interim winning probability from them (to the strongest bidder) and to them (from the weakest bidder).
Which force dominates depends on fine details of the model.

The comparison has a direct counterpart when the auctioneer maximizes expected welfare, $\sum_{i=1}^n\mathbf{E}[\theta_i h(x_i(\theta_i))]$, rather than revenue.
In this case, $J_i(x_i,\theta_i) = \theta_i h(x_i)$, and the natural benchmark is the efficient auction, which allocates the good to the bidder with the highest value and is optimal when $h$ is linear.
Along its principal curve, the efficient auction equates bidders' values rather than their MVVs, so a bidder whose distribution is higher in the sense of first-order stochastic dominance, denoted $\geq_{fosd}$, has a higher $\delta$-transform and, since $h'$ is increasing, a higher MR.
Hence, the argument above goes through with MVVs replaced by values---if the environment is regular and $F_1 \geq_{fosd} F_i \geq_{fosd} F_n$ for all $i$, then $x_1^* \geq x_1^{\rm E}$ and $x_n^* \leq x_n^{\rm E}$, where $\bm x^{\rm E}$ is the reduced form induced by the efficient auction---and the intuition is as before.
Neither the hazard rate order nor the star order is needed here, since both were used only to rank bidders' MVVs in the quantile space.

\subsection{Methodology behind Theorem \ref{thm: optimal no exclusion}}
\label{sec: methodology}

The derivation of \eqref{eq: static system} in Section \ref{sec: optimality} was heuristic---it presumed the conjecture \eqref{eq: fixed point H}, extracted its implications for bidders on the principal curve, and simply required that no bidder at the top want to join.
The proof of Theorem \ref{thm: optimal no exclusion} must instead establish that the equi-MR reduced form solves the auctioneer's problem among all feasible reduced forms, conjectured or not.
This is where the $\delta$-transforms fully pay off---they turn the auctioneer's problem into a textbook optimal control problem in which \eqref{eq: static system} reappears as the first-order condition, and regularity guarantees its sufficiency.

Given a reduced form $\bm x$ (not necessarily feasible), the auctioneer's revenue can be written as
\begin{align}
\label{eq: revenue delta}
\sum_{i=1}^n \int_0^\infty e^{-t}R_i(\delta_i(t),t)dt,
\;\;\text{where}\;\;
R_i(\delta_i,t) := \int_0^{\delta_i}\frac{\partial}{\partial x_i}H_i(e^{\tau-t},e^{-\tau})d\tau.
\end{align}
The flow payoff $R_i$ aggregates the MRs of bidder $i$'s types along the parametrization underlying the $\delta$-transform up to position $\delta_i$, so, by construction,
\begin{align}
\label{eq: R derivative}
\frac{\partial}{\partial \delta_i}R_i(\delta_i,t) = \frac{\partial}{\partial x_i}H_i\left(e^{\delta_i-t},e^{-\delta_i}\right)
\end{align}
is exactly the MR that \eqref{eq: static system} equalizes.

Putting the revenue and the feasibility characterization in Theorem \ref{thm: feasibility}, we obtain the following representation of the auctioneer's problem.

\begin{proposition}
\label{prop: control reformulation}
A reduced form is optimal if and only if its $\delta$-transforms solve the following program:
\begin{align}
\label{eq: control problem}
\max_{\bm\delta \in \mathscr{D}^n}
\int_0^\infty e^{-t}\sum_{i=1}^n R_i(\delta_i(t),t)dt
\;\;\text{s.t.}\;\;\int_0^t e^{-\tau}\sum_{i=1}^n \delta_i'(\tau)d\tau \leq 1 - e^{-\sum_{i=1}^n \delta_i(t)} \;\;\forall t \in \mathbf{R}_+.
\end{align}
Moreover, there exists an optimal reduced form that is extremal if and only if there exists a solution to \eqref{eq: control problem} satisfying
\begin{align}
\label{eq: delta front}
\sum_{i=1}^n \delta_i'(t) = \bm 1_{[0,T]}(t) \;\;\text{for some } T \in [0,\infty].
\end{align}
\end{proposition}

Proposition \ref{prop: control reformulation} restates the optimal auction problem as maximization over principal curves parametrized by $\bm\delta = (\delta_1,...,\delta_n) \in \mathscr{D}^n$.
Specifically, the program in \eqref{eq: control problem} is an optimal control problem in which $\bm\delta$ serves as the state variable and the rates $\bm\delta'$ are the controls.
The feasibility constraint is then a pure state constraint in the sense of Chapter 6 in \citet{seierstad1986optimal}.
Indeed, it can be written as $\mu(t) \leq 1-e^{-\sum_{i=1}^n\delta_i(t)}$, where $\mu$ is an auxiliary state variable with dynamics $\mu'(t)=e^{-t}\sum_{i=1}^n\delta_i'(t)$ and $\mu(0)=0$.
The second part of the proposition, which follows from Theorem \ref{thm: extremality}, sharpens the search---any extremal reduced form corresponds to a front-loaded path that assigns probability at the unit aggregate rate up to the cutoff time $T$ and ceases thereafter.
Since Theorems \ref{thm: feasibility} and \ref{thm: extremality} do not rely on regularity, Proposition \ref{prop: control reformulation} is valid in irregular environments as well.

To give \eqref{eq: static system} a variational meaning, apply the maximum principle to \eqref{eq: control problem} along the trajectories satisfying \eqref{eq: delta front} with $\delta_i' > 0$ whenever $\delta_i$ is positive.
Since the Hamiltonian is linear in the controls, the unit rate of growth is assigned to the coordinates with the highest MRs, and $p^{\rm s}$ equals that maximum.
The growing coordinates thus share the same MR, whereas those staying at zero have weakly lower ones---which is exactly \eqref{eq: static system}.
Hence, \eqref{eq: static system} is a necessary optimality condition along such trajectories, and the equi-MR path is the trajectory that satisfies it at each time.

The proof itself does not rely on the maximum principle.
In regular environments, bidders' flow payoffs are strictly concave in own $\delta$-transforms, so the control problem in \eqref{eq: control problem} is concave.
For such problems, a feasible path is the unique global optimum whenever it satisfies the first-order conditions together with appropriate multipliers on the state constraint.
In Supplementary Appendix \ref{sec: proof optimal no exclusion}, we construct these multipliers explicitly from $p^{\rm s}$ and verify these conditions for the equi-MR path.

\section{Extensions: the case of exclusion}
\label{sec: optimal exclusion}

In the previous section, we developed our methodology for characterizing the optimal auction when keeping the good is suboptimal.
In some cases, however, exclusion can be beneficial.
For instance, in the regular linear model, this happens when bidders' MVVs are negative at the bottom of the type distribution.
In the language of the previous section, the common MR eventually turns negative, so the equi-MR reduced form assigns positive winning probabilities to types with negative MRs and is no longer optimal.
The question then becomes whom to exclude and how.

Below, we answer this question in nonlinear environments, focusing on two economically important cases---concave environments, covering auctions with risk-averse bidders, as in Section \ref{sec: CRA}; and convex environments, covering the endogenous-valuation model, as in Section \ref{sec: EVs}.
It turns out that the structure of exclusion is very different in the two cases.
In concave environments, the auctioneer excludes at the intensive margin, scaling winning probabilities down.
By contrast, in convex environments, the auctioneer excludes at the extensive margin, cutting low types off.

\subsection{Concave settings}
\label{sec: optimal concave}

We start with the concave case, in which each $H_i$ is concave in $x_i$.
In regular environments, the common MR strictly decreases and reaches zero at the time
\[
T^{\rm s} = \inf\{t \geq 0: p^{\rm s}(t) \leq 0\}.
\]
On the equi-MR path at time $T^{\rm s}$, bidder $i$'s coordinate sits at the type
\[
\overline u_i := e^{-\delta_i^{\rm s}(T^{\rm s})},
\]
and, for every bidder active at $T^{\rm s}$,
\[
\frac{\partial}{\partial x_i}H_i(x_i^{\rm s}(u_i),u_i) < 0 = \frac{\partial}{\partial x_i}H_i(x_i^{\rm s}(\overline u_i),\overline u_i) \;\;\forall u_i < \overline u_i.
\]
This suggests that the auctioneer can do better by reducing bidders' winning probabilities relative to the equi-MR reduced form for bidder $i$'s types below $\overline u_i$.
In what follows, we formalize this idea using fractional score allocations.

\begin{definition}
\label{def: fractional score allocations}
An allocation $\bm z$ is a \textbf{fractional score allocation} if there exist scores $\bm q$ as in Definition \ref{def: score allocation} and right-continuous functions $\bm r = (r_1,\ldots,r_n): [0,1] \to [0,1]^n$ such that
\begin{align}
\label{eq: winner fractional q}
z_i(\bm u) =
\begin{cases}
r_i(u_i)
& \text{if } q_i(u_i) \geq \max_{j \neq i} q_j(u_j),
\\
0
& \text{otherwise}.
\end{cases}
\end{align}
\end{definition}

Fractional score allocations naturally generalize score allocations by letting the winner receive the good with own \textbf{winning fraction} $r_i(u_i) \in [0,1]$ rather than with certainty, while the auctioneer keeps the good with the remaining probability.
Equipped with this definition, we can now characterize the optimum.

\begin{theorem}
\label{thm: optimal concave}
Suppose the environment is regular with each $H_i$ concave in $x_i$ and supermodular.
Then, the auctioneer's problem admits a unique optimal reduced form given by
\begin{align}
\label{eq: optimal x concave}
x_i^*(u_i) =
\begin{cases}
x_i^{\rm s}(u_i) & u_i \geq \overline u_i ,\\
\max\mathrm{argmax}_{x_i \in [0,1]}H_i(x_i,u_i) & u_i < \overline u_i .
\end{cases}
\end{align}
Moreover, $\bm x^*$ is induced by the fractional score allocation whose scores are the canonical scores of the equi-MR reduced form and whose winning fractions are $r_i^* = x_i^*/x_i^{\rm s}$.%
\footnote{As a convention, we set $\sfrac{0}{0} = 0$.}
\end{theorem}

Theorem \ref{thm: optimal concave} extends the characterization of optimal auctions to regular concave environments.
Here, the good is withheld with positive probability, and yet low types may continue to receive it---the auctioneer excludes only the types whose MR is negative already at zero winning probability.%
\footnote{The optimal allocation is not unique; implementing it through fractional score allocations is one convenient way to induce the optimal reduced form.}
Economically, withholding at the intensive margin dominates exclusion at the extensive margin because the units of winning probability below each type's unconstrained optimum are precisely the profitable ones.
Cutting a type off destroys these units, whereas scaling the winning fraction down removes only the rest.
We elaborate on this point in the next section, where we examine the optimal pattern of discrimination in the CRA-preferences model.

The proof extends the multipliers constructed for Theorem \ref{thm: optimal no exclusion} by setting them to zero after $T^{\rm s}$, which is consistent precisely because $p^{\rm s}(T^{\rm s}) = 0$.
The added supermodularity assumption ensures that ironing is not needed when bidders' winning probabilities are scaled down.
Supplementary Appendix \ref{sec: proof optimal concave} verifies in addition that the continuation is feasible---the unconstrained optima never exceed the winning probabilities along the principal curve, so the good is never overallocated.

\subsection{Application to discrimination in the CRA-preferences model}
\label{sec: CRA application}

Before going to the convex case, we use Theorem \ref{thm: optimal concave} to examine how the optimal auction favors bidders as a function of their risk attitudes in the CRA-preferences model.
It turns out that the optimal allocation discriminates against more risk-averse bidders at high types and against less risk-averse bidders at low types.
To make the comparison across risk attitudes as transparent as possible, we assume that all bidders share the same type distribution $F$ with inverse $v = F^{-1}$ and belong to one of two groups, risk-neutral and risk-averse.

A risk-neutral bidder has a linear certainty equivalent equal to $x$, whereas a risk-averse bidder has a strictly convex certainty equivalent $g(x) \leq x$ satisfying $g(0)=0$, $g(1)=1$, $g'>0$, and $g''>0$.
There are $k$ risk-neutral bidders, indexed by $i=1,\ldots,k$, and $n-k$ risk-averse bidders, indexed by $i=k+1,\ldots,n$.
As explained in Section \ref{sec: CRA}, bidders' revenue functions in the quantile space are
\[
H_i(x_i,u_i) = v(u_i)x_i - (1-u_i)v'(u_i)g_i(x_i),
\]
where $g_i(x)=x$ for the risk-neutral group and $g_i = g$ for the risk-averse one, so each $H_i$ is concave in $x_i$.
Supermodularity, which we assume throughout, completes the requirements of Theorem \ref{thm: optimal concave}.

Because the environment is concave and symmetric within each group, the optimal auction is group-symmetric as well.
Hence, it suffices to characterize the interim winning probabilities and their $\delta$-transforms for a representative risk-neutral bidder, denoted $i=1$, and a representative risk-averse bidder, denoted $i=n$.
The system \eqref{eq: static system} specializes to
\begin{align}
\label{eq: CRA system}
\begin{split}
p &= v\left(e^{-\delta_1}\right) - v'\left(e^{-\delta_1}\right)(1-e^{-\delta_1})
= v\left(e^{-\delta_n}\right) - v'\left(e^{-\delta_n}\right)(1-e^{-\delta_n})g'\left(e^{\delta_n-t}\right),\\
t &= k\delta_1 + (n-k)\delta_n,
\end{split}
\end{align}
whose solution gives $\delta_1^{\rm s}$, $\delta_n^{\rm s}$, and the common MR $p^{\rm s}$.
By Theorem \ref{thm: optimal concave}, the optimum coincides with the equi-MR reduced form at the types reached before $T^{\rm s}$.
Once $p^{\rm s}$ reaches zero at $T^{\rm s}$, the risk-neutral bidders' types below $\overline u_1$---the type at which the common MVV vanishes---are excluded outright, whereas the risk-averse bidders' low types continue to trade at the unconstrained optima.
Intuitively, awarding a fractional unit to risk-averse bidders at the bottom of the type distribution lowers information rents while generating only small efficiency losses.

We are interested in comparing $x_1^*$ and $x_n^*$.
Since they are determined through the MR equalization in \eqref{eq: CRA system}, the optimal pattern of discrimination depends on comparing $\frac{\partial H_1}{\partial x_1}$ and $\frac{\partial H_n}{\partial x_n}$, and thus on comparing $g'$ with $1$.
It is easy to see that in this case, $x-g(x)$ is unimodal, attaining its maximum at some $m \in (0,1)$.%
\footnote{Indeed, since $g''>0$, the function $x-g(x)$ is strictly concave. Since $g(0)=0$ and $g(1)=1$, it vanishes at both endpoints and is positive on $(0,1)$. Hence, it attains its maximum at a unique $m \in (0,1)$, characterized by $g'(m)=1$, and since $g'$ is strictly increasing, $g'(x) \gtrless 1$ if and only if $x \gtrless m$.}
For instance, in the quadratic model $g(x)=\alpha x^2 + (1-\alpha)x$, we have $m=\sfrac{1}{2}$ regardless of $\alpha \in (0,1]$.
As shown by GMSZ-CRA, the quadratic model is consistent with many important models of non-expected utility, such as versions of the disappointment-aversion theories of \citet{loomes1986disappointment} and \citet{jia2001generalized}, \citet{koszegi2006model}'s loss-averse utility, and modified mean-variance preferences \citep{blavatskyy2010modifying}.

\begin{proposition}
\label{prop: CRA discrimination}
Consider the model as above, and assume that the environment is regular and that $\frac{\partial}{\partial x_n}H_n(m, m^{\sfrac{1}{(n-1)}})>0$.
Then, $x_1^*(u) \gtreqless x_n^*(u)$ if and only if $u \gtreqless m^{\sfrac{1}{(n-1)}}$ for all $u \in [\overline u_n,1)$, whereas $x_1^*(u) = 0 \leq x_n^*(u)$ for all $u < \overline u_n$.
\end{proposition}

Proposition \ref{prop: CRA discrimination} shows that the optimal auction tilts in favor of the more risk-averse bidder at low quantiles and in favor of the less risk-averse bidder at high quantiles.
Intuitively, when types are low, allocating probability to risk-averse bidders generates relatively smaller information rents.
For high types, the concern shifts toward allocative efficiency, and the auction favors bidders whose valuations respond more strongly to additional probability---namely, the less risk-averse ones.
This reversal of priorities across the type distribution is a direct consequence of MR equalization along the optimal principal curve.

\begin{example}
\label{ex: CRA}
Consider the quadratic specification with $\alpha=1$ and the uniform type distribution, so that $g(x)=x^2$ and $v(u)=u$.
To illustrate, fix $n=3$.%
\footnote{It is routine to verify that the environment is regular for all $n,k$ (see Supplementary Appendix \ref{sec: appendix regularity} for the discussion of regularity in terms of primitives) and that the added assumption in Proposition \ref{prop: CRA discrimination} is satisfied.
The derivations for this example are collected in Supplementary Appendix \ref{sec: appendix CRA example}.}
When $k=2$ (two risk-neutral bidders and one risk-averse bidder), the optimal reduced form satisfies
\[
x_1^*(u) = \frac{2u^2+2u-1}{2u+\sfrac{1}{u}}\cdot\bm 1_{[\sfrac{1}{2},1]}(u),
\;\;
x_n^*(u) = \frac{u}{2(1-u)}\cdot\bm 1_{[0,\sfrac{1}{3})}(u) + \frac{2-u^2-\sqrt{3-2u^2}}{2(1-u)^2}\cdot\bm 1_{[\sfrac{1}{3},1]}(u).
\]
Here, $T^{\rm s}$ is finite, and low types of the risk-averse bidder receive their unconstrained optima, $x_n^*(u)=\frac{u}{2(1-u)}$, as predicted by Theorem \ref{thm: optimal concave}.
When $k=1$ (one risk-neutral bidder and two risk-averse bidders), the optimal reduced form satisfies
\[
x_1^*(u) = \frac{12u^2-8u+(2u-1)^{\sfrac{3}{2}}\sqrt{10u-1}+1}{8u^2}\cdot\bm 1_{[\sfrac{1}{2},1]}(u),
\;\;
x_n^*(u) = \frac{u+1}{2(1-u)+\sfrac{2}{u}}.
\]
In contrast to the previous case, $T^{\rm s}$ is infinite, and the optimum coincides with the equi-MR reduced form everywhere.
In both cases, the tipping point in Proposition \ref{prop: CRA discrimination} equals $m^{\sfrac{1}{2}} = \frac{1}{\sqrt 2} \approx 0.707$, and risk-averse bidders are favored for $u<\frac{1}{\sqrt 2}$ and disadvantaged for $u>\frac{1}{\sqrt 2}$.
\end{example}

\subsection{Convex settings}
\label{sec: optimal convex}

We now complete the characterization of optimal auctions with exclusion by taking up the convex case, in which each $H_i$ is supermodular and $\max\{0,H_i\}$ is convex in $x_i$, as in the endogenous-valuation model of Section \ref{sec: EVs}.
Since $H_i(\cdot,u_i)$ is then nonpositive up to some winning probability and convex beyond it, its maximum over $[0,1]$ is attained at a corner for every type $u_i$.
Exclusion therefore operates at the extensive margin: the optimal reduced form assigns zero winning probability to each bidder's types below own exclusion type, which is now potentially positive.

A natural conjecture, in view of Theorem \ref{thm: optimal concave}, is to truncate the equi-MR reduced form at the types $\overline{\bm u}$.
That is, one sets each bidder's winning probabilities to zero below own type $\overline u_i$, which is reached at the time $T^{\rm s}$ when the common MR hits zero.
In nonlinear settings, however, the auctioneer can typically do better by excluding more types.
To see why, recall from \eqref{eq: revenue delta} that the revenue of the truncated reduced form can be expressed as
\[
\sum_{i=1}^n \int_{\overline u_i}^1 H_i(x_i^{\rm s}(u_i),u_i)du_i = \sum_{i=1}^n \int_0^{T^{\rm s}} e^{-t}R_i(\delta_i^{\rm s}(t),t)dt + \sum_{i=1}^n \int_{T^{\rm s}}^\infty e^{-t}R_i(\delta_i^{\rm s}(T^{\rm s}),t)dt.
\]
Truncating instead at the slightly earlier time $T^{\rm s}-\varepsilon$ changes the revenue at the rate
\[
-\sum_{i=1}^n \delta_i^{\rm s\prime}(T^{\rm s})\int_{T^{\rm s}}^{\infty} e^{-t}\frac{\partial}{\partial x_i}H_i\left(e^{\delta_i^{\rm s}(T^{\rm s})-t},e^{-\delta_i^{\rm s}(T^{\rm s})}\right)dt = - \sum_{i=1}^n  \delta_i^{\rm s\prime}(T^{\rm s}) \cdot \overline u_i H_i(x_i^{\rm s}(\overline u_i),\overline u_i),
\]
where the second expression follows from changing the variable of integration from $t$ to bidder $i$'s winning probability.
For every bidder active at $T^{\rm s}$, $\frac{\partial}{\partial x_i}H_i(x_i^{\rm s}(\overline u_i),\overline u_i) = 0$, which implies $H_i(x_i^{\rm s}(\overline u_i),\overline u_i) \leq 0$ by convexity of the positive part of $H_i$.
The inequality is strict outside of degenerate cases such as Myerson's linear model.
In fact, since these terms differ across bidders, truncating different bidders' winning probabilities at different types can do even better.

To account for this, we let the auctioneer choose bidders' exclusion types directly.
Fix a profile of exclusion types $\underline{\bm u} = \big(\underline u_1,\ldots,\underline u_n\big) \in [0,1]^n$, and let the auctioneer equalize MRs subject to each bidder's types below own exclusion type receiving zero winning probability.
This leads to the following modification of \eqref{eq: static system}:
\begin{align}
\label{eq: capped static system}
\frac{\partial}{\partial x_i}H_i\left(e^{\delta_i(t)-t}, e^{-\delta_i(t)}\right)
\begin{cases}
\leq p(t) & \text{if } \delta_i(t) = 0 < -\ln\underline u_i,\\
= p(t) & \text{if } \delta_i(t) \in (0,-\ln\underline u_i),\\
\geq p(t) & \text{if } \delta_i(t) = -\ln\underline u_i > 0,
\end{cases}
\quad
\sum_{i=1}^n \delta_i(t) = \min\left\{t, -\sum_{j=1}^n\ln\underline u_j\right\}.
\end{align}
In words, \eqref{eq: capped static system} is a system of $n+1$ equations in $n+1$ unknowns $(\bm\delta(t),p(t))$ at each time $t$.
Akin to \eqref{eq: static system}, it asks that MRs be equalized among the bidders who have left the top, $e^{-\delta_i(t)} < 1$, and are not excluded, $e^{-\delta_i(t)} > \underline u_i$, and that the bidders at the top have weakly lower MRs.
However, the excluded bidders may have weakly higher MRs than the remaining ones; as a result, the solution to \eqref{eq: capped static system} in general differs from the equi-MR path of Section \ref{sec: main result}.

To ensure that this construction is well-behaved for every profile of exclusion types, we strengthen regularity accordingly.
We say that an environment is \textbf{strongly regular} if Condition (A) holds together with the following strengthening of Condition (B).
\begin{itemize}
\item[(B$^*$)] For each profile $\underline{\bm u}$, the solution $\big(\bm\delta^{\rm s}\big(t|\underline{\bm u}\big), p^{\rm s}\big(t|\underline{\bm u}\big)\big)$ to \eqref{eq: capped static system}---unique under Condition (A) at each $t < -\sum_{j=1}^n\ln\underline u_j$ under the same convention as in Section \ref{sec: optimality}---is such that $p^{\rm s}\big(\cdot|\underline{\bm u}\big)$ is strictly decreasing on that time interval and each $\delta_i^{\rm s}\big(\cdot|\underline{\bm u}\big)$ is weakly increasing, and strictly so whenever bidder $i$ is active, i.e., $0 < \delta_i(t) < -\ln\underline u_i$.
\end{itemize}
Condition (B) is Condition (B$^*$) evaluated at $\underline{\bm u} = \bm 0$, so a strongly regular environment is regular and $\bm\delta^{\rm s} = \bm\delta^{\rm s}\big(\cdot|\bm 0\big)$.
We refer to $\bm\delta^{\rm s}\big(\cdot|\underline{\bm u}\big)$ as the \textbf{$\underline{\bm u}$-capped equi-MR path}.
We denote the reduced form obtained from it via \eqref{eq: delta transform inverse} by $\bm x^{\rm s}\big(\cdot|\underline{\bm u}\big)$, the \textbf{$\underline{\bm u}$-capped equi-MR reduced form}.
By Theorem \ref{thm: feasibility}, it is a feasible reduced form, and its exclusion types are $\underline{\bm u}$.%
\footnote{Strictly speaking, this is true when $\prod_{i=1}^n \underline u_i > 0$.
However, if $\prod_{i=1}^n \underline u_i = 0$, then bidder $i$'s exclusion type $\lim_{t \to \infty}e^{-\delta_i^{\rm s}\big(t|\underline{\bm u}\big)}$ may be larger than $\underline u_i$.
This distinction is inconsequential for our analysis.}
For later reference, we let $T_i\big(\underline{\bm u}\big) := -\ln\big(\underline u_i x_i^{\rm s}\big(\underline u_i|\underline{\bm u}\big)\big)$ denote the \textbf{exclusion time} at which bidder $i$'s coordinate reaches own exclusion type, which is infinite if it never does.
We also set $T\big(\underline{\bm u}\big) := \max_i T_i\big(\underline{\bm u}\big)$ for the time at which the $\underline{\bm u}$-capped equi-MR path stops.

The revenue that a profile of exclusion types generates can be expressed as
\begin{align}
\label{eq: revenue given T}
\sum_{i=1}^n \int_{\underline u_i}^1 H_i\big(x_i^{\rm s}\big(u_i|\underline{\bm u}\big),u_i\big)du_i.
\end{align}
The exclusion types enter \eqref{eq: revenue given T} through two channels---directly, through the lower limits of integration; and indirectly, through the $\underline{\bm u}$-capped equi-MR reduced form, which depends on the entire profile $\underline{\bm u}$.
Because of the latter, the exclusion types cannot be optimized bidder by bidder.
Yet, maximizing \eqref{eq: revenue given T} is a finite-dimensional problem, and it yields the following characterization.

\begin{theorem}
\label{thm: optimal convex}
Suppose the environment is strongly regular with each $H_i$ supermodular and $\max\{0,H_i\}$ convex in $x_i$.
Then, the auctioneer's problem admits a unique optimal reduced form $\bm x^*$, namely, the $\underline{\bm u}^*$-capped equi-MR reduced form, where $\underline{\bm u}^*$ is a maximizer of \eqref{eq: revenue given T}.
Moreover, for every bidder $i$,
\begin{align}
\label{eq: FOC exclusion times}
\underline u_i^* H_i\big(x_i^*\big(\underline u_i^*\big),\underline u_i^*\big)
\leq \int_{T_i\big(\underline{\bm u}^*\big)}^{T\big(\underline{\bm u}^*\big)} e^{-t}p^{\rm s}\big(t|\underline{\bm u}^*\big)dt
\;\;\text{with equality whenever } \underline u_i^* < 1.
\end{align}
Finally, $\bm x^*$ is extremal and is induced by the score allocation with its canonical scores.
\end{theorem}

Theorem \ref{thm: optimal convex} shows that the optimal auction can be found by choosing bidders' exclusion types and then equalizing MRs above them.
In contrast to Theorem \ref{thm: optimal concave}, the optimum remains extremal and is still induced by the canonical scores.

The optimality conditions \eqref{eq: FOC exclusion times} pin down the exclusion types.
To interpret them, it is convenient to think about the optimum through its principal curve, which descends from the northeast corner as time progresses.
Bidder $i$'s coordinate reaches own exclusion type at the time $T_i\big(\underline{\bm u}^*\big)$ and stays there forever after, and we say that the bidder drops out at that time.
Of course, the auction itself is static, and an earlier drop-out means nothing more than a higher exclusion type.
In this language, \eqref{eq: FOC exclusion times} says that the revenue generated by bidder $i$'s exclusion type---the left-hand side---equals the flow of the common MR collected between own drop-out time and the last one---the right-hand side.
In particular, the bidder that drops out last does so when own exclusion type generates zero revenue, whereas every other bidder drops out while own exclusion type still generates positive revenue.

To see where this trade-off comes from, differentiate \eqref{eq: revenue given T} with respect to $\underline u_i$ to obtain
\[
-H_i\big(x_i^{\rm s}\big(\underline u_i|\underline{\bm u}\big),\underline u_i\big)
+ \sum_{j=1}^n \int_{\underline u_j}^1 \frac{\partial}{\partial x_j}H_j\big(x_j^{\rm s}\big(u_j|\underline{\bm u}\big),u_j\big)\frac{\partial}{\partial \underline u_i}x_j^{\rm s}\big(u_j|\underline{\bm u}\big)du_j.
\]
The first term is the revenue lost by excluding bidder $i$'s marginal type $\underline u_i$.
The second term is the change in the revenue from all other types, whose winning probabilities are rearranged along the $\underline{\bm u}$-capped equi-MR path.
We show in Supplementary Appendix \ref{sec: proof optimal convex} that it equals
\begin{align}
\label{eq: revenue change exclusion}
\frac{1}{\underline u_i}\int_{T_i\big(\underline{\bm u}\big)}^{T\big(\underline{\bm u}\big)} e^{-t}p^{\rm s}\big(t|\underline{\bm u}\big)dt.
\end{align}
To see why, note that only the types reached after $T_i\big(\underline{\bm u}\big)$ are affected, since the path up to that time does not depend on $\underline u_i$, and that their MRs equal the common MR by \eqref{eq: capped static system}.
Furthermore, differentiating the last equation of \eqref{eq: capped static system} and using the definition of the principal curve, one can obtain
\[
\sum_{j=1}^n \frac{\partial}{\partial \underline u_i}x_j^{\rm s}\big(u_j|\underline{\bm u}\big)du_j = \frac{e^{-t}}{\underline u_i}dt
\]
for the types reached at time $t \in \big(T_i\big(\underline{\bm u}\big),T\big(\underline{\bm u}\big)\big)$.
Weighting this by the common MR and integrating over $t$ gives \eqref{eq: revenue change exclusion}.

The optimality condition \eqref{eq: FOC exclusion times} also explains why all bidders drop out at once in Myerson's linear model but at different times in general.
In the linear model, $H_i(x_i,u_i) = \zeta_i(u_i)x_i$, so, at any positive winning probability, $H_i = 0$ if and only if $\frac{\partial H_i}{\partial x_i} = 0$.
As a result, both sides of \eqref{eq: FOC exclusion times} vanish together when every bidder drops out at $T^{\rm s}$.
In nonlinear settings, this equivalence breaks down because MRs depend on winning probabilities.
When the auctioneer excludes bidder $i$'s low types, the other bidders' winning probabilities go up.
The auctioneer thus gives up the revenue of the excluded types in exchange for the revenue that the other bidders generate at higher winning probabilities, and \eqref{eq: FOC exclusion times} balances these two at bidder-specific exclusion types.
The following example illustrates.

\begin{example}
\label{ex: ev exclusion}
Consider the modified setting of Section \ref{sec: example} with the auctioneer's value (or shipping cost) $c \in (0,1)$, $\beta_i = \beta$, and $v_i(a_i,\theta_i)=\kappa_i \theta_i + 2\sqrt{(1-\kappa_i)\theta_i}\,a_i$ for some $\kappa_i \in [0,1]$, so that $\omega_i(x_i,\theta_i) = \theta_i(\kappa_ix_i+(1-\kappa_i)x_i^2)$ and
\[
H_i(x_i,u_i)= u_i^{1/\beta}(\kappa_ix_i+(1-\kappa_i)x_i^2)-cx_i.
\]
Clearly, we recover the symmetric version of our illustrative example (but with $c > 0$) when $\kappa_i = 0$ and the standard symmetric linear model when $\kappa_i = 1$.
In the latter case, the optimal auction is the second-price auction with the reserve type $c^{\beta}$, at which bidders' MVVs equal $c$, so that all bidders drop out at the same level of MR, that is, $0$.

Allow now for heterogeneous $\kappa_i$, ordered as $\kappa_1 \geq \cdots \geq \kappa_n$.
Using Theorem \ref{thm: optimal convex}, one can show that the bidders are excluded sequentially along the optimal principal curve, $T_1\big(\underline{\bm u}^*\big) \geq \cdots \geq T_n\big(\underline{\bm u}^*\big)$, and that every bidder but the last one is excluded while the common MR is still positive.%
\footnote{Strong regularity holds whenever $\beta$ times the number of bidders with $\kappa_i > 0$ is strictly less than $1$, in particular, in the case considered below, so the theorem is directly applicable.}

For example, take $\beta = \sfrac{1}{2}$ with one linear bidder, $\kappa_1 = 1$, and $n-1$ bidders whose values stem from investment only, $\kappa_2 = \cdots = \kappa_n = 0$, and let $c < 2^{-4/3}$, which ensures that bidder $1$ is active when the others drop out.
The latter bidders are symmetric and hence drop out together at some $T_2\big(\underline{\bm u}^*\big) = \cdots = T_n\big(\underline{\bm u}^*\big)$.
After that time, only bidder $1$ is active, so that $e^{-\delta_1^{\rm s}(t|\underline{\bm u}^*)} = e^{-t}\big/\big(\underline u_n^*\big)^{n-1}$ and
\[
p^{\rm s}\big(t|\underline{\bm u}^*\big) = \frac{e^{-2t}}{\big(\underline u_n^*\big)^{2(n-1)}} - c.
\]
Bidder $1$ is excluded last, once the common MR reaches zero, that is, at Myerson's reserve type $\underline u_1^* = c^{1/2}$.
At $T_n\big(\underline{\bm u}^*\big)$, the common MR also equals the MR of the other bidders, $2\underline u_n^*e^{-T_n(\underline{\bm u}^*)} - c$, which gives $e^{-T_n(\underline{\bm u}^*)} = 2\big(\underline u_n^*\big)^{2n-1}$.
Substituting these expressions into \eqref{eq: FOC exclusion times}, we finally obtain
\[
\underline u_2^* = \cdots = \underline u_n^* = \left(\frac{c^{3/2}}{2}\right)^{1/( 3n)} > c^{1/2} = \underline u_1^*.
\]
Thus, the bidders whose values stem from investment are excluded at higher types than in Myerson's auction, whereas the exclusion type of the linear bidder remains the same.
\end{example}

The proof of Theorem \ref{thm: optimal convex} starts by using convexity of $\max\{0,H_i\}$ to conclude that the optimal reduced form is extremal.
By Theorem \ref{thm: extremality}, this means that its $\delta$-transforms sum to $\min\{t,T\}$ for some $T \in [0,\infty]$.
We then relax the auctioneer's optimal control problem of Section \ref{sec: methodology} at the exclusion types of the optimum itself.
Specifically, we keep only the constraints of \eqref{eq: capped static system}---each coordinate between zero and own exclusion level, and their sum equal to $\min\{t,-\sum_{j}\ln\underline u_j^*\}$---so that the relaxed problem can be solved time by time.
Strong regularity ensures that the time-by-time solution is the $\delta$-transform of a reduced form, so the relaxation is tight and the optimum is the $\underline{\bm u}^*$-capped equi-MR path.
Maximality of \eqref{eq: revenue given T} at $\underline{\bm u}^*$ and the optimality conditions \eqref{eq: FOC exclusion times} then follow from the envelope theorem applied to the exclusion types.

\section{Final remarks}
\label{sec: conclusion}

In this paper, we develop new tools to study asymmetric auction settings with unidimensional types, unit demands, and a single good for sale.
We believe that these tools extend well beyond our specific framework and can help analyze other economically significant problems.
One immediate application is procurement with pre-auction contractible quality as in \citet{scoringandfavoritism} and another is information design under privacy constraints in the spirit of \citet{he2021private}, who established a connection between such information structures and Border-type feasibility in auction settings.
Other potential applications include persuasion, delegation, and related problems studied in \citet{kleiner2021extreme} through the lens of symmetric Border-type feasibility.

Extending our approach to feasibility and extremality to more general auction settings is another fruitful direction for future research.
Concretely, our method for identifying the principal curve of tightest feasibility constraints delivers a sharp characterization of feasibility and an explicit description of extremality.
It may extend to settings with multiunit demands or even multidimensional types.
Border-type feasibility characterizations are known in some of these environments (e.g., see \citet{che2013generalized}), and our approach may help make them more tractable and operational.

Our approach to optimality may also prove valuable in broader principal--agent models and in persuasion.
Even within our auction framework, we focus on regular environments in which ironing is not needed, and we illustrate how the results apply to settings with endogenous valuations and non-EU preferences.
This focus is largely due to space constraints.
Further developing the economics of optimal auctions in these applications and extending the analysis beyond regular environments are natural next steps.

\bibliographystyle{abbrvnat}
\bibliography{references}

\appendix
\section{Appendix}
\label{sec: appendix}

\subsection{The $\delta$-transform}
\label{sec: appendix delta}

In this section, we discuss the reparametrization of CDFs on the unit interval through the $\delta$-transform.
Throughout, every CDF $x \in \mathscr{X}$ is equipped with its \textit{right-continuous inverse}, i.e.,
\[
x^{-1}(\iota) := \sup\{u \in [0,1] : x(u) \leq \iota\}.
\]
For a function $\delta \in \mathscr{D}$, we use $\delta^{-1}$ to denote its \textit{left-continuous inverse} defined by
\[
\delta^{-1}(\iota) := \inf\{t \in \mathbf{R}_+ : \delta(t) \geq \iota\}
\]
with the convention $\inf\emptyset = \infty$.

As explained in Section \ref{sec: delta transform}, the $\delta$-transform is constructed in two steps, and it is convenient to give the inverse in \eqref{eq: delta transform} a name.
For a CDF $x \in \mathscr{X}$, let
\begin{align}
\label{eq: phi transform}
\psi(\iota) := (ux)^{-1}(\iota)
\end{align}
denote the inverse of the product $u \mapsto ux(u)$, termed the \textbf{$\psi$-transform of $x$}.
Since $ux(u) \leq u$, we have $\psi(\iota) \in [\iota,1]$, so that $\iota/\psi(\iota)$ is well-defined for $\iota > 0$ and weakly increasing.
Furthermore, since $ux$ is strictly increasing whenever positive, $\psi$ is continuous with $\psi(0) = \underline u$, which is the exclusion type of $x$ defined in Section \ref{sec: preliminaries}.
We abbreviate by $\Psi$ the subset of $\mathscr{X}$ that consists of CDFs satisfying these properties.

Figure \ref{fig: phi transform} illustrates the construction: panel (a) shows a step-shaped CDF $x$, panel (b) its product with $u$, and panel (c) the resulting $\psi$.
The picture highlights how $\psi$ inverts the joint map $ux$---jumps of $x$ become flats of $\psi$, and flats of $x$ become linear stretches of $\psi$.

\begin{figure}[!htb]
\centering
\begin{subfigure}[b]{0.32\textwidth}
\centering
\begin{tikzpicture}
\draw (0,0) -- (0,4) -- (4,4) -- (4,0) -- (0,0);
\draw[ultra thick, black] (0,0) -- (1,0);
\draw[ultra thick, black] (1,1) -- (2,2);
\draw[ultra thick, black] (2,2) -- (3,2);
\draw[ultra thick, black] (3,3) -- (4,4);
\draw[thick, dotted, black] (1,0) -- (1,1);
\node[circle,fill=black,inner sep=0pt,minimum size=4pt] at (3,3) {};
\node[circle,fill=black,inner sep=0pt,minimum size=4pt] at (1,1) {};
\draw[thick, dotted, black] (1,1) -- (1,0) node[anchor=north]{\sfrac{1}{4}};
\draw[thick, dotted, black] (2,2) -- (2,0) node[anchor=north]{\sfrac{1}{2}};
\draw[thick, dotted, black] (3,3) -- (3,0) node[anchor=north]{\sfrac{3}{4}};
\draw[thick, dotted, black] (4,0) -- (4,0) node[anchor=north]{1};
\draw[thick, dotted, black] (1,1) -- (0,1) node[anchor=east]{\sfrac{1}{4}};
\draw[thick, dotted, black] (2,2) -- (0,2) node[anchor=east]{\sfrac{1}{2}};
\draw[thick, dotted, black] (3,3) -- (0,3) node[anchor=east]{\sfrac{3}{4}};
\draw[thick, dotted, black] (0,4) -- (0,4) node[anchor=east]{1};
\draw[thick, dotted, black] (0,0) -- (0,0) node[anchor=north]{0};
\end{tikzpicture}
\caption{$x$}
\end{subfigure}
\begin{subfigure}[b]{0.32\textwidth}
\centering
\begin{tikzpicture}
\draw (0,0) -- (0,4) -- (4,4) -- (4,0) -- (0,0);
\draw[ultra thick, black] (0,0) -- (1,0);
\draw[ultra thick, black] (1,1/16 * 4) .. controls(1.5, 9/64 * 4) .. (2, 4/16 * 4);
\draw[ultra thick, black] (2, 4/16 * 4) -- (3, 6/16 * 4);
\draw[ultra thick, black] (3, 9/16 * 4) .. controls(3.5, 49/64 * 4) .. (4, 4);
\node[circle,fill=black,inner sep=0pt,minimum size=4pt] at (3,9/16*4) {};
\node[circle,fill=black,inner sep=0pt,minimum size=4pt] at (1,1/16 * 4) {};
\draw[thick, dotted, black] (1,1/16 * 4) -- (1,0) node[anchor=north]{\sfrac{1}{4}};
\draw[thick, dotted, black] (2,4/16 * 4) -- (2,0) node[anchor=north]{\sfrac{1}{2}};
\draw[thick, dotted, black] (3,9/16 * 4) -- (3,0) node[anchor=north]{\sfrac{3}{4}};
\draw[thick, dotted, black] (4,0) -- (4,0) node[anchor=north]{1};
\draw[thick, dotted, black] (1,1/16 * 4) -- (0,1/16 * 4) node[anchor=east]{\sfrac{1}{16}};
\draw[thick, dotted, black] (2,4/16 * 4) -- (0,4/16 * 4) node[anchor=east]{\sfrac{1}{4}};
\draw[thick, dotted, black] (3,6/16 * 4) -- (0,6/16 * 4) node[anchor=east]{\sfrac{3}{8}};
\draw[thick, dotted, black] (3,9/16*4) -- (0,9/16*4) node[anchor=east]{\sfrac{9}{16}};
\draw[thick, dotted, black] (0,4) -- (0,4) node[anchor=east]{1};
\draw[thick, dotted, black] (0,0) -- (0,0) node[anchor=north]{0};
\end{tikzpicture}
\caption{$ux$}
\end{subfigure}
\begin{subfigure}[b]{0.32\textwidth}
\centering
\begin{tikzpicture}
\draw (0,0) -- (0,4) -- (4,4) -- (4,0) -- (0,0);
\draw[ultra thick, black] (0,1) -- (1/16 * 4,1);
\draw[ultra thick, black] (1/16 * 4,1) .. controls(9/64 * 4,1.5) .. (4/16 * 4,2);
\draw[ultra thick, black] (4/16 * 4,2) -- (6/16 * 4,3) -- (9/16 * 4,3);
\draw[ultra thick, black] (9/16 * 4,3) .. controls(49/64 * 4, 3.5) .. (4,4);
\node[circle,fill=black,inner sep=0pt,minimum size=4pt] at (0,1) {};
\draw[thick, dotted, black] (1/16 * 4,1) -- (0,1) node[anchor=east]{\sfrac{1}{4}};
\draw[thick, dotted, black] (4/16 * 4,2) -- (0,2) node[anchor=east]{\sfrac{1}{2}};
\draw[thick, dotted, black] (6/16 * 4,3) -- (0,3) node[anchor=east]{\sfrac{3}{4}};
\draw[thick, dotted, black] (0,4) -- (0,4) node[anchor=east]{1};
\draw[thick, dotted, black] (1/16 * 4,1) -- (1/16 * 4,0) node[anchor=north,xshift=4pt]{\sfrac{1}{16}};
\draw[thick, dotted, black] (4/16 * 4,2) -- (4/16 * 4,0) node[anchor=north]{\sfrac{1}{4}};
\draw[thick, dotted, black] (6/16 * 4,3) -- (6/16 * 4,0) node[anchor=north]{\sfrac{3}{8}};
\draw[thick, dotted, black] (9/16 * 4,3) -- (9/16 * 4,0) node[anchor=north]{\sfrac{9}{16}};
\draw[thick, dotted, black] (4,0) -- (4,0) node[anchor=north]{1};
\draw[thick, dotted, black] (0,0) -- (0,0) node[anchor=north]{0};
\end{tikzpicture}
\caption{$\psi$}
\end{subfigure}
\caption{Illustration of $\psi$ for $x(u)=u\cdot\bm 1_{[\sfrac{1}{4},\sfrac{1}{2})\cup [\sfrac{3}{4},1]}(u)+\sfrac{1}{2}\cdot\bm 1_{[\sfrac{1}{2},\sfrac{3}{4})}(u)$.}
\label{fig: phi transform}
\end{figure}

Then, the $\delta$-transform of $x$ at time $t$ equals $-\ln \psi(\iota)$ evaluated at $\iota = e^{-t}$.
Since $\psi \in \Psi$, it is easy to see that $\delta \in \mathscr{D}$, i.e., $\delta$ is absolutely continuous with $\delta' \in [0,1]$ and $\delta(0) = 0$.
Indeed, evaluating at $t = 0$, we have $\delta(0) = -\ln \psi(1) = 0$.
Monotonicity of $\psi$ makes $\delta$ weakly increasing, and monotonicity of $\iota/\psi(\iota) = e^{\delta(t)-t}\big|_{t=-\ln\iota}$ makes $\delta(t) - t$ weakly decreasing.
Such a function is necessarily Lipschitz with constant $1$, thus it is an element of $\mathscr{D}$.

Below, we first discuss the opposite direction and complete the proof of Lemma \ref{lm: delta transform}.
Then, we elaborate more on the geometry of the $\delta$-transform.

\subsubsection{Proof of Lemma \ref{lm: delta transform}}
\begin{proof}
We have already argued that the assignment in \eqref{eq: delta transform} maps $\mathscr{X}$ into $\mathscr{D}$.
By construction, for any $u$ with $ux(u) > 0$, we have $\delta^{-1}(-\ln u) = -\ln\psi^{-1}(u) < \infty$, where $\psi^{-1}$ denotes the right-continuous inverse of $\psi$ and satisfies $\psi^{-1}(u) = ux(u)$.
Hence,
\[
e^{\delta(t)-t}\big|_{t = \delta^{-1}(-\ln u)} = \frac{e^{-t}}{\psi(e^{-t})}\bigg|_{e^{-t} = ux(u)} = \frac{ux(u)}{\psi(\psi^{-1}(u))} = x(u).
\]
On the other hand, if $u > 0$ and $ux(u) = 0$, then $\delta^{-1}(-\ln u) = \infty$ and $e^{\delta(t)-t}\big|_{t = \delta^{-1}(-\ln u)} = 0 = x(u)$, since $\psi(0) \geq u > 0$.
Taking the two cases together, we conclude that the output of \eqref{eq: delta transform inverse} applied to the $\delta$-transform of the CDF $x$ is $x$ itself.

We now turn to the converse assignment in \eqref{eq: delta transform inverse}.
Given a function $\delta \in \mathscr{D}$, we set $\psi(\iota) := e^{-\delta(-\ln\iota)}$ for $\iota \in (0,1]$ and $\psi(0) := \lim_{\iota \to 0}\psi(\iota)$, which exists by monotonicity.
Due to the properties of $\delta$, $\psi$ is an element of $\Psi$.
Hence, $x$, defined as $x(u) := \psi^{-1}(u)/u$, extended to $u = 0$ by right continuity, is an element of $\mathscr{X}$.
By construction, $ux(u) = \psi^{-1}(u)$, meaning that the $\delta$-transform of the CDF $x$ is $\delta$ itself.
\end{proof}

\subsubsection{Geometry of the $\delta$-transform}
\label{sec: appendix geometry}

The $\delta$-transform yields an explicit geometric decomposition of a CDF $x \in \mathscr{X}$ into
(a) vertical segments (jumps), corresponding to $\delta' = 0$;
(b) horizontal segments with no mass, corresponding to $\delta' = 1$;
and (c) strictly increasing continuous segments, corresponding to $\delta' \in (0,1)$.
To see it, recall from \eqref{eq: delta parametrization} that the curve $t \mapsto \big(e^{-\delta(t)},e^{\delta(t)-t}\big)$ traces the completed graph of $x$ from the top-right corner downward.
When $\delta' = 0$, the type coordinate dwells while the probability coordinate falls---the curve descends through a jump of $x$.
When $\delta' = 1$, the probability coordinate dwells while the type coordinate falls---the curve moves left along a flat of $x$.
When $\delta' \in (0,1)$, both coordinates fall, and the curve follows a strictly increasing stretch of $x$.

For example, Figure \ref{fig: phi transform} plots the $\psi$-transform of
\[
x(u) = u\cdot\bm 1_{[\sfrac{1}{4},\sfrac{1}{2})\cup[\sfrac{3}{4},1]}(u) + \sfrac{1}{2}\cdot\bm 1_{[\sfrac{1}{2},\sfrac{3}{4})}(u),
\]
which (a) jumps at $u = \sfrac{1}{4}$ and $u = \sfrac{3}{4}$, (b) is horizontal on $[\sfrac{1}{2},\sfrac{3}{4})$, and (c) coincides with the diagonal on $[\sfrac{1}{4},\sfrac{1}{2})$ and $[\sfrac{3}{4},1]$; below the exclusion type $\underline u = \sfrac{1}{4}$, it vanishes.
The resulting $\delta$-transform satisfies
\[
\delta'(t) =
\begin{cases}
0, & t \in [2\ln(\sfrac{4}{3}),\ln(\sfrac{8}{3})) \cup [2\ln 4,\infty),
\\
1, & t \in [\ln(\sfrac{8}{3}),2\ln 2),
\\
\sfrac{1}{2}, & t \in [0,2\ln(\sfrac{4}{3})) \cup [2\ln 2,2\ln 4),
\end{cases}
\]
and is illustrated in Figure \ref{fig: delta illustration}.
Because the $\delta$-transform proceeds from $u = 1$ downward, the first jump of $x$, at the exclusion type $u = \sfrac{1}{4}$, appears as the last interval where $\delta' = 0$.
This interval is unbounded, since the curve dwells at this jump forever.
The other pieces of the decomposition are read off similarly.
The stretch $[0,\sfrac{1}{4})$ on which $x$ vanishes is never traversed by the curve.

\begin{figure}[!htb]
\centering
\begin{tikzpicture}[scale=1.15]

\draw[->] (0,0)--(0,4.5);
\draw[->] (0,0)--(4.5,0) node[below]{$t$};

\pgfmathsetmacro{\ta}{2*ln(4/3)}
\pgfmathsetmacro{\tb}{ln(8/3)}
\pgfmathsetmacro{\tc}{2*ln(2)}
\pgfmathsetmacro{\td}{2*ln(4)}
\pgfmathsetmacro{\T}{\td}
\pgfmathsetmacro{\Xscale}{4/\T}

\pgfmathdeclarefunction{delta}{1}{%
	\pgfmathparse{
		(#1 < \ta) * (0.5*#1)
	  + (#1 >= \ta && #1 < \tb) * (ln(4/3))
	  + (#1 >= \tb && #1 < \tc) * (ln(4/3) + (#1 - \tb))
	  + (#1 >= \tc && #1 < \td) * (0.5*#1)
	  + (#1 >= \td) * (ln(4))
	}%
}

\draw[thick, red]
plot[domain=0:\T, samples=500]
({\Xscale*\x},{4*exp(-delta(\x))});

\draw[ultra thick, blue, dotted]
plot[domain=0:\T, samples=500]
({\Xscale*\x},{4*exp(delta(\x)-\x)});

\pgfmathsetmacro{\ustarA}{3/4}
\pgfmathsetmacro{\tstarA}{\ta}
\pgfmathsetmacro{\XstarA}{\Xscale*\tstarA}
\pgfmathsetmacro{\YstarA}{4*\ustarA}

\draw[thick, dotted, black] (0,\YstarA) -- (\XstarA,\YstarA);
\node[anchor=east] at (0,\YstarA) {$u'=\sfrac{3}{4}$};
\draw[thick, dotted, black] (\XstarA,\YstarA) -- (\XstarA,0);
\filldraw[black] (\XstarA,\YstarA) circle (2pt);
\node[anchor=south west] at (\XstarA+0.05,\YstarA) {$x(\sfrac{3}{4})=\sfrac{3}{4}$};

\pgfmathsetmacro{\ustarB}{1/2}
\pgfmathsetmacro{\tstarB}{\tc}
\pgfmathsetmacro{\XstarB}{\Xscale*\tstarB}
\pgfmathsetmacro{\YstarB}{4*\ustarB}

\draw[thick, dotted, black] (0,\YstarB) -- (\XstarB,\YstarB);
\node[anchor=east] at (0,\YstarB) {$u''=\sfrac{1}{2}$};
\draw[thick, dotted, black] (\XstarB,\YstarB) -- (\XstarB,0);
\filldraw[black] (\XstarB,\YstarB) circle (2pt);
\node[anchor=south west] at (\XstarB+0.05,\YstarB) {$x(\sfrac{1}{2})=\sfrac{1}{2}$};

\pgfmathsetmacro{\Xta}{\Xscale*\ta}
\pgfmathsetmacro{\Xtc}{\Xscale*\tc}
\draw (\Xta,0) -- (\Xta,-0.12)
	node[below, xshift=-8pt] {$t'=2\ln(\sfrac{4}{3})$};
\draw (\Xtc,0) -- (\Xtc,-0.12)
	node[below, xshift=8pt] {$t''=2\ln 2$};

\end{tikzpicture}
\caption{Reading of the CDF $x$ in Figure \ref{fig: phi transform} at $u = u', u''$ from the parametric curves $e^{-\delta(t)}$ and $e^{\delta(t)-t}$, shown in red (solid) and blue (dotted).}
\label{fig: delta illustration}
\end{figure}

\subsection{Feasibility}
\label{sec: appendix feasibility}

This section proves Theorem \ref{thm: feasibility}.
The proof proceeds in two steps.
The first key step is Lemma \ref{lm: aux G}, which analyzes \eqref{eq: auxiliary problem} and expresses its solution in terms of the $\psi$-transforms $\bm\psi = (\psi_1,\ldots,\psi_n)$ of a reduced form $\bm x$, defined coordinatewise in \eqref{eq: phi transform}, and the map $\bm\rho =(\rho_1,\ldots,\rho_n)$ given by
\begin{align}
\label{eq: rho}
\rho_i(s) := \psi_i\left(\bar\psi^{-1}(s)\right),\;\;\text{where}\;\; \bar\psi := \sqrt[n]{\prod_{i=1}^n \psi_i}.
\end{align}
We use these objects extensively in the proofs of Theorems \ref{thm: feasibility} and \ref{thm: extremality}.

\subsubsection{The auxiliary problem}
\begin{lemma}
\label{lm: aux G}
A reduced form $\bm x$ is feasible if and only if $G \leq 1$, where $G$ is the optimal value in \eqref{eq: auxiliary problem}.
Furthermore, letting $\bm\psi$ denote the vector of $\psi$-transforms of $\bm x$ and $\bm\rho$ be defined by \eqref{eq: rho}, we have the following.
\begin{enumerate}
\item[(1)]
The point $\bm\rho(s)$ is optimal in \eqref{eq: auxiliary problem} and uniquely optimal whenever $s \geq \bar\psi(0)$ and $s > 0$.
\item[(2)]
The map $s \in [0,1] \mapsto \bm\rho(s)$ is continuous and weakly increasing, and it is strictly increasing in at least one coordinate on $[\bar\psi(0),1]$.
\item[(3)]
The optimal value is given by
\begin{align}
\label{eq: G explicit}
G(s) = s^n + \sum_{i=1}^n\int_{\bar\psi^{-1}(s)}^1 \iota d\ln\psi_i(\iota).
\end{align}
\end{enumerate}
\end{lemma}

\begin{proof}
Note that, in \eqref{eq: auxiliary problem}, $s^n$ is fixed in the objective, so decreasing $u_i$ increases $\int_{u_i}^1 x_i(u)du$ due to monotonicity of $x_i$.
As a result, a maximizer can be taken with the product constraint binding, which gives
\[
G(s) = \max_{\bm u \in [0,1]^n} B(\bm u) \;\;\text{s.t.}\;\;\prod_{i=1}^n u_i = s^n
\]
for each $s \in [0,1]$.
Since every cutoff vector $\bm u \in [0,1]^n$ satisfies the constraint at $s^n = \prod_{i=1}^n u_i$, the reduced form $\bm x$ is feasible if and only if $G(s) \leq 1$ for all $s \in [0,1]$.
Below, we prove three remaining claims.

\textit{Claim (1).}
To begin, remark that, for each $s \in (0,1)$, the auxiliary program is concave, since each integral $\int_{u_i}^1 x_i(u)du$ is concave in $u_i$ and the constraint set is convex.
It follows that the auxiliary program can be solved through the following necessary (and sufficient due to concavity) first-order conditions:
\[
\lambda\prod_{j=1}^n u_j \in [x_i(u_i-)u_i,x_i(u_i)u_i] \;\;\forall i,
\]
where $\lambda \geq 0$ is the Lagrange multiplier on the constraint, and we use the convention $x_i(0-) = 0$.

If $\lambda > 0$, then the constraint binds, so $s^n = \prod_{i=1}^n u_i$.
It follows from the definition of the right-continuous inverse that the first-order conditions can be rewritten as
\[
u_i = \psi_i(\lambda s^n) \;\;\forall i.
\]
Taking the geometric mean across bidders gives $s = \bar\psi(\lambda s^n)$, and therefore $s \geq \bar\psi(0)$ due to monotonicity of $\bar\psi$.
On the other hand, if $\lambda = 0$, the first-order conditions require $x_i(u_i-) = 0$, that is,
\[
u_i \in [0,\psi_i(0)] \;\;\forall i,
\]
which implies $s \leq\bar\psi(0)$.

Equipped with this description of the first-order conditions, we distinguish between various values of the parameter $s$.

\begin{itemize}

\item
\textit{Case $s \in \big(\bar\psi(0),1\big)$.}
As shown above, we must have $\lambda > 0$, so that $s =\bar\psi(\lambda s^n)$.
The largest Lagrange multiplier with this property can be obtained by taking the right-continuous inverse of $\bar\psi$, that is, $\lambda s^n =\bar\psi^{-1}(s)$, which is strictly positive by continuity of $\bar\psi$.
By the definition of $\bm\psi$, the point $\bm\rho(s)$ satisfies the first-order conditions with this multiplier and the constraint as equality, which implies that it is optimal in \eqref{eq: auxiliary problem}.
Since each $\psi_i$ is weakly increasing, it is constant on every interval on which $\bar\psi$ is constant.
Hence, if $\tilde\lambda \in (0,\lambda)$ is such that $s =\bar\psi\big(\tilde\lambda s^n\big)$, then $\bm\psi\big(\tilde\lambda s^n\big) = \bm\rho(s)$, which shows that $\bm\rho(s)$ is the unique maximizer.

\item
\textit{Case $s =\bar\psi(0) > 0$.}
By construction, $\bm\psi(0)$ satisfies the first-order conditions with $\lambda = 0$; thus, it is optimal.
Any maximizer with $\lambda = 0$ satisfies $\prod_{i=1}^n u_i \geq s^n = \prod_{i=1}^n \psi_i(0)$ and $u_i \leq \psi_i(0)$ for all $i$, which forces $\bm u = \bm\psi(0)$.
Any maximizer with $\lambda > 0$ equals $\bm\psi(\lambda s^n) = \bm\psi(0)$ because each $\psi_i$ is constant on $\big[0,\bar\psi^{-1}(\bar\psi(0))\big]$.
Hence, $\bm\psi(0)$ is the unique maximizer, and it coincides with $\bm\rho(s)$ at $s =\bar\psi(0)$ by the same constancy.

\item
\textit{Case $s \in \big(0,\bar\psi(0)\big)$.}
By the same argument, $\bm\psi(0)$ satisfies the first-order conditions with $\lambda = 0$ and is optimal, and it equals $\bm\rho(s)$, since $\bar\psi^{-1}(s) = 0$.

\item
\textit{Case $s = 1$.}
Remark that the only feasible point in the auxiliary program is $\bm\psi(1) = (1,\ldots,1)$, and thus it is uniquely optimal.

\item \textit{Case $s = 0$.}
Note that the constraint in the auxiliary program has no bite.
By construction, $x_i = 0$ on $\big[0,\psi_i(0)\big)$; as a result, $\bm\psi(0) = \big(\underline u_1,\ldots,\underline u_n\big)$ maximizes the objective in \eqref{eq: auxiliary problem} when the constraint is ignored.
\end{itemize}

\textit{Claim (2).}
The map $s \mapsto \bm\rho(s)=\bm\psi\big(\bar\psi^{-1}(s)\big)$ is weakly increasing as a composition of weakly increasing maps.
It is continuous because $\bar\psi^{-1}$ can jump only across an interval on which $\bar\psi$ is constant, where each $\psi_i$ is constant as well.
To see strict monotonicity in at least one coordinate, take $\underline s < \overline s$ in $\big[\bar\psi(0),1\big]$.
By continuity of $\bar\psi$, we have $\prod_{i=1}^n \rho_i(s) = s^n$ for $s \in \big[\bar\psi(0),1\big]$; hence,
\[
\prod_{i=1}^n \rho_i\big(\underline s\big) = \underline s^n < \overline s^n = \prod_{i=1}^n \rho_i\big(\overline s\big),
\]
which shows that at least one coordinate strictly increases.

\textit{Claim (3).}
By Claim (1), $\bm\rho(s)$ is optimal in \eqref{eq: auxiliary problem} for every $s \in [0,1]$, so $G(s) = s^n + \sum_{i=1}^n\int_{\rho_i(s)}^1 x_i(u)du$, and it suffices to verify that
\[
\int_{\psi_i(\iota)}^1 x_i(u)du = \int_\iota^1 \sigma d\ln\psi_i(\sigma)
\]
for each $i$ and $\iota = \bar\psi^{-1}(s)$.
Indeed, since $\psi_i^{-1}(u) = ux_i(u)$ on $[\psi_i(0),1]$, we obtain
\begin{align*}
\int_{\psi_i(\iota)}^1 x_i(u)du
= & \int_{\psi_i(0)}^{1}\bm 1_{(\iota,1]}\big(\psi_i^{-1}(u)\big)\psi_i^{-1}(u)d\ln u\\
= & \int_{0}^{1}\bm 1_{(\iota,1]}(\sigma)\sigma d\ln\psi_i(\sigma)\\
= & \int_\iota^1 \sigma d\ln\psi_i(\sigma).
\end{align*}
The first line uses $x_i(u)du=\psi_i^{-1}(u)d\ln u$, and the second line is due to the change of variables of integration, i.e., $u=\psi_i(\sigma)$, which maps $d\ln u$ on $(\psi_i(0),1]$ to $d\ln\psi_i(\sigma)$ on $(0,1]$.
\end{proof}

\subsubsection{Proof of Theorem \ref{thm: feasibility}}
\begin{proof}
In view of Lemma \ref{lm: aux G}, since $G$ is strictly increasing on $\big[0,\bar\psi(0)\big]$, $\bm x$ is feasible if and only if $G(s) \leq 1$ for all $s \geq\bar\psi(0)$.
On this interval, the maximizers of the auxiliary program described in that lemma trace the principal curve.
Indeed, recalling that $\nu_i(t) = e^{-\delta_i(t)} = \psi_i(e^{-t})$ by the construction of the $\delta$-transform, we have $\bar\psi(e^{-t}) = e^{-\frac{1}{n}\sum_{i=1}^n\delta_i(t)}$.
Moreover, setting $s =\bar\psi(e^{-t})$, we get $\bm\rho(s) = \bm\psi(e^{-t}) = \bm\nu(t)$, because each $\psi_i$ is constant wherever $\bar\psi$ is constant.
Since $t \mapsto\bar\psi(e^{-t})$ is continuous and takes every value above $\bar\psi(0)$, we conclude that $G(s) \leq 1$ on $\big[\bar\psi(0),1\big]$ if and only if $B(\bm\nu(t)) \leq 1$ for all $t \in \mathbf{R}_+$.
This establishes the first part of the theorem.

To see the expression for the Border constraint along the principal curve, evaluate \eqref{eq: G explicit} at $s =\bar\psi(e^{-t})$ to obtain
\[
B(\bm\nu(t)) = e^{-\sum_{i=1}^n\delta_i(t)} + \sum_{i=1}^n\int_{e^{-t}}^1 \iota d\ln\psi_i(\iota)
= e^{-\sum_{i=1}^n\delta_i(t)} + \int_0^t e^{-\tau}\sum_{i=1}^n\delta_i'(\tau)d\tau,
\]
where the second equality follows from the change of variables $\iota = e^{-\tau}$ and $\ln\psi_i(e^{-\tau}) = -\delta_i(\tau)$.
This is exactly \eqref{eq: border along curve}.
\end{proof}

\section{Supplementary Appendix}
\label{sec: supplementary appendix}
\subsection{Extremality}
\label{sec: appendix extremality}

In this section, we study extremality of reduced forms and prove Theorems \ref{thm: extremality} and \ref{thm: score allocations}.
We show the two theorems in conjunction with each other, establishing three statements in the following order:
\begin{enumerate}
\item[]
\textit{Statement I.}
If a feasible reduced form $\bm x$ is extremal, then Border's constraint binds along the principal curve, i.e., \eqref{eq: extremality delta} holds, as explained in Section \ref{sec: extremal reduced forms}.
\item[]
\textit{Statement II.}
If a feasible reduced form $\bm x$ satisfies \eqref{eq: extremality delta}, then it is induced by the score allocation with its canonical scores.
\item[]
\textit{Statement III.}
If a feasible reduced form $\bm x$ is induced by a score allocation, then it is extremal and the allocation inducing it is unique.
\end{enumerate}
Statement I corresponds to the ``only if'' direction of Theorem \ref{thm: extremality}, and the other two statements establish Theorem \ref{thm: score allocations} provided that Statement I has been shown.
Taken together, I--III imply the ``if'' direction of Theorem \ref{thm: extremality}.

The most challenging part of the analysis is the first statement, which says that Border's constraint binds along the principal curve for extremal reduced forms.
In view of Lemma \ref{lm: aux G} and the proof of Theorem \ref{thm: feasibility}, this is equivalent to requiring that the function $G$, which is defined in \eqref{eq: auxiliary problem} and then explicitly solved for in \eqref{eq: G explicit}, satisfy
\[
G(s) = 1\;\;\forall s \in \big[\bar\psi(0),1\big].
\]
This condition is equivalent to nonexistence of a \textbf{sag}, that is, an open interval $\big(\underline s,\overline s\big)$ in $\big[\bar\psi(0),1\big]$ on which $G$ is strictly less than $1$.

For each feasible (potentially nonextremal) reduced form $\bm x$, if the set of sags is nonempty, then, since $G(1)=1$ and $G<1$ on $\big[0,\bar\psi(0)\big)$, there are two mutually exclusive and exhaustive cases.
\begin{enumerate}
\item
There exists a sag so that $G$ equals $1$ at its endpoints.
\item
There exists a sag that starts at $\bar\psi(0)$ and ends at $\overline s$ so that $G$ equals $1$ on $\big[\overline s,1\big]$.
\end{enumerate}
We shall show that both cases are inconsistent with extremality of $\bm x$, thereby establishing that no extremal reduced form admits a sag.
To this end, we first state and prove six auxiliary lemmas that record various properties of elements of the feasible set and its extreme points.
With these tools in hand, we then prove the two theorems stated in Section \ref{sec: extremality}.

\subsubsection{Preliminaries}

Lemmas \ref{lm: perturbations top} and \ref{lm: perturbations top and bot} show that certain modifications of a feasible reduced form leave (some of) Border's constraints unchanged.
For brevity, when we need to emphasize the dependence on the reduced form, we write $B_{\bm x}$ for Border's constraint in \eqref{border} corresponding to $\bm x$.

The other four lemmas in this section study sags of extremal reduced forms.
Specifically, Lemma \ref{lm: piecewisec and jumps} establishes that $x_i$ is piecewise-constant with finitely many jumps along a sag of the principal curve.
Lemmas \ref{lm: jumps endpoints top} and \ref{lm: jumps endpoints bot} discuss jumps at the endpoints of the sag.
Finally, Lemma \ref{lm: persons gap} shows that any sag corresponds to a nontrivial portion of the principal curve.
Throughout this analysis, we use the following notation.
Given a reduced form $\bm x$ and two cutoffs $\underline v<\overline v$, define the set of jumps of $x_i$ in $\big[\underline v,\overline v\big]$ by
\[
D_i\big(\underline v,\overline v\big):=\left\{u\in[\underline v,\overline v]|x_i(u-)\neq x_i(u)\right\}
\]
with the convention $x_i(0-)=0$.
We extend the definition of $D_i$ to degenerate intervals by setting
\[
D_i(v,v):=\{v\}.
\]

\begin{lemma}[Perturbations at the bottom of the type space]
\label{lm: perturbations top}
Let $\bm x$ be a feasible reduced form that satisfies $B_{\bm x}(\overline{\bm v})=1$ for some $\overline{\bm v}\in[0,1]^n$, $\prod_{i=1}^n\overline v_i>0$.
Consider $\bm\Delta=(\Delta_1,\ldots,\Delta_n)$ with $\Delta_i:[0,1]\to\mathbf{R}$ such that $\bm x+\bm\Delta$ is a reduced form, $\bm\Delta\neq0$, and $\Delta_i(u)=0$ on $\big(\overline v_i,1\big]$ for each $i$.
Then, $\bm x+\bm\Delta$ is feasible if and only if
\begin{align}
B_{\bm x+\bm\Delta}(\bm u)\leq1\;\;\forall\bm u\;\;\text{s.t.}\;\;\bm 0\leq\bm u\leq\overline{\bm v}.
\label{eq: Border reduced top}
\end{align}
\end{lemma}
\begin{proof}
The ``only if'' direction is immediate, so we focus on the converse.
Specifically, suppose that \eqref{eq: Border reduced top} holds.
We now show that $\bm x+\bm\Delta$ is feasible, meaning that the inequality in \eqref{eq: Border reduced top} holds for all $\bm u$ other than $\bm 0\leq\bm u\leq\overline{\bm v}$.
To this end, we distinguish between several cases.

\textit{Case $\prod_{i=1}^n u_i=0$.}
In this case,
\[
B_{\bm x+\bm\Delta}(\bm u)
\leq \sum_{i=1}^n \int_{0}^1 (x_i(u)+\Delta_i(u))du
= B_{\bm x+\bm\Delta}(\bm 0)
\leq 1,
\]
where the first inequality holds because $\prod_{i=1}^n u_i = 0$ and $x_i+\Delta_i \geq 0$, and the last inequality follows from \eqref{eq: Border reduced top}.

Conversely, suppose $\bm u\in[0,1]^n$ is such that $\prod_{i=1}^n u_i>0$ and $u_i>\overline v_i$ precisely for the first $k \geq 1$ coordinates, i.e., $i=1,\ldots,k$---after relabeling if necessary.

\textit{Case $k=n$.}
Since $\bm u\geq\overline{\bm v}$, we have $\int_{u_i}^1 \Delta_i(u)du=0$ for each $i$, and hence
\[
B_{\bm x+\bm\Delta}(\bm u)=B_{\bm x}(\bm u)\leq 1.
\]

\textit{Case $k<n$.}
Set $a_1=\prod_{i=1}^{k}u_i$, $\overline a_1=\prod_{i=1}^{k}\overline v_i$, $a_2=\prod_{i=k+1}^{n}u_i$, and $\overline a_2=\prod_{i=k+1}^{n}\overline v_i$.

Since $\bm x$ is feasible and $B_{\bm x}(\overline{\bm v})=1$, we have
\[
B_{\bm x}\big(u_1,\ldots,u_k,\overline v_{k+1},\ldots,\overline v_n\big)\leq B_{\bm x}\big(\overline{\bm v}\big),
\]
which expands to
\begin{align}
a_1 \cdot \overline a_2 + \sum_{i=1}^{k}\int_{u_i}^{1}x_i(u)du
\leq
\overline a_1 \cdot \overline a_2 + \sum_{i=1}^{k}\int_{\overline v_i}^{1}x_i(u)du.
\label{eq: top aux 1}
\end{align}
Next, consider evaluating $B_{\bm x+\bm\Delta}$ at $\big(\overline v_1,\ldots,\overline v_k,u_{k+1},\ldots,u_n\big)$.
Using $\Delta_i(u)=0$ for $u>\overline v_i$, we obtain
\begin{align}
\overline a_1 \cdot a_2
+\sum_{i=1}^{k}\int_{\overline v_i}^{1}x_i(u)du
+\sum_{i=k+1}^{n}\int_{u_i}^{1}x_i(u)du
+\sum_{i=1}^n\int_{u_i}^{1}\Delta_i(u)du
\leq 1,
\label{eq: top aux 2}
\end{align}
due to \eqref{eq: Border reduced top}.

Since $a_1 \geq \overline a_1$ and $a_2 \leq \overline a_2$, we have $(a_1-\overline a_1)\cdot (a_2-\overline a_2)\leq 0$, i.e.,
\[
a_1 \cdot a_2 \leq
\overline a_1 \cdot a_2 + a_1 \cdot \overline a_2 -\overline a_1 \cdot \overline a_2.
\]
Combining this inequality with \eqref{eq: top aux 1} and \eqref{eq: top aux 2}, we conclude that
\begin{align*}
B_{\bm x+\bm\Delta}(\bm u)
&= a_1 \cdot a_2 + \sum_{i=1}^n \int_{u_i}^1 x_i(u)du + \sum_{i=1}^n \int_{u_i}^1 \Delta_i(u)du
\\
&\leq \overline a_1 \cdot a_2 + a_1 \cdot \overline a_2 -\overline a_1 \cdot \overline a_2 + \sum_{i=1}^n \int_{u_i}^1 x_i(u)du + \sum_{i=1}^n \int_{u_i}^1 \Delta_i(u)du
\\
&\leq \overline a_1 \cdot a_2 + \sum_{i=1}^k \int_{\overline v_i}^1 x_i(u)du + \sum_{i=k+1}^n \int_{u_i}^1 x_i(u)du + \sum_{i=1}^n \int_{u_i}^1 \Delta_i(u)du
\\
&\leq 1,
\end{align*}
as claimed.
\end{proof}

\begin{lemma}[Perturbations in the middle of the type space]
\label{lm: perturbations top and bot}
Let $\bm x$ be a feasible reduced form that satisfies $B_{\bm x}\big(\underline{\bm v}\big)=B_{\bm x}\big(\overline{\bm v}\big)=1$ for some $\underline{\bm v},\overline{\bm v}\in[0,1]^n$ such that $\prod_{i=1}^n\overline v_i>0$, $\underline{\bm v}\neq\overline{\bm v}$, and $\underline{\bm v}\leq\overline{\bm v}$.
Consider $\bm\Delta=(\Delta_1,\ldots,\Delta_n)$ with $\Delta_i:[0,1]\to\mathbf{R}$ such that $\bm x+\bm\Delta$ is a reduced form, $\bm\Delta\neq0$, $\Delta_i(u)=0$ on $\big[0,\underline v_i\big)\cup\big(\overline v_i,1\big]$ for each $i$, and
\[
\sum_{i=1}^n\int_{\underline v_i}^{\overline v_i}\Delta_i(u)du=0.
\]
Then, $\bm x+\bm\Delta$ is feasible if and only if
\begin{align}
B_{\bm x+\bm\Delta}(\bm u)\leq1\;\;\forall\bm u\;\;\text{s.t.}\;\;\underline{\bm v}\leq\bm u\leq\overline{\bm v}.
\label{eq: Border reduced top and bot}
\end{align}
\end{lemma}
\begin{proof}
In view of Lemma \ref{lm: perturbations top}, the reduced form $\bm x+\bm\Delta$ is feasible if and only if \eqref{eq: Border reduced top} holds.
We thus need to show that \eqref{eq: Border reduced top} is satisfied if and only if \eqref{eq: Border reduced top and bot} holds.
The ``only if'' direction is immediate, so we focus on the converse, again distinguishing between several cases.

\textit{Case $\prod_{i=1}^n u_i=0$.}
In this case,
\[
B_{\bm x+\bm\Delta}(\bm u)
\leq \sum_{i=1}^n \int_{0}^1 (x_i(u)+\Delta_i(u))du
= B_{\bm x}(\bm 0)
\leq 1,
\]
where the first inequality holds because $\prod_{i=1}^n u_i=0$ and $x_i+\Delta_i \geq 0$, the equality follows from $\sum_{i=1}^n\int_{0}^{1}\Delta_i(u)du=0$, and the last inequality follows from feasibility of $\bm x$.

Now, suppose $\bm u\leq\overline{\bm v}$ is such that $\prod_{i=1}^n u_i>0$ and $u_i<\underline v_i$ precisely for the first $k \geq 1$ coordinates, i.e., $i=1,\ldots,k$---after relabeling if necessary.

\textit{Case $k=n$.}
Since $\bm u\leq\underline{\bm v}$, it follows that $\sum_{i=1}^{k}\int_{u_i}^{1}\Delta_i(u)du=0$, and hence
\[
B_{\bm x+\bm\Delta}(\bm u)=B_{\bm x}(\bm u)\leq 1.
\]

\textit{Case $k<n$.}
The argument is identical to the one in the proof of Lemma \ref{lm: perturbations top} but replaces $\overline{\bm v}$ with $\underline{\bm v}$.
Setting $a_1=\prod_{i=1}^{k}u_i$, $\underline a_1=\prod_{i=1}^{k}\underline v_i$, $a_2=\prod_{i=k+1}^{n}u_i$, $\underline a_2=\prod_{i=k+1}^{n}\underline v_i$, and repeating the same steps as before, we can obtain
\begin{gather*}
a_1 \cdot \underline a_2+\sum_{i=1}^{k}\int_{u_i}^{1}x_i(u)du
\leq
\underline a_1 \cdot \underline a_2+\sum_{i=1}^{k}\int_{\underline v_i}^{1}x_i(u)du,\\
\underline a_1 \cdot a_2
+\sum_{i=1}^{k}\int_{\underline v_i}^{1}x_i(u)du
+\sum_{i=k+1}^{n}\int_{u_i}^{1}x_i(u)du
+\sum_{i=1}^n\int_{u_i}^{1}\Delta_i(u)du
\leq 1,\\
a_1 \cdot a_2
\leq
a_1 \cdot \underline a_2+\underline a_1 \cdot a_2-\underline a_1 \cdot \underline a_2.
\end{gather*}
Combining these inequalities yields the desired result:
\[
B_{\bm x+\bm\Delta}(\bm u)\leq 1.
\]
\end{proof}

\begin{lemma}[Structure along a sag]
\label{lm: piecewisec and jumps}
Let $\bm x$ be an extremal reduced form.
Consider a sag $\big(\underline s,\overline s\big)$.
Then, for each $i$ such that $\rho_i\big(\underline s\big)<\rho_i\big(\overline s\big)$, the function $x_i$ is piecewise-constant on $\big[\rho_i\big(\underline s\big),\rho_i\big(\overline s\big)\big]$ with finitely many elements in $D_i\big(\rho_i\big(\underline s\big),\rho_i\big(\overline s\big)\big)$, i.e., $x_i$ takes finitely many values.
\end{lemma}
\begin{proof}
Let $i$ be such that $\rho_i\big(\underline s\big)<\rho_i\big(\overline s\big)$.
There are two mutually exclusive possibilities: either all (weak) lower contour sets of $x_i$ are open in $\big(\rho_i\big(\underline s\big),\rho_i\big(\overline s\big)\big)$, i.e.,
\begin{align}
\label{eq: openness}
\left\{u \in \big(\rho_i\big(\underline s\big),\rho_i\big(\overline s\big)\big)| x_i(u) \leq x_i\big(\underline v\big)\right\} \text{ is open }\;\forall \underline v \in \big(\rho_i\big(\underline s\big),\rho_i\big(\overline s\big)\big),
\end{align}
or there exists at least one such lower contour set that is not open.
In what follows, we first prove, in two logical steps, that the second case cannot occur given that $\bm x$ is extremal, and then complete the proof in one additional step.

\textit{Step 1.}
Towards a contradiction, suppose that \eqref{eq: openness} is violated at $\underline v$, i.e., the associated lower contour set equals $\big(\rho_i\big(\underline s\big),\underline v\big]$.
This means that there exists a closed interval $\big[\underline v,\overline v\big]$ compactly contained in $\big(\rho_i\big(\underline s\big),\rho_i\big(\overline s\big)\big)$ such that $x_i(u) > x_i\big(\underline v\big)$ for all $u \in \big(\underline v,\overline v\big]$.
Denote the value of
\[
\max_{\bm u\in[0,1]^n}B_{\bm x}(\bm u)\;\;\text{s.t.}\;\;u_i\in\big[\underline v,\overline v\big]
\]
by $\omega$.
Since $\bm x$ is feasible, we have $\omega \leq 1$.
We now show that $\omega$ is strictly less than one.

Suppose not, i.e., $\omega=1$.
Let $\bm u$ be a maximizer that attains $\omega$ and set $s \in [0,1]$ so that $s^n=\prod_{j=1}^n u_j$.
Note that $s > 0$; otherwise, $B_{\bm x}(\bm u) = 1$ would require $x_j = 0$ on $[0,u_j)$ for all $j$, i.e., $u_j \leq \psi_j(0)$, contradicting $u_i \geq \underline v > \rho_i\big(\underline s\big) \geq \psi_i(0)$.
By definition of the auxiliary problem in Lemma \ref{lm: aux G}, $G(s)$ is an upper bound on $\omega$.
This bound is tight because $\bm x$ is assumed to be feasible, i.e., $G \leq 1$, and $\omega = 1$.
Since $G$ is strictly less than one on $\big[0,\bar\psi(0)\big)$, we must have $s \geq \bar\psi(0)$; thus,
\[
B_{\bm x}(\bm u) = G(s) = B_{\bm x}(\bm\rho(s))
\]
as established in the proof of Theorem \ref{thm: feasibility}.

Since $G$ is strictly less than one on $\big(\underline s,\overline s\big)$, we cannot have $s \in \big(\underline s,\overline s\big)$.
By Claim (2) of Lemma \ref{lm: aux G}, $\rho_i(s) \leq \rho_i\big(\underline s\big) < \underline v$ whenever $s \leq \underline s$, and $\rho_i(s) \geq \rho_i\big(\overline s\big) > \overline v$ whenever $s \geq \overline s$.
In either case, we obtain $\rho_i(s)\not\in \big[\underline v,\overline v\big]$.
As a result, $\bm u\neq \bm\rho(s)$, contradicting the uniqueness of the optimizer stated in Lemma \ref{lm: aux G}.
Therefore, we conclude that $\omega < 1$, as claimed.

\textit{Step 2.}
We now construct a perturbation $\bm\Delta:[0,1]\to\mathbf{R}^n$ so that $\bm x \pm \varepsilon \bm\Delta$ is a feasible reduced form provided that $\varepsilon > 0$ is small enough.
The existence of such a perturbation contradicts extremality of $\bm x$, thereby ruling out the existence of a non-open lower contour set.

First of all, for each $j \neq i$, set $\Delta_j = 0$.
As shown in the proof of Theorem 1 in \citet{kleiner2021extreme}, there exist a number $x^* \in \big[x_i\big(\underline v\big),x_i\big(\overline v\big)\big]$ and an interval $\big[\underline v^*,\overline v^*\big]$ contained in $\big[\underline v,\overline v\big]$ so that $\Delta_i$ defined by
\begin{align*}
\Delta_i(u) & =
\begin{cases}
x_i(u)-x_i\big(\underline v\big) & u\in\big[\underline v,\underline v^*\big),\\
x^*-x_i(u) & u\in\big[\underline v^*,\overline v^*\big),\\
x_i(u)-x_i\big(\overline v\big) & u\in\big[\overline v^*,\overline v\big],\\
0             & u\not\in\big[\underline v,\overline v\big]
\end{cases}
\end{align*}
is such that $\Delta_i\neq 0$, $x_i \pm \Delta_i$ is in $\mathscr{X}$, and $\int_{\underline v}^{\overline v}\Delta_i(u)du=0$.
In fact, \citet{kleiner2021extreme} provided a closed-form expression for the endpoints of the subinterval used above:
\begin{align*}
\underline v^*  = & \inf_{u\in\big[\underline v,\overline v\big]}u\;\;\text{s.t.}\;\;2x_i(u)\geq x_i\big(\underline v\big)+x^*,\\
\overline v^*  = & \inf_{u\in\big[\underline v,\overline v\big]}u\;\;\text{s.t.}\;\;2x_i(u)\geq x_i\big(\overline v\big)+x^*
\end{align*}
and proved that the value of $x^*$ making $\int_{\underline v}^{\overline v}\Delta_i(u)du=0$ (provided the endpoints are chosen as above) is well-defined.

Set
\[
\overline\varepsilon =\frac{1-\omega}{\max_{u_i\in\big[\underline v,\overline v\big]}\left|\int_{u_i}^{\overline v}\Delta_i(u)du\right|}>0
\]
and choose $\varepsilon \in (0,1)$ such that $\varepsilon \leq \overline\varepsilon$.
Then, for every $\bm u\in[0,1]^n$ with $u_i\in\big[\underline v,\overline v\big]$, we have
\[
B_{\bm x \pm \varepsilon \bm\Delta}(\bm u) \leq  B_{\bm x}(\bm u)+\varepsilon\left|\int_{u_i}^{1}\Delta_i(u)du\right| \leq  B_{\bm x}(\bm u)+\left(1-\omega\right).
\]
Since $B_{\bm x}(\bm u)\leq\omega$ for every $\bm u\in[0,1]^n$ with $u_i\in\big[\underline v,\overline v\big]$, Border's constraint holds for $\bm x\pm\varepsilon \bm\Delta$ at all such $\bm u$.
For every $\bm u\in[0,1]^n$ with $u_i\not\in \big[\underline v,\overline v\big]$, we have $B_{\bm x \pm \varepsilon \bm\Delta}(\bm u)=B_{\bm x}(\bm u)\leq1$.
Hence, the reduced form $\bm x\pm\varepsilon \bm\Delta$ is feasible.
However, since $\bm x$ is assumed to be an extreme point of the feasible set, this contradicts extremality.
Therefore, no such interval $\big[\underline v,\overline v\big]$ can exist.

To sum up, we have shown that all (weak) lower contour sets of $x_i$ are open in $\big(\rho_i\big(\underline s\big),\rho_i\big(\overline s\big)\big)$.
This means that $x_i$ is piecewise-constant and every jump in that interval other than the last one has an immediate successor. 
We can thus meaningfully talk about consecutive jumps.

\textit{Step 3.}
To conclude the proof, we need to show that the number of jumps is finite.
Towards a contradiction, suppose that $x_i$ admits three consecutive jump points in $\big(\rho_i\big(\underline s\big),\rho_i\big(\overline s\big)\big)$, denoted $\underline v < v^* < \overline v$.
Similarly to the previous step, we rule out this possibility by explicitly constructing a perturbation $\bm\Delta=(\Delta_1,\ldots,\Delta_n):[0,1]\to\mathbf{R}^n$ so that $\bm x \pm \varepsilon \bm\Delta$ is a feasible reduced form provided that $\varepsilon > 0$ is small enough.

First of all, for each $j \neq i$, set $\Delta_j = 0$.
For coordinate $i$, set
\begin{align}
\label{eq: kappa jumps}
\Delta_i(u) = \kappa \cdot\left(\frac{\bm 1_{\big[\underline v,v^*\big)}(u)}{v^*-\underline v} - \frac{\bm 1_{\big[v^*,\overline v\big)}(u)}{\overline v-v^*}\right)
\end{align}
for some $\kappa>0$ small enough to ensure $\Delta_i\neq0$ and $x_i\pm\Delta_i\in \mathscr{X}$, e.g., any choice such that
\begin{align*}
\frac{\kappa}{v^*-\underline v} \leq x_i\big(\underline v\big)-x_i\big(\underline v-\big),
\;\frac{\kappa}{v^*-\underline v} + \frac{\kappa}{\overline v-v^*}\leq x_i(v^*)-x_i(v^*-),
\;\frac{\kappa}{\overline v-v^*} \leq x_i\big(\overline v\big)-x_i\big(\overline v-\big)
\end{align*}
will do.

Then, applying the same argument as in the previous construction, we conclude that the perturbed reduced form $\bm x \pm \varepsilon \bm\Delta$ is feasible for sufficiently small $\varepsilon>0$.
This contradicts extremality of $\bm x$, and hence no such triple jump structure can exist.
\end{proof}

\begin{lemma}[Discontinuity at the top of a sag]
\label{lm: jumps endpoints top}
Let $\bm x$ be an extremal reduced form.
Consider a sag $\big(\underline s,\overline s\big)$ satisfying $G\big(\overline s\big)=1$.
Then, $\rho_i\big(\overline s\big) \in D_i\big(\rho_i\big(\underline s\big),\rho_i\big(\overline s\big)\big)$ for all $i$, i.e., $x_i$ is discontinuous at $\rho_i\big(\overline s\big)$ whenever $\rho_i\big(\underline s\big)<\rho_i\big(\overline s\big)$.
\end{lemma}
\begin{proof}
Fix an index $i$ with $\rho_i\big(\underline s\big)<\rho_i\big(\overline s\big)$, and suppose, towards a contradiction, that $x_i$ is continuous at $\rho_i\big(\overline s\big)$.
We distinguish between two cases.

\textit{Case $\rho_i\big(\overline s\big)<1$.}
Clearly, we have $\overline s<1$.
For $s$ near $\overline s$, define $\bm u(s)$ by keeping all coordinates except $i$ fixed at their values on $\rho(\cdot)$ at $\overline s$, that is, $u_j(s) = \rho_j\big(\overline s\big)$ for $j\neq i$, and adjusting coordinate $i$ to satisfy the product constraint:
\[
u_i(s)=\frac{\rho_i\big(\overline s\big)}{\overline s^{n}}s^n.
\]
By construction, $\prod_{j=1}^n u_j(s)=s^n$, so $\bm u(s)$ is feasible in the optimization problem that defines $G(s)$ in Lemma \ref{lm: aux G}.
Hence, $G(s) \geq \underline G(s)$, where $\underline G(s)$ is the objective evaluated at $\bm u(s)$:
\begin{align}
\underline G(s)
=s^n+\sum_{j=1}^n\int_{u_j(s)}^1x_j(u)du
=s^n+\int_{\frac{\rho_i\big(\overline s\big)}{\overline s^{n}}s^n}^{1}x_i(u)du+\sum_{j \neq i}\int_{\rho_j\big(\overline s\big)}^1x_j(u)du.
\label{eq: G lower bound}
\end{align}
On the one hand, $\underline G\big(\overline s\big)=G\big(\overline s\big)=1$.
On the other hand, since $\big(\underline s,\overline s\big)$ is a sag, we have $\underline G(s) \leq G(s) < 1$ for $s < \overline s$.

By Lemma \ref{lm: piecewisec and jumps}, the function $x_i$ is piecewise-constant on $\big[\rho_i\big(\underline s\big),\rho_i\big(\overline s\big)\big]$ with finitely many jumps.
Since $x_i$ is assumed to be continuous at the right endpoint, it is constant on a left neighborhood of this point.
Consequently, $\underline G$ is an affine function of $s^n$ for $s<\overline s$.
But $\underline G(s)<\underline G\big(\overline s\big)$ for $s < \overline s$, meaning that the (left) slope of this affine function is strictly positive.
Since $x_i$ is continuous at $\rho_i\big(\overline{s}\big)$, the slope of $\underline{G}$ is continuous at $\overline{s}$; therefore, $\underline G(s)>1$ for $s>\overline s$ close to $\overline s$.
This implies that $G(s)$ strictly exceeds $1$, contradicting $\underline G \leq G \leq 1$ on a right neighborhood of $\overline s$, as required by feasibility.
Thus, we conclude that $x_i$ cannot be continuous at $\rho_i\big(\overline s\big)$.

\textit{Case $\rho_i\big(\overline s\big)=1$.}
The argument is virtually identical to the previous case, with the only difference that now $x_i(1)=1$ and continuity at $\rho_i\big(\overline s\big)=1$ implies $x_i(u)=1$ on a left neighborhood of $1$.
Hence, for $s<\overline s$ close enough to $\overline s$, the same construction of $\bm u(s)$ and the same lower bound $\underline G$ in \eqref{eq: G lower bound} apply, and $\underline G$ is affine in $s^n$ on a left neighborhood of $\overline s$.
Its (left) slope equals
\[
\frac{\partial}{\partial s^n}\underline G\big(\overline s-\big)
=
1-\frac{\rho_i\big(\overline s\big)}{\overline s^{n}}x_i\big(\rho_i\big(\overline s\big)\big)
=
1-\frac{1}{\overline s^{n}}\leq 0.
\]
Since $\underline G\big(\overline s\big)=G\big(\overline s\big)=1$, linearity and the nonpositive slope imply that $\underline G(s) \geq 1$ for all $s<\overline s$ close enough to $\overline s$, contradicting that $\big(\underline s,\overline s\big)$ is a sag.
Therefore, $x_i$ must be discontinuous at $\rho_i\big(\overline s\big)=1$.
\end{proof}

\begin{lemma}[Discontinuity at the bottom of a sag]
\label{lm: jumps endpoints bot}
Let $\bm x$ be an extremal reduced form.
Consider a sag $\big(\underline s,\overline s\big)$ satisfying $G\big(\underline s\big)=1$ with $\underline s>0$.
Then, $\rho_i\big(\underline s\big) \in D_i\big(\rho_i\big(\underline s\big),\rho_i\big(\overline s\big)\big)$ for all $i$, i.e., $x_i$ is discontinuous at $\rho_i\big(\underline s\big)$ whenever $\rho_i\big(\underline s\big)<\rho_i\big(\overline s\big)$.
\end{lemma}
\begin{proof}
Let $i$ be such that $\rho_i\big(\underline s\big)<\rho_i\big(\overline s\big)$.
Since $\underline s>0$, it must be the case that $\rho_i\big(\underline s\big)>0$.
The argument is identical to the first case in the proof of Lemma \ref{lm: jumps endpoints top}.
Specifically, suppose, towards a contradiction, that $x_i$ is continuous at $\rho_i\big(\underline s\big)$.
Then, the lower bound $\underline G$ is well-defined in a neighborhood of $\underline s$, and we have
\[
\frac{\partial}{\partial s^n}\underline G\big(\underline s-\big) = \frac{\partial}{\partial s^n}\underline G\big(\underline s+\big)
= 1-\frac{\rho_i\big(\underline s\big)}{\underline s^{n}}x_i\big(\rho_i\big(\underline s\big)\big)<0,
\]
where the equality follows from the continuity of $x_i$ at $\rho_i\big(\underline s\big)$, and the inequality follows from the fact that $x_i$ is constant in a right neighborhood of $\rho_i\big(\underline s\big)$, together with the assumption that $\underline G \leq G<1$ on $\big(\underline s,\overline s\big)$.
Since the slope of $\underline{G}$ is continuous at $\underline{s}$, this implies that $G\geq \underline G>1$ on a left neighborhood of $\underline s$, contradicting feasibility.
Therefore, $x_i$ is discontinuous at $\rho_i\big(\underline s\big)$.
\end{proof}

\begin{lemma}[Nontriviality along a sag]
\label{lm: persons gap}
Let $\bm x$ be a feasible reduced form.
Consider a sag $\big(\underline s,\overline s\big)$.
Then, we have (1) $\rho_i\big(\underline s\big) < \rho_i\big(\overline s\big)$ for at least one index $i$; and (2) it is true for at least two such indices whenever $G$ equals $1$ at the endpoints of the sag.
\end{lemma}
\begin{proof}
The first part of this claim has been established in Lemma \ref{lm: aux G}, and so we focus on the second part.

Towards a contradiction, suppose that $\rho_i\big(\underline s\big)<\rho_i\big(\overline s\big)$ but $\rho_j\big(\underline s\big)=\rho_j\big(\overline s\big)$ for all $j \neq i$.
For each such $j \neq i$, since $\rho_j$ is weakly increasing by Claim (2) of Lemma \ref{lm: aux G}, it equals $\rho_j\big(\overline s\big)$ on $\big[\underline s,\overline s\big]$.
Therefore, for all $s \in\big[\underline s,\overline s\big]$, we have
\[
\rho_i(s) = \frac{s^n}{\prod_{j \neq i}\rho_j\big(\overline s\big)}=\frac{\rho_i\big(\overline s\big)}{\overline s^{n}}s^n.
\]
It follows that $G(s)$ equals $\underline G(s)$ defined in \eqref{eq: G lower bound} for all $s \in \big[\underline s,\overline s\big]$.
Clearly, $G$ is a concave function of $s^n$ on the respective interval.
Since $G\big(\underline s\big)=G\big(\overline s\big)=1$, it must be that $G=1$ on $\big(\underline s,\overline s\big)$ due to concavity.
This contradicts the premise of the claim.
\end{proof}

\subsubsection{Proofs of Theorems \ref{thm: extremality} and \ref{thm: score allocations}: Statement I}

\begin{proof}
Let $\bm x$ be an extremal reduced form.
As explained above, there are two mutually exclusive cases.
We now analyze each in detail, starting with the latter.

\textit{Case 2 with $\overline s>\bar\psi(0)$.}
Assume $G<1$ on $\big[\bar\psi(0),\overline s\big)$ and $G=1$ on $\big[\overline s,1\big]$ for some $\overline s>\bar\psi(0)$.
As explained in Lemma \ref{lm: aux G}, we also know that $G$ is strictly increasing on $\big[0,\bar\psi(0)\big)$, and thus $G<1$ on that interval.

To simplify notation, set $\overline{\bm v}=\bm\rho\big(\overline s\big)$ and $\underline{\bm v}=\bm\psi(0)$.
Since $\overline s>0$, we have $\prod_{i=1}^n\overline v_i>0$.
Lemmas \ref{lm: piecewisec and jumps} and \ref{lm: jumps endpoints top} jointly imply that, for each $i$ with $\underline v_i<\overline v_i$, the function $x_i$ is piecewise-constant on this interval with finitely many jumps.
Furthermore, the upper endpoint is a discontinuity point, i.e., $\overline v_i \in D_i\big(\underline v_i,\overline v_i\big)$.
Since $x_i = 0$ on $\big[0,\underline v_i\big)$ by construction, it follows that the discontinuity points up to $\overline v_i$ are the same whether one starts at $\underline v_i$ or $0$, that is
\[
D_i\big(\underline v_i,\overline v_i\big) = D_i\big(0,\overline v_i\big).
\]

\textit{Step 1.}
To begin, we show that there exists an index $i$ for which the discontinuity set $D_i\big(0,\overline v_i\big)$ contains at least two elements.
Towards a contradiction, suppose that
\[
D_i\big(0,\overline v_i\big) \subseteq \big\{\overline v_i\big\} \;\;\forall i.
\]
This means that $x_i = 0$ on $\big[0,\overline v_i\big)$.
Since, by the definition of the transformation $\psi_i$, we have $\psi_i(0) = \sup\{u \in [0,1]: ux_i(u) = 0\}$, the upper endpoint $\overline v_i$ is not higher than $\psi_i(0)$ for all $i$.
As a result,
\[
\overline s = \sqrt[n]{\prod_{i=1}^n \overline v_i}\leq \bar\psi(0),
\]
contradicting the assumption that $\overline s > \bar\psi(0)$.

\textit{Step 2.}
Next, fix $i$ such that $D_i\big(0,\overline v_i\big)$ contains at least two elements, which exists due to the argument above.
To rule out the second case with $\overline s>\bar\psi(0)$, we construct a perturbation $\bm\Delta: [0,1] \to \mathbf{R}^n$ so that $\bm x \pm \varepsilon \bm\Delta$ is a feasible reduced form provided $\varepsilon > 0$ is small enough.
As in the proof of Lemma \ref{lm: piecewisec and jumps}, the existence of such $\bm\Delta$ contradicts extremality.

To this end, let $\underline v^*$, $\overline v^*$ be the last two points in $D_i\big(0,\overline v_i\big)$, where $\overline v^* = \overline v_i$ as explained above.
Denote the value of
\[
\max_{\bm u \in [0,1]^n} B_{\bm x}(\bm u) \;\;\text{s.t.}\;\;\bm u \in \prod_{j=1}^n \big(D_j\big(0,\overline v_j\big) \cup \big\{0,\overline v_j\big\}\big),\;\; \bm u \neq \overline{\bm v}
\]
by $\omega$.
Due to finiteness of bidders' discontinuity sets and $G<1$ on $\big[0,\overline s\big)$, $\omega$ is well-defined and strictly below $1$.

First of all, for each $j \neq i$, set $\Delta_j = 0$.
For coordinate $i$, set
\[
\Delta_i(u) = \kappa \cdot \frac{\bm 1_{\big[\underline v^*,\overline v^*\big)}(u)}{\overline v^*-\underline v^*}
\]
for some $\kappa>0$ small enough to ensure $x_i\pm\Delta_i\in \mathscr{X}$.
For example, the choice analogous to \eqref{eq: kappa jumps} will do.

Next, pick any $\varepsilon \in (0,1)$ such that $\varepsilon \leq \overline\varepsilon$, where
\[
\overline\varepsilon = \frac{1-\omega}{\max_{u_i\in\big[\underline v^*,\overline v^*\big]}\left|\int_{u_i}^{\overline v^*}\Delta_i(u)du\right|}>0.
\]
Since $\bm x$ is piecewise-constant on $\big[0,\overline{\bm v}\big]$, the mapping
\[
u_j \in \big[0,\overline v_j\big] \mapsto B_{\bm x \pm \varepsilon \bm\Delta}(\bm u)
\]
is piecewise linear, holding other coordinates $(u_1,\ldots,u_{j-1},u_{j+1},\ldots,u_n)$ fixed in $\prod_{k \neq j}\big[0,\overline v_k\big]$.
As a result, \eqref{eq: Border reduced top} is equivalent to
\begin{align}
B_{\bm x \pm \varepsilon \bm\Delta}(\bm u) \leq 1 \;\;\forall\bm u\in\prod_{j=1}^{n}\big(D_j\big(0,\overline v_j\big) \cup \big\{0,\overline v_j\big\}\big).
\label{eq: Border A}
\end{align}
By construction, our perturbation $\bm\Delta$ verifies \eqref{eq: Border A} at all grid points, including $\overline{\bm v}$ (due to the definition of $\bm\Delta$), and hence \eqref{eq: Border reduced top} holds as well.
But then Lemma \ref{lm: perturbations top} implies that $\bm x \pm \varepsilon \bm\Delta$ is feasible, which contradicts extremality of $\bm x$.

To sum up, we have shown that Case 2 with $\overline s>\bar\psi(0)$ is inconsistent with extremality.
We now turn to Case 1, where we consider a sag $\big(\underline s,\overline s\big)$ so that $G$ equals $1$ at its endpoints.
As before, we denote $\overline{\bm v}=\bm\rho\big(\overline s\big)$ and $\underline{\bm v}=\bm\rho\big(\underline s\big)$, which satisfy $\prod_{i=1}^n\overline v_i>0$ due to $\overline s>0$ and $\prod_{i=1}^n\underline v_i>0$ whenever $\underline s>0$.
Below, we study $\underline s > 0$ and $\underline s = 0$ separately, as Lemma \ref{lm: jumps endpoints bot} does not cover the latter.

\textit{Case 1 with $\underline s>0$.}
By Lemmas \ref{lm: piecewisec and jumps}, \ref{lm: jumps endpoints top}, and \ref{lm: jumps endpoints bot}, for each $i$ with $\underline v_i<\overline v_i$, the function $x_i$ is piecewise-constant on this interval with finitely many jumps, including $\underline v_i$, $\overline v_i$.

By Lemma \ref{lm: persons gap}, we can find two distinct coordinates $i \neq j$ such that $\underline v_i<\overline v_i$ and $\underline v_j<\overline v_j$.
Following the logic of Step 2, we shall construct a perturbation to contradict extremality of $\bm x$ assumed in the premise.

Let $\underline v_i^*$, $\overline v_i^*$ and $\underline v_j^*$, $\overline v_j^*$ be the last two points in $D_i\big(\underline v_i,\overline v_i\big)$ and $D_j\big(\underline v_j,\overline v_j\big)$, respectively, where $\overline v_i^* = \overline v_i$ and $\overline v_j^* = \overline v_j$.
As before, denote the value of
\[
\max_{\bm u\in[0,1]^n}B_{\bm x}(\bm u)\;\;\text{s.t.}\;\;\bm u \in \prod_{k=1}^n D_k\big(\underline v_k,\overline v_k\big),\;\bm u \neq \underline{\bm v}, \overline{\bm v}
\]
by $\omega$, which is again well-defined and strictly less than $1$ due to $G<1$ on $\big(\underline s,\overline s\big)$.

Now, define a perturbation $\bm\Delta:[0,1]\to\mathbf{R}^n$ by setting $\Delta_k = 0$ for all indices $k$ other than $i,j$.
For coordinates $i$ and $j$, set
\[
\Delta_i(u) = \kappa \cdot \frac{\bm 1_{\big[\underline v_i^*,\overline v_i^*\big)}(u)}{\overline v_i^*-\underline v_i^*},\;\;\Delta_j(u) = -\kappa \cdot \frac{\bm 1_{\big[\underline v_j^*,\overline v_j^*\big)}(u)}{\overline v_j^*-\underline v_j^*}
\]
for some $\kappa>0$ small enough to ensure $x_i\pm\Delta_i \in \mathscr{X}$ and $x_j\pm\Delta_j\in \mathscr{X}$.
For example, the minimum of the two values defined analogously to \eqref{eq: kappa jumps} for coordinates $i$ and $j$ will suffice.

Note that
\[
\int_{\underline v_i}^{\overline v_i}\Delta_i(u)du+\int_{\underline v_j}^{\overline v_j}\Delta_j(u)du=0,
\]
which implies
\begin{align}
B_{\bm x \pm \varepsilon \bm\Delta}(\bm u)\leq 1\;\;\forall\bm u\in\prod_{k=1}^{n}D_k\big(\underline v_k,\overline v_k\big),
\label{eq: Border B 2}
\end{align}
for all $\varepsilon\in(0,1)$ such that $\varepsilon \leq \overline\varepsilon$, where
\[
\overline\varepsilon = \frac{1}{2}\frac{1-\omega}{\max\left\{
\max_{u_i\in\big[\underline v_i^*,\overline v_i^*\big]}\left|\int_{u_i}^{\overline v_i^*}\Delta_i(u)du\right|,
\max_{u_j\in\big[\underline v_j^*,\overline v_j^*\big]}\left|\int_{u_j}^{\overline v_j^*}\Delta_j(u)du\right|\right\}}>0.
\]
Since each function $x_k$ is piecewise-constant on $\big[\underline v_k,\overline v_k\big]$, the argument in Case 2 implies that \eqref{eq: Border B 2} is equivalent to \eqref{eq: Border reduced top and bot}.
We conclude that $\bm x \pm \varepsilon \bm\Delta$ is feasible due to Lemma \ref{lm: perturbations top and bot}, which contradicts extremality of $\bm x$.

\textit{Case 1 with $\underline s=0$.}
Since $\underline s = 0$, we have $G(0) = 1$, which means $\bar\psi(0) = 0$ because $G<1$ on $\big[0,\bar\psi(0)\big)$ due to Lemma \ref{lm: aux G} and the proof of Theorem \ref{thm: feasibility}.
Since $x_i = 0$ on $\big[0,\underline v_i\big)$, the discontinuity points up to $\overline v_i$ are the same whether one starts at $\underline v_i$ or $0$, exactly as happens in Case 2.
By the first step in the proof of that case, the discontinuity set $D_i\big(0,\overline v_i\big)$ contains at least two points for some index $i$.
Note also that, for any $\bm u\leq\overline{\bm v}$ with $\prod_{j=1}^n u_j=0$ and $\bm u\not\leq\underline{\bm v}$, we have $u_j>\underline v_j$ for some $j$ and $x_j>0$ above $\underline v_j$, so that
\[
B_{\bm x}(\bm u)=\sum_{j=1}^n\int_{u_j}^1x_j(u)du<\sum_{j=1}^n\int_{\underline v_j}^1x_j(u)du=1 .
\]
Furthermore, by the argument of Case 1 with $\underline s>0$, there cannot be two such indices; hence, we have
\[
D_j\big(0,\overline v_j\big) \subseteq \big\{\overline v_j\big\} \;\;\forall j \neq i.
\]

We proceed in two logical steps.
First, we argue that $D_i\big(0,\overline v_i\big)$ must contain at least three elements, and then we complete the proof by ruling out this possibility as well.

\textit{Step 1.}
Towards a contradiction, suppose that $D_i\big(0,\overline v_i\big)$ contains exactly two discontinuity points, where one of them is necessarily $\overline v_i$ due to Lemma \ref{lm: jumps endpoints top}.
Let $\underline v_i^* < \overline v_i$ be the other discontinuity of $x_i$.

Suppose first that $\underline v_i^*>0$ is strictly positive.
Then, for each sufficiently small $s > 0$ (any value strictly less than $\sqrt[n]{\underline v_i^* \prod_{j \neq i}\overline v_j}$ will do), we have
\[
G(s) = s^n + \sum_{j=1}^n \int_0^1 x_j(u)du
\]
because $x_j = 0$ on $\big[0,\overline v_j\big)$ for all $j$ other than $i$ and $x_i = 0$ on $\big[0,\underline v_i^*\big)$.
But this implies that
\[
\sum_{j=1}^n \int_0^1 x_j(u)du < 1,
\]
contradicting the fact that $G\big(\underline s\big) = 1$, since $\bar\psi(0) = \underline s = 0$.

It follows that we must have $\underline v_i^* = 0$, meaning that $x_i$ jumps at $u = 0$ and $u = \overline v_i$.
Let $\chi$ denote the size of the first jump, that is, the constant value of $x_i$ on the whole interval $\big[0,\overline v_i\big)$.
Recall that $G$ equals $1$ at the endpoints of the sag $\big(\underline s,\overline s\big)$, which can be expressed as
\begin{align}
\label{eq: extra conditions s=0}
\overline v_i \cdot \chi + \sum_{j=1}^n \int_{\overline v_j}^1 x_j(u)du = \overline v_i \cdot \prod_{j \neq i}\overline v_j + \sum_{j=1}^n \int_{\overline v_j}^1 x_j(u)du = 1.
\end{align}
One immediate implication of \eqref{eq: extra conditions s=0} is that the jump size of $x_i$ satisfies
\[
\chi = \prod_{j \neq i}\overline v_j.
\]

Take any $s \in \big(0,\overline s\big)$ and consider $\bm u$ that equals $\overline{\bm v}$ for all coordinates but
\[
u_i  = \frac{\overline v_i}{\overline s^{n}} s^n.
\]
Evaluate Border's constraint at this point to obtain
\[
B_{\bm x}(\bm u) = s^n + \big(\overline v_i-u_i\big)\cdot \chi +  \sum_{j=1}^n \int_{\overline v_j}^1 x_j(u)du = 1,
\]
where the last equality follows from the definition of $\chi$, $u_i$, and \eqref{eq: extra conditions s=0}.
Since $\bm u$ is feasible in the problem defining $G(s)$ in Lemma \ref{lm: aux G}, we must have $G(s) = 1$, which contradicts the assumption that $G<1$ on $\big(\underline s,\overline s\big)$, as $\underline s=0$.

\textit{Step 2.}
We have established that there are at least three elements in $D_i\big(0,\overline v_i\big)$.
To conclude the proof, we shall again invoke a perturbational argument that has been extensively used earlier in this proof as well as to establish the auxiliary lemmas.

Let $\underline v_i^* < \overline v_i^* < \overline v_i$ be the last three such points, where the fact that $x_i$ jumps at the right endpoint is due to Lemma \ref{lm: jumps endpoints top}.
Denote the value of the following program
\[
\max_{\bm u\in[0,1]^n}B_{\bm x}(\bm u)\;\;\text{s.t.}\;\;\bm u \in \prod_{j=1}^n \big(D_j\big(0,\overline v_j\big) \cup \big\{0,\overline v_j\big\}\big),\;\bm u \not\leq \underline{\bm v},\; \bm u \neq \overline{\bm v}
\]
by $\omega$.
It is well-defined because each discontinuity set is finite, and it is strictly less than one because $G<1$ on the interior of the sag and, as noted above, $B_{\bm x}<1$ at grid points with a zero coordinate.
Grid points $\bm u \leq \underline{\bm v}$ are excluded, since $B_{\bm x}$ equals $1$ there.

Consider a perturbation $\bm\Delta:[0,1]\to\mathbf{R}^n$ used in the proof of the second part of Lemma \ref{lm: piecewisec and jumps}, that is, $\Delta_j = 0$ for all $j\neq i$, whereas, for coordinate $i$, we have
\[
\Delta_i(u) = \kappa \cdot\left(\frac{\bm 1_{\big[\underline v_i^*,\overline v_i^*\big)}(u)}{\overline v_i^*-\underline v_i^*} - \frac{\bm 1_{\big[\overline v_i^*,\overline v_i\big)}(u)}{\overline v_i-\overline v_i^*}\right)
\]
for some $\kappa>0$ small enough to ensure $\Delta_i\neq 0$ and $x_i\pm\Delta_i\in \mathscr{X}$.
Again, the choice analogous to \eqref{eq: kappa jumps} will do.

By construction,
\[
\int_{\underline v_i^*}^{\overline v_i}\Delta_i(u)du=0.
\]
Hence, the perturbation preserves Border's constraint at $\overline{\bm v}$ and at any grid point $\bm u \leq \underline{\bm v}$, where $u_i \leq \underline v_i \leq \underline v_i^*$.
In particular, we obtain
\begin{align}
B_{\bm x \pm \varepsilon \bm\Delta}(\bm u)\leq 1\;\;\forall\bm u\in\prod_{j=1}^{n}\big(D_j\big(0,\overline v_j\big) \cup \big\{0,\overline v_j\big\}\big)
\label{eq: Border B 1}
\end{align}
for every $\varepsilon\in(0,1)$ such that $\varepsilon \leq \overline\varepsilon$, where
\[
\overline\varepsilon = \frac{1-\omega}{\max_{u_i\in\big[\underline v_i^*,\overline v_i\big]}\left|\int_{u_i}^{\overline v_i}\Delta_i(u)du\right|}>0.
\]
Similarly to Case 2, since each function $x_j$ is piecewise-constant on $\big[0,\overline v_j\big]$, \eqref{eq: Border B 1} is equivalent to \eqref{eq: Border reduced top}.
Thus, by Lemma \ref{lm: perturbations top}, both $\bm x \pm \varepsilon \bm\Delta$ are feasible, contradicting extremality of $\bm x$.

\textit{Putting all pieces together.}
To sum up, we have established that, if $\bm x$ is an extremal reduced form, then we cannot have any sag $\big(\underline s,\overline s\big)$ with $\overline s>\bar\psi(0)$.
But this means that
\begin{align}
\label{eq: G identity}
G(s)=1 \;\;\forall s \in \big[\bar\psi(0),1\big],
\end{align}
which, as explained in the proof of Theorem \ref{thm: feasibility}, is equivalent to Border's constraint being identically $1$ along the whole principal curve.

As explained in Appendix \ref{sec: appendix delta}, the $\psi$-transform can be inverted to obtain the underlying CDF in $\mathscr{X}$.
It is easy to see that $\overline x = {\bar\psi}^{-1}(u)/u$ is the CDF corresponding to $\bar\psi$.
Using $\overline x$, we can unpack \eqref{eq: G identity} for any $s \geq \bar\psi(0)$ as
\begin{align*}
\frac{1 - s^n}{n}
& = \int_{\bar\psi^{-1}(s)}^1 \iota d\ln \bar\psi(\iota) \\
& = \int_{\bar\psi\big(\bar\psi^{-1}(s)\big)}^1 \overline x(u)du\\
& = \int_{s}^1 \overline x(u)du,
\end{align*}
where the first line uses \eqref{eq: G explicit} in Lemma \ref{lm: aux G}, the second line is due to the integral identity in the proof of that lemma, and the third line follows from continuity of $\bar\psi$.
But this means that $\overline x(u)$ equals $u^{n-1}$ above $\bar\psi(0)$ and is zero below that threshold.
This is equivalent to $\bar\psi(\iota)=\max\left\{\bar\psi(0),\sqrt[n]{\iota}\right\}$, or \eqref{eq: extremality delta} in the theorem, since $\bar\psi(e^{-t}) = e^{-\frac{1}{n}\sum_{i=1}^n\delta_i(t)}$.
\end{proof}

\subsubsection{Proofs of Theorems \ref{thm: extremality} and \ref{thm: score allocations}: Statement II}
\begin{proof}
Take a feasible reduced form $\bm x$ that satisfies \eqref{eq: extremality delta}, which, as shown at the end of the proof of Statement I, is equivalent to $\bar\psi(\iota)=\max\left\{\bar\psi(0),\sqrt[n]{\iota}\right\}$ in terms of the $\psi$-transforms $\bm\psi = (\psi_1,\ldots,\psi_n)$.
By the definition of the $\psi$-transform and its bijectivity established in Appendix \ref{sec: appendix delta}, we have $ux_i(u) = \psi_i^{-1}(u)$ for all bidders $i$, that is, the canonical scores of $\bm x$ are $q_i(u) = \psi_i^{-1}(u)$.
These scores are weakly increasing and right-continuous, and they satisfy \eqref{eq: winner q} because each $\psi_i^{-1}$ is strictly increasing wherever positive, making a tie at a positive score a null event.
Thus, the score allocation with the canonical scores is well-defined.

Fix an index $i$ and let $\tilde x_i$ be the interim winning probability of that bidder induced by the score allocation with these scores, that is,
\[
\widetilde x_i(u)=
\begin{cases}
\prod_{j\neq i}\mathbf{E}\left[\bm 1_{\big[0,\psi_i^{-1}(u)\big)}\left(\psi_j^{-1}(u_j)\right)\right]
&\;\;\text{if}\;\; u\geq\psi_i(0),
\\
0
&\;\;\text{if}\;\;u<\psi_i(0).
\end{cases}
\]
We shall show that $x_i$ and $\tilde x_i$ coincide.
Clearly, both these functions equal $0$ on $[0,\psi_i(0))$, and so it suffices to consider their values on the complement of that set.

Let $u\geq\psi_i(0)$.
Since, for each $j$, by definition, $\psi_j^{-1}$ is zero on $[0,\psi_j(0))$ and strictly increasing on its complement, we have
\[
\tilde x_i(u)=\prod_{j\neq i}\psi_j\left(\psi_i^{-1}(u)\right).
\]
Multiply each side by $u=\psi_i\left(\psi_i^{-1}(u)\right)$ to obtain
\begin{align*}
u\widetilde x_i(u)
&= \prod_{j=1}^n\psi_j\left(\psi_i^{-1}(u)\right)\\
&= \left[\bar\psi\left(\psi_i^{-1}(u)\right)\right]^{n}\\
&= \max\left\{\left[\bar\psi(0)\right]^{n},\psi_i^{-1}(u)\right\} \\
&= ux_i(u),
\end{align*}
where the second line follows from the definition of $\bar\psi$, and the third line is due to $\bar\psi(\iota)=\max\left\{\bar\psi(0),\sqrt[n]{\iota}\right\}$.
Finally, the last line is due to $ux_i(u)=\psi_i^{-1}(u)$ and
\[
\psi_i^{-1}(u)\geq \psi_i^{-1}\big(\psi_i(0)\big) \geq\left[\bar\psi(0)\right]^{n}.
\]
To see the last inequality, note that $\psi_i(\iota) = \psi_i(0)$ for $\iota \leq \left[\bar\psi(0)\right]^{n}$ due to monotonicity of $(\psi_1,\ldots,\psi_n)$.
\end{proof}

\subsubsection{Proofs of Theorems \ref{thm: extremality} and \ref{thm: score allocations}: Statement III}
\begin{proof}
This directly follows from the argument in Section \ref{sec: score allocations} establishing that every score allocation induces an extremal reduced form and is the a.e.-unique allocation inducing it.
\end{proof}

\subsection{Optimality}
\label{sec: appendix optimality}

\subsubsection{Proof of Proposition \ref{prop: control reformulation}}
\label{sec: proof control reformulation}
\begin{proof}

Using $H_i(0,u) = 0$ and the $\psi$-transform of $x_i$ denoted $\psi_i$, we can obtain
\begin{align*}
\int_0^1 H_i(x_i(u),u)du
&= \int_0^1\int_0^{ux_i(u)}\frac{\partial}{\partial x_i}H_i\left(\frac{\iota}{u},u\right)d\iota d\ln u
\\
&= \int_0^1\int_{\psi_i(\iota)}^1\frac{\partial}{\partial x_i}H_i\left(\frac{\iota}{u},u\right)d\ln ud\iota
\\
&= \int_0^\infty e^{-t}R_i(\delta_i(t),t)dt,
\end{align*}
where the first line uses the change of variables $x_i = \iota/u$,  and the second one follows from the change of the order of integration. 
The last line is due to the change of variables $\iota = e^{-t}$, $u = e^{-\tau}$, and the definition of $R_i$ in \eqref{eq: revenue delta}.

Next, by Theorem \ref{thm: feasibility} and \eqref{eq: border along curve}, a reduced form is feasible if and only if its $\delta$-transforms satisfy the constraint in \eqref{eq: control problem} at all times.
Since the $\delta$-transform is a bijection between $\mathscr{X}$ and $\mathscr{D}$ by Lemma \ref{lm: delta transform}, the first part of the proposition follows.
For the second part, Theorem \ref{thm: extremality} says that a feasible reduced form is extremal if and only if $\sum_{i=1}^n\delta_i(t) = \min\{t,T\}$ for some $T \in [0,\infty]$.
As $\bm\delta(0) = \bm 0$ and $\bm\delta$ is absolutely continuous, this is \eqref{eq: delta front} almost everywhere.
\end{proof}

\subsubsection{Proof of Theorem \ref{thm: optimal no exclusion}}
\label{sec: proof optimal no exclusion}
\begin{proof}
The proof proceeds in two steps.
First, we show that the equi-MR path $\bm\delta^{\rm s}$ is the unique solution to the optimal control problem in Proposition \ref{prop: control reformulation}.
Second, we establish that the corresponding reduced form is extremal, has no exclusion, and satisfies \eqref{eq: fixed point H}.

\textit{Step 1.}
By the implicit function theorem applied to \eqref{eq: static system}, $p^{\rm s}$ is continuously differentiable on each of the at most $n$ time intervals on which the set of active bidders is constant.
At a time at which a bidder becomes active, the system of equations changes, and $p^{\rm s}$ may have a kink.
Since $p^{\rm s}$ is strictly decreasing and nonnegative, $\lambda := -(p^{\rm s})'$ is nonnegative, piecewise continuous with at most $n$ jumps, and integrable with
\[
\int_t^\infty\lambda(\tau)d\tau = p^{\rm s}(t) - p^{\rm s}(\infty),
\]
where $p^{\rm s}(\infty) \geq 0$ denotes the limit of $p^{\rm s}$ as $t \to \infty$.
For any $\bm\delta$ that is feasible in the optimal control problem of Proposition \ref{prop: control reformulation}, define
\begin{align}
\label{eq: cs inequality oc 1}
\begin{split}
\mathscr{C}(\bm\delta) := {}& \int_0^\infty\left[1 - e^{-\sum_{i=1}^n\delta_i(t)} - \int_0^t e^{-\tau}\sum_{i=1}^n\delta_i'(\tau)d\tau\right]\lambda(t)dt \\
&+ p^{\rm s}(\infty)\left[1 - \int_0^\infty e^{-\tau}\sum_{i=1}^n\delta_i'(\tau)d\tau\right],
\end{split}
\end{align}
and remark that
\[
\mathscr{C}(\bm\delta) \geq 0,
\]
because each term in the square brackets is nonnegative---the first by the constraint of the control problem at time $t$, and the second by the same constraint in the limit as $t \to \infty$.
To proceed, we rewrite $\mathscr{C}(\bm\delta)$ solely in terms of $\bm\delta$ rather than its derivatives.
First, integrating by parts and using $\bm\delta(0) = \bm 0$ together with $e^{-t}\sum_{i=1}^n\delta_i(t) \to 0$, we can obtain
\begin{multline}
\label{eq: cs inequality oc 2}
\mathscr{C}(\bm\delta) = \int_0^\infty\left(1 - e^{-\sum_{i=1}^n\delta_i(t)} - e^{-t}\sum_{i=1}^n\delta_i(t) - \int_0^t e^{-\tau}\sum_{i=1}^n\delta_i(\tau)d\tau\right)\lambda(t)dt
\\
+ p^{\rm s}(\infty)\left(1 - \int_0^\infty e^{-\tau}\sum_{i=1}^n\delta_i(\tau)d\tau\right).
\end{multline}
Second, changing the order of integration in the terms involving $\int_0^t e^{-\tau}\sum_{i=1}^n\delta_i(\tau)d\tau$ and using the relationship between  $\lambda$ and $p^{\rm s}$, we can further rewrite \eqref{eq: cs inequality oc 2} as
\begin{align}
\label{eq: cs inequality oc 3}
\mathscr{C}(\bm\delta) = \int_0^\infty\left(\left(1 - e^{-\sum_{i=1}^n\delta_i(t)} - e^{-t}\sum_{i=1}^n\delta_i(t)\right)\lambda(t) - e^{-t}\sum_{i=1}^n\delta_i(t)p^{\rm s}(t)\right)dt + p^{\rm s}(\infty).
\end{align}

By construction, the expression in \eqref{eq: cs inequality oc 3} is nonnegative for all feasible policies irrespective of their extremality or optimality; therefore,
\begin{multline}
\label{eq: bound extremality 1}
\int_0^\infty e^{-t}\sum_{i=1}^n R_i(\delta_i(t),t)dt
\leq \int_0^\infty e^{-t}\sum_{i=1}^n R_i(\delta_i(t),t)dt + \mathscr{C}(\bm\delta)
\\
= p^{\rm s}(0) + \int_0^\infty e^{-t}\left[\sum_{i=1}^n R_i(\delta_i(t),t) - \left(e^{t-\sum_{i=1}^n\delta_i(t)} + \sum_{i=1}^n\delta_i(t)\right)\lambda(t) - \sum_{i=1}^n\delta_i(t)p^{\rm s}(t)\right]dt.
\end{multline}
It is easy to see that $\mathscr{C}(\bm\delta^{\rm s}) = 0$; therefore, the bound in \eqref{eq: bound extremality 1} is actually tight for the equi-MR path.
We now make the bound more permissive while keeping its tightness when evaluated at $\bm\delta^{\rm s}$.
Specifically, since the term in the square brackets in \eqref{eq: bound extremality 1} is strictly concave in $\bm\delta(t)$, this concave function lies below its tangent at $\bm\delta^{\rm s}(t)$.
Combining this gradient inequality with $\left(-e^{t-\sum_{i=1}^n\delta_i^{\rm s}(t)} + 1\right)\lambda(t) = 0$, which is due to our construction of $\bm\delta^{\rm s}$ and $\lambda$, we obtain
\begin{multline}
\label{eq: bound extremality 2}
\int_0^\infty e^{-t}\sum_{i=1}^n R_i(\delta_i(t),t)dt + \mathscr{C}(\bm\delta) \leq \int_0^\infty e^{-t}\sum_{i=1}^n R_i(\delta_i^{\rm s}(t),t)dt
\\
+ \int_0^\infty e^{-t}\sum_{i=1}^n\left[\frac{\partial}{\partial\delta_i}R_i(\delta_i^{\rm s}(t),t) - p^{\rm s}(t)\right]\left(\delta_i(t) - \delta_i^{\rm s}(t)\right)dt.
\end{multline}
The term in the square brackets is zero for every active bidder due to the definition of the equi-MR path in \eqref{eq: static system}, whereas, for a bidder with $\delta_i^{\rm s}(t) = 0$, it is nonpositive and multiplies $\delta_i(t) - \delta_i^{\rm s}(t) = \delta_i(t) \geq 0$.
Hence, the second line of \eqref{eq: bound extremality 2} is nonpositive for every feasible policy and vanishes at $\bm\delta^{\rm s}$, thereby establishing optimality of this policy in the optimal control problem.

As for uniqueness, any optimal policy attains equality throughout \eqref{eq: bound extremality 1} and \eqref{eq: bound extremality 2}.
Since the function in the second line of \eqref{eq: bound extremality 1} is strictly concave in $\bm\delta$ under Condition (A), equality in the gradient inequality forces $\bm\delta(t) = \bm\delta^{\rm s}(t)$ for almost every $t$, hence for every $t$ by continuity.

\textit{Step 2.}
Let $\bm x^{\rm s}$ be the reduced form that corresponds to $\bm\delta^{\rm s}$, i.e., by \eqref{eq: delta transform inverse}, we have
\begin{align}
\label{eq: x* extremal}
ux_i^{\rm s}(u) = e^{-(\delta_i^{\rm s})^{-1}(-\ln u)}.
\end{align}
By construction, $\bm\delta^{\rm s}$ satisfies $\sum_{i=1}^n\delta_i^{\rm s}(t) = t$ at all times, meaning that $\bm x^{\rm s}$ is extremal due to Theorem \ref{thm: extremality}.
Moreover, since some bidder's coordinate grows without bound, the score allocation with the canonical scores \eqref{eq: x* extremal} awards the good with probability one at almost every type profile, so the auction has no exclusion.
In view of Theorem \ref{thm: score allocations}, $\bm x^{\rm s}$ is induced by awarding the good to bidder $i$ whenever
\[
u_ix_i^{\rm s}(u_i) \geq u_jx_j^{\rm s}(u_j) \;\;\forall j \neq i.
\]
Equivalently, using \eqref{eq: x* extremal} and applying $p^{\rm s}$, which is strictly decreasing, to each side, we can rewrite this condition as the requirement that bidder $i$'s MR
\[
p^{\rm s}\left((\delta_i^{\rm s})^{-1}(-\ln u_i)\right) = \frac{\partial}{\partial x_i}H_i(x_i^{\rm s}(u_i),u_i)
\]
be the highest one, where the equality is due to \eqref{eq: static system}.
Since $p^{\rm s} \geq 0$, this is exactly \eqref{eq: fixed point H}, with the types below the exclusion types covered by the argument in Footnote \ref{fn: below exclusion} of the main text.
\end{proof}

\subsubsection{Proof of Proposition \ref{prop: optimum vs myerson}}
\label{sec: proof optimum vs myerson}
\begin{proof}
Let $\bm\delta^{\rm M}$ and $\bm\delta^*$ be the $\delta$-transforms of Myerson's reduced form and the optimum.
We begin with two details deferred in the main text.
First, Myerson's auction has no exclusion, since bidders' MRs are their MVVs times $h' > 0$, MRs and MVVs have the same sign, so Lemma \ref{lm: positive common MR} applies to both auctions alike.
Second, \eqref{eq: optimum ev myerson} holds for all bidders, including those who are inactive in Myerson's auction.
For a bidder $i$ that has left the top, it holds by the MVV equalization, $\zeta_1\big(e^{-\delta_1^{\rm M}}\big) = \zeta_i\big(e^{-\delta_i^{\rm M}}\big) \geq 0$, together with $\delta_1^{\rm M} \geq \delta_i^{\rm M}$ and monotonicity of $h'$.
For a bidder $i$ who is inactive in Myerson's auction at time $t$, we have $\delta_i^{\rm M}(t) = 0$ and $\zeta_i(1) \leq \zeta_1\big(e^{-\delta_1^{\rm M}(t)}\big)$. Since again $h'\big(e^{\delta_1^{\rm M}(t)-t}\big) \geq h'(e^{-t})$, \eqref{eq: optimum ev myerson} follows.

Equipped with these two facts, we are ready to complete the argument in the main text.
Suppose, towards a contradiction, that there exists some time $t$ so that $\delta_1^{\rm M}(t) > \delta_1^{*}(t)$.
Then, there must exist a bidder $i$ for whom the inequality is reversed at that time, that is, $\delta_i^{\rm M}(t) < \delta_i^{*}(t)$, because $\sum_{i=1}^n\delta_i^{\rm M}(t) = \sum_{i=1}^n\delta_i^{*}(t) = t$.
In particular, $\delta_i^{*}(t) > 0$, so bidder $i$ is active at the optimum and, by \eqref{eq: static system} and strict monotonicity of MRs in own coordinate under Condition (A),
\[
p^{\rm s}(t) = \frac{\partial}{\partial\delta_i}R_i(\delta_i^{*}(t),t) < \frac{\partial}{\partial\delta_i}R_i(\delta_i^{\rm M}(t),t) \leq \frac{\partial}{\partial\delta_1}R_1(\delta_1^{\rm M}(t),t) < \frac{\partial}{\partial\delta_1}R_1(\delta_1^{*}(t),t),
\]
where the middle inequality is \eqref{eq: optimum ev myerson}.
But \eqref{eq: static system} requires $\frac{\partial}{\partial\delta_1}R_1(\delta_1^{*}(t),t) \leq p^{\rm s}(t)$ whether bidder $1$ is active or inactive, a contradiction.
We conclude that $\delta_1^{\rm M}(t) \leq \delta_1^{*}(t)$ for all $t$.

Finally, note that $\delta_1^{*}$ being larger than $\delta_1^{\rm M}$ implies that bidder $1$'s interim winning probability at the optimum $x_1^*$ is larger than in Myerson's auction due to \eqref{eq: delta transform inverse}.
The argument for the weakest bidder is identical and hence omitted.
\end{proof}

\subsubsection{Proof of Theorem \ref{thm: optimal concave}}
\label{sec: proof optimal concave}
\begin{proof}
The proof of Theorem \ref{thm: optimal no exclusion} extends almost as is, and we focus on the differences that arise with $T^{\rm s}<\infty$.
Set $\lambda := -(p^{\rm s})'$ on $[0,T^{\rm s})$ and $\lambda := 0$ afterward, so that $\int_t^\infty\lambda(\tau)d\tau = \max\{p^{\rm s}(t),0\}$, and define $\mathscr{C}(\bm\delta)$ by the first term of \eqref{eq: cs inequality oc 1} alone.
The chain \eqref{eq: cs inequality oc 1}--\eqref{eq: bound extremality 2} goes through verbatim with $p^{\rm s}$ replaced by $\max\{p^{\rm s},0\}$, and $\mathscr{C}(\bm\delta^*) = 0$, since $\sum_{i=1}^n\delta_i^*(t) = t$ on $[0,T^{\rm s})$, where $\lambda$ is supported.
The term in the square brackets in \eqref{eq: bound extremality 2} is now
\[
\frac{\partial}{\partial\delta_i}R_i(\delta_i^*(t),t) - \max\{p^{\rm s}(t),0\}.
\]
This term is zero for active bidders on $[0,T^{\rm s})$ and nonpositive at $\delta_i^*(t) = 0$.
On $[T^{\rm s},\infty)$, it is bidder $i$'s MR at the unconstrained optimum in \eqref{eq: optimal x concave}, which is zero whenever $0<\delta_i^*(t) < t$, nonpositive at $\delta_i^*(t) = 0$, where it multiplies $\delta_i(t)\geq 0$, and nonnegative at $\delta_i^*(t) = t$, where it multiplies $\delta_i(t) - t \leq 0$.
Hence, the second line of \eqref{eq: bound extremality 2} is again nonpositive for every feasible policy and vanishes at $\bm\delta^*$, and uniqueness follows as before.
To sum up, we have established that
\[
\int_0^\infty e^{-t}\sum_{i=1}^n R_i(\delta_i^*(t),t)dt \geq \int_0^\infty e^{-t}\sum_{i=1}^n R_i(\delta_i(t),t)dt
\]
for every $\bm\delta$ feasible in \eqref{eq: control problem}.

It remains to show that $\bm\delta^*$, and hence the reduced form $\bm x^*$, is itself feasible.
On $[0,T^{\rm s})$, $\bm\delta^*$ coincides with the equi-MR path, which is in $\mathscr{D}^n$ by Condition (B) and satisfies the constraint in \eqref{eq: control problem} with equality, since $\sum_{i=1}^n\delta_i^{\rm s}(t) = t$.
On $[T^{\rm s},\infty)$, $\delta_i^*(t)$ is the maximizer of $R_i(\cdot,t)$ on $[0,t]$, that is, the unconstrained optimum in \eqref{eq: optimal x concave} along the level set $u_ix_i = e^{-t}$; it solves $\frac{\partial}{\partial x_i}H_i\big(e^{\delta_i^*(t)-t},e^{-\delta_i^*(t)}\big) = 0$ whenever this equation has a root in $(0,t)$, and equals $0$ or $t$ otherwise, with $(\delta_i^*)' \in \{0,1\}$ there.
In the former case, differentiating with respect to $t$ gives
\[
(\delta_i^*)' = \frac{\frac{\partial}{\partial\ln x_i}\left(\frac{\partial H_i}{\partial x_i}\right)}{\frac{\partial}{\partial\ln x_i}\left(\frac{\partial H_i}{\partial x_i}\right) - \frac{\partial}{\partial\ln u_i}\left(\frac{\partial H_i}{\partial x_i}\right)} \in [0,1]
\]
at $(x_i,u_i) = \big(e^{\delta_i^*(t)-t},e^{-\delta_i^*(t)}\big)$ whenever $0 < \delta_i^*(t) < t$, as the numerator is nonpositive by concavity of $H_i$ in $x_i$ and the denominator is negative by Condition (A).
The two pieces agree at $T^{\rm s}$, where $p^{\rm s}(T^{\rm s}) = 0$ makes the equi-MR condition and the unconstrained optimum coincide.
Hence, $\bm\delta^* \in \mathscr{D}^n$.
Furthermore, since $\delta_i^* \leq \delta_i^{\rm s}$ on $[T^{\rm s},\infty)$ due to $0 \geq p^{\rm s}$ there and strict monotonicity of MRs in own coordinate, we have $x_i^* \leq x_i^{\rm s}$.
Hence, $\bm x^*$ is feasible along with $\bm x^{\rm s}$, and $r_i^* = x_i^*/x_i^{\rm s} \in [0,1]$ under the convention $\sfrac{0}{0} = 0$.
\end{proof}

\subsubsection{Proof of Theorem \ref{thm: optimal convex}}
\label{sec: proof optimal convex}
\begin{proof}
The proof proceeds in three steps.
First, we show that the optimal reduced form is unique and extremal.
Second, we show that it is the $\underline{\bm u}^*$-capped equi-MR reduced form for a maximizer $\underline{\bm u}^*$ of \eqref{eq: revenue given T}.
Third, we derive \eqref{eq: FOC exclusion times}.

\textit{Step 1.}
We first record two consequences of our assumptions.
Since $H_i(0,\cdot) = 0$ and $\frac{\partial H_i}{\partial x_i}$ is weakly increasing in $u_i$ by supermodularity,
\[
H_i(x_i,u_i) = \int_0^{x_i}\frac{\partial}{\partial x_i}H_i(s,u_i)ds
\]
is weakly increasing in $u_i$.
Since $\max\{0,H_i\}$ is convex in $x_i$, nonnegative, and vanishes at $x_i = 0$, for $x_i < y_i$ we have
\[
\max\{0,H_i(x_i,u_i)\} \leq \frac{x_i}{y_i}\max\{0,H_i(y_i,u_i)\} \leq \max\{0,H_i(y_i,u_i)\},
\]
so $\max\{0,H_i\}$ is weakly increasing in $x_i$.

For a feasible reduced form $\bm x$, let $\tilde{\bm x}$ withhold the good whenever $H_i \leq 0$, that is, $\tilde x_i$ is the right-continuous modification of $x_i\cdot\bm 1_{(0,\infty)}(H_i(x_i(u),u))$.
Due to the two properties of $H_i$ derived above and monotonicity of each $x_i$, $\tilde{\bm x}$ is a reduced form, and it is feasible, as it is dominated by $\bm x$.
Its revenue satisfies
\[
\sum_{i=1}^n\int_0^1 H_i(\tilde x_i(u),u)du = \sum_{i=1}^n\int_0^1\max\{0,H_i(x_i(u),u)\}du \geq \sum_{i=1}^n\int_0^1 H_i(x_i(u),u)du.
\]
Next, the optimal reduced form $\bm x^*$ is unique, since the objective in \eqref{eq: control problem} is strictly concave in $\bm\delta$ under Condition (A), its feasible set is convex, and $\bm\delta \mapsto \bm x$ is a bijection by Lemma \ref{lm: delta transform}.
It follows that $\bm x^*$ maximizes the convex functional $\sum_{i=1}^n\int_0^1\max\{0,H_i(x_i(u),u)\}du$, and every other maximizer of this functional is at least as large as $\bm x^*$, because it truncates to $\bm x^*$ by uniqueness.
Since a maximizer of a convex functional that is dominated by every other maximizer is extremal, we conclude that $\bm x^*$ is extremal.

\textit{Step 2.}
Let $\underline{\bm u}^*$ be the exclusion types of $\bm x^*$.
By Step 1 and Theorem \ref{thm: extremality}, the $\delta$-transforms of $\bm x^*$ satisfy $\delta_i^*(t) \in \big[0,-\ln\underline u_i^*\big]$ for all $i$ and $\sum_{i=1}^n\delta_i^*(t) = \min\big\{t,-\sum_{j=1}^n\ln\underline u_j^*\big\}$ for all $t$, and the constraint in \eqref{eq: control problem} then holds automatically.
Hence, by Proposition \ref{prop: control reformulation}, the auctioneer's problem can be expressed as
\begin{multline}
\label{eq: extremal problem}
\max_{\underline{\bm u} \in [0,1]^n}\;\max_{\bm\delta \in \mathscr{D}^n}\int_0^\infty e^{-t}\sum_{i=1}^n R_i(\delta_i(t),t)dt
\;\;\text{s.t.}\;\;
\delta_i(t) \in \big[0,-\ln\underline u_i\big] \;\;\forall i,
\\
\sum_{i=1}^n\delta_i(t) = \min\Big\{t,-\sum_{j=1}^n\ln\underline u_j\Big\} \;\;\forall t \in \mathbf{R}_+.
\end{multline}

Now, consider the relaxation of the inner problem in which the requirement $\bm\delta \in \mathscr{D}^n$ is dropped.
This relaxed problem separates across $t$, and its value at $t$ is
\[
W\big(t,\underline{\bm u}\big) := \max_{\bm\delta \in \mathbf{R}^n}\sum_{i=1}^n R_i(\delta_i,t)
\;\;\text{s.t.}\;\;
\delta_i \in \big[0,-\ln\underline u_i\big] \;\;\forall i,
\;\;\sum_{i=1}^n\delta_i = \min\Big\{t,-\sum_{j=1}^n\ln\underline u_j\Big\}.
\]
Let $p$, $\lambda_i \geq 0$, and $\mu_i \geq 0$ denote the multipliers on the equality constraint, the lower bound $\delta_i \geq 0$, and the cap $\delta_i \leq -\ln\underline u_i$, respectively.
The first-order conditions read
\[
\frac{\partial}{\partial\delta_i}R_i(\delta_i,t) - p + \lambda_i - \mu_i = 0,
\qquad
\lambda_i\delta_i = 0,
\qquad
\mu_i\big(\delta_i + \ln\underline u_i\big) = 0
\;\;\forall i.
\]
By \eqref{eq: R derivative}, these are exactly \eqref{eq: capped static system}, and, since the objective is strictly concave under Condition (A), they are also sufficient.
Hence, the unique solution is $\bm\delta^{\rm s}\big(t|\underline{\bm u}\big)$, the multiplier on the equality constraint is $p^{\rm s}\big(t|\underline{\bm u}\big)$, and the multiplier on bidder $i$'s cap is
\[
\mu_i = \frac{\partial}{\partial\delta_i}R_i\big(-\ln\underline u_i,t\big) - p^{\rm s}\big(t|\underline{\bm u}\big)
\;\;\text{if}\;\; t > T_i\big(\underline{\bm u}\big),
\qquad
\mu_i = 0 \;\;\text{otherwise}.
\]

Condition (B$^*$) is exactly what makes the relaxation tight.
Under this condition, each coordinate of $\bm\delta^{\rm s}\big(\cdot|\underline{\bm u}\big)$ is weakly increasing, and, since the coordinates sum to $\min\big\{t,-\sum_{j=1}^n\ln\underline u_j\big\}$, each is then $1$-Lipschitz, so $\bm\delta^{\rm s}\big(\cdot|\underline{\bm u}\big) \in \mathscr{D}^n$.
Hence, the inner problem in \eqref{eq: extremal problem} is uniquely solved by the $\underline{\bm u}$-capped equi-MR path, and its value $V(\underline{\bm u})$ equals \eqref{eq: revenue given T} by \eqref{eq: revenue delta}.
Thus, $\underline{\bm u}^*$ maximizes \eqref{eq: revenue given T}, $\bm x^*$ is the $\underline{\bm u}^*$-capped equi-MR reduced form by uniqueness of the inner solution, and every other maximizer yields the same reduced form by uniqueness of $\bm x^*$.
The last claim of the theorem follows from Theorem \ref{thm: score allocations}.

\textit{Step 3.}
It remains to derive \eqref{eq: FOC exclusion times}.
Differentiating \eqref{eq: revenue given T} with respect to $\underline u_i$ gives
\[
\frac{\partial }{\partial\underline u_i}V(\underline{\bm u}) = -H_i\big(x_i^{\rm s}\big(\underline u_i|\underline{\bm u}\big),\underline u_i\big)
+ \sum_{j=1}^n\int_{\underline u_j}^1\frac{\partial}{\partial x_j}H_j\big(x_j^{\rm s}\big(u_j|\underline{\bm u}\big),u_j\big)\frac{\partial}{\partial\underline u_i}x_j^{\rm s}\big(u_j|\underline{\bm u}\big)du_j,
\]
and \eqref{eq: FOC exclusion times} is the first-order condition for $\underline{\bm u}^*$ multiplied by $\underline u_i^*$ once we show that the second term equals
\begin{align}
\label{eq: rearrangement term}    
\frac{1}{\underline u_i}\int_{T_i(\underline{\bm u})}^{T(\underline{\bm u})} e^{-t}p^{\rm s}\big(t|\underline{\bm u}\big)dt.
\end{align}

To this end, we compute $\frac{\partial V}{\partial\underline u_i}$ a second time, now from \eqref{eq: revenue delta}, which represents $V\big(\underline{\bm u}\big)$ as $\int_0^\infty e^{-t}\sum_{j=1}^n R_j\big(\delta_j^{\rm s}\big(t|\underline{\bm u}\big),t\big)dt$.
By the implicit function theorem applied to \eqref{eq: capped static system}, $\delta_j^{\rm s}\big(t|\underline{\bm u}\big)$ is differentiable in $\underline u_i$ at every $t$ at which the set of bidders with $\delta_j^{\rm s}\big(t|\underline{\bm u}\big) \in \big(0,-\ln\underline u_j\big)$ is locally constant, which excludes at most finitely many $t$ under Condition (B$^*$).
Hence,
\begin{align}
\label{eq: V derivative delta}
\frac{\partial }{\partial\underline u_i}V(\bm u)
= \int_0^\infty e^{-t}\sum_{j=1}^n\frac{\partial}{\partial\delta_j}R_j\big(\delta_j^{\rm s}\big(t|\underline{\bm u}\big),t\big)\frac{\partial}{\partial\underline u_i}\delta_j^{\rm s}\big(t|\underline{\bm u}\big)dt.
\end{align}
Below, we evaluate the derivatives of $\bm\delta^{\rm s}\big(\cdot|\underline{\bm u}\big)$ on various time intervals.

\begin{itemize}
\item
\textit{Case $t < T_i\big(\underline{\bm u}\big)$.}
By construction, $\underline u_i$ does not enter \eqref{eq: capped static system}, and thus $\frac{\partial}{\partial\underline u_i}\delta_j^{\rm s}\big(t|\underline{\bm u}\big) = 0$ for all $j$.
\item
\textit{Case $t > T\big(\underline{\bm u}\big)$.}
Here, $\delta_j^{\rm s}\big(t|\underline{\bm u}\big) = -\ln\underline u_j$ for all $j$, and thus $\frac{\partial}{\partial\underline u_i}\delta_i^{\rm s}\big(t|\underline{\bm u}\big) = -\frac{1}{\underline u_i}$ and $\frac{\partial}{\partial\underline u_i}\delta_j^{\rm s}\big(t|\underline{\bm u}\big) = 0$ for $j \neq i$.
\item
\textit{Case $t \in \big(T_i\big(\underline{\bm u}\big),T\big(\underline{\bm u}\big)\big)$.}
As in the previous case, $\delta_i^{\rm s}\big(t|\underline{\bm u}\big) = -\ln\underline u_i$, and thus $\frac{\partial}{\partial\underline u_i}\delta_i^{\rm s}\big(t|\underline{\bm u}\big) = -\frac{1}{\underline u_i}$.
For $j \neq i$, we necessarily have
\[
\left(\frac{\partial}{\partial\delta_j}R_j\big(\delta_j^{\rm s}\big(t|\underline{\bm u}\big),t\big) - p^{\rm s}\big(t|\underline{\bm u}\big)\right)\frac{\partial}{\partial\underline u_i}\delta_j^{\rm s}\big(t|\underline{\bm u}\big) = 0
\]
due to the first-order conditions.
Since $\sum_{j \neq i}\delta_j^{\rm s}\big(t|\underline{\bm u}\big) = t - \delta_i^{\rm s}\big(t|\underline{\bm u}\big)$, we obtain
\[
\sum_{j \neq i}\frac{\partial}{\partial\delta_j}R_j\big(\delta_j^{\rm s}\big(t|\underline{\bm u}\big),t\big)\frac{\partial}{\partial\underline u_i}\delta_j^{\rm s}\big(t|\underline{\bm u}\big)
= p^{\rm s}\big(t|\underline{\bm u}\big)\sum_{j \neq i}\frac{\partial}{\partial\underline u_i}\delta_j^{\rm s}\big(t|\underline{\bm u}\big)
= \frac{p^{\rm s}\big(t|\underline{\bm u}\big)}{\underline u_i}.
\]
\end{itemize}
Substituting the three cases into \eqref{eq: V derivative delta}, we obtain
\[
\frac{\partial}{\partial\underline u_i}V\big(\underline{\bm u}\big)
= -\frac{1}{\underline u_i}\int_{T_i(\underline{\bm u})}^\infty e^{-t}\frac{\partial}{\partial\delta_i}R_i\big(-\ln\underline u_i,t\big)dt
+ \frac{1}{\underline u_i}\int_{T_i(\underline{\bm u})}^{T(\underline{\bm u})} e^{-t}p^{\rm s}\big(t|\underline{\bm u}\big)dt.
\]
By \eqref{eq: R derivative}, the first integrand equals $e^{-t}\frac{\partial}{\partial x_i}H_i\big(\frac{e^{-t}}{\underline u_i},\underline u_i\big)$, so the change of variables $x = \frac{e^{-t}}{\underline u_i}$, which maps $T_i\big(\underline{\bm u}\big)$ to $x_i^{\rm s}\big(\underline u_i|\underline{\bm u}\big)$ by the definition of the exclusion time, turns the first term into $-H_i\big(x_i^{\rm s}\big(\underline u_i|\underline{\bm u}\big),\underline u_i\big)$.
Comparing with the first display of this step establishes \eqref{eq: rearrangement term}.
\end{proof}

\subsubsection{Proof of Proposition \ref{prop: CRA discrimination}}
\label{sec: proof CRA discrimination}

\begin{proof}
As explained in the text, it is without loss of generality to look at group-symmetric policies, i.e., $\delta_1^{\rm s},\delta_n^{\rm s}$ and $\delta_1^*,\delta_n^*$, corresponding to the risk-neutral and risk-averse bidders.
By Theorem \ref{thm: optimal concave}, $\delta_i^*$ coincides with $\delta_i^{\rm s}$ before $T^{\rm s}$ and is given by the unconstrained optimum in \eqref{eq: optimal x concave} after that threshold.
As shown in the proof of that theorem, $\delta_n^*$ is weakly increasing on $[T^{\rm s},\infty)$, whereas $\delta_1^*$ stays at $\delta_1^{\rm s}(T^{\rm s})$ there, since $H_1$ is linear in $x_1$.
We prove the proposition in two steps.
First, we shall show that both $\delta_1^{\rm s}(t),\delta_n^{\rm s}(t) \lesseqgtr t + \ln m$ if and only if $t \gtreqless t^o$, where $t^o = -\frac{n}{n-1}\ln m < T^{\rm s}$.
Then, we shall rank $\delta_1^{\rm s}(t),\delta_n^{\rm s}(t)$ (as well as $\delta_1^*(t),\delta_n^*(t)$) before $t^o$ and after that time, and use this ranking to complete the proof.

\textit{Step 1.}
Note first that $\delta_1^{\rm s}(t)$ and $\delta_n^{\rm s}(t)$ cannot lie on opposite sides of the line $t+\ln m$ with at least one of them strictly.
Indeed, if $\delta_1^{\rm s}(t)\leq t+\ln m<\delta_n^{\rm s}(t)$, then $g'\big(e^{\delta_n^{\rm s}(t)-t}\big)>1$ due to unimodality, and hence
\[
\frac{\partial}{\partial\delta_n}R_n(\delta_n^{\rm s}(t),t)
< v\big(e^{-\delta_n^{\rm s}(t)}\big)-v'\big(e^{-\delta_n^{\rm s}(t)}\big)\big(1-e^{-\delta_n^{\rm s}(t)}\big)
< \frac{\partial}{\partial\delta_1}R_1(\delta_1^{\rm s}(t),t),
\]
where the second inequality is due to regularity and $\delta_n^{\rm s}(t)>\delta_1^{\rm s}(t)$, contradicting the MR equalization in \eqref{eq: CRA system}.
The case $\delta_n^{\rm s}(t)\leq t+\ln m<\delta_1^{\rm s}(t)$ is ruled out symmetrically, since then $g'\big(e^{\delta_n^{\rm s}(t)-t}\big)\leq 1$ and both inequalities are reversed.

Next, both $\delta_1^{\rm s}(t),\delta_n^{\rm s}(t)$ are strictly above $t + \ln m$ for $t$ sufficiently close to $0$, as $t + \ln m \to \ln m < 0$ as $t \to 0$.
On the other hand, at least one of them is strictly below $t + \ln m$ for $t$ sufficiently large, as
\[
\min\{\delta_1^{\rm s}(t),\delta_n^{\rm s}(t)\} \leq \frac{k}{n}\delta_1^{\rm s}(t) + \left(1-\frac{k}{n}\right)\delta_n^{\rm s}(t) = \frac{t}{n}.
\]
By continuity, the minimum of $\delta_1^{\rm s},\delta_n^{\rm s}$ crosses the line $t + \ln m$ at some time, say $t^o$.
By the observation above, $\delta_1^{\rm s}(t^o)=\delta_n^{\rm s}(t^o)=t^o+\ln m$, so that $g'\big(e^{\delta_n^{\rm s}(t^o)-t^o}\big)=g'(m)=1$ and, by the second line of \eqref{eq: CRA system}, $\delta_1^{\rm s}(t^o)=\delta_n^{\rm s}(t^o)=\frac{t^o}{n}$.
Clearly, this solves to the unique crossing time $t^o = -\frac{n}{n-1}\ln m$.
Consequently, both $\delta_1^{\rm s}(t),\delta_n^{\rm s}(t)$ are strictly above $t+\ln m$ for $t<t^o$ and strictly below it for $t>t^o$.
This crossing time is below $T^{\rm s}$, since
\[
p^{\rm s}(t^o) = v\left(m^{\sfrac{1}{(n-1)}}\right)-v'\left(m^{\sfrac{1}{(n-1)}}\right)\left(1-m^{\sfrac{1}{(n-1)}}\right) = \frac{\partial}{\partial x_n}H_n\left(m,m^{\sfrac{1}{(n-1)}}\right) > 0,
\]
by assumption.

\textit{Step 2.}
Let $t < t^o$.
Since $\delta_n^{\rm s}(t)>t+\ln m$, unimodality gives $g'\big(e^{\delta_n^{\rm s}(t)-t}\big)>1$, and the MR equalization yields
\[
\frac{\partial}{\partial\delta_1}R_1(\delta_1^{\rm s}(t),t) = \frac{\partial}{\partial\delta_n}R_n(\delta_n^{\rm s}(t),t) < \frac{\partial}{\partial\delta_1}R_1(\delta_n^{\rm s}(t),t),
\]
which implies $\delta_1^{\rm s}(t) > \delta_n^{\rm s}(t)$ due to regularity.
For $t\in(t^o,T^{\rm s})$, we have $\delta_n^{\rm s}(t)<t+\ln m$, so $g'\big(e^{\delta_n^{\rm s}(t)-t}\big)<1$, both inequalities are reversed, and $\delta_1^{\rm s}(t) < \delta_n^{\rm s}(t)$.
This ranking persists on $[T^{\rm s},\infty)$, where $\delta_1^*$ is constant and $\delta_n^*$ is weakly increasing.
Putting all these pieces together, we conclude: $\delta_1^*(t)>\delta_n^*(t)$ for $t<t^o$, $\delta_1^*(t)=\delta_n^*(t)$ for $t=t^o$, and $\delta_1^*(t)<\delta_n^*(t)$ for $t>t^o$.
Finally, using \eqref{eq: delta transform inverse}, we can obtain $x_1^*(u)>x_n^*(u)$ for $u>e^{-\delta_i^*(t^o)}=m^{\sfrac{1}{(n-1)}}$, $x_1^*(u)<x_n^*(u)$ for $\overline u_n \leq u<m^{\sfrac{1}{(n-1)}}$, and $x_1^*(u)=x_n^*(u)(=m)$ for $u=m^{\sfrac{1}{(n-1)}}$, as claimed.
Below $\overline u_n < \overline u_1$, $x_1^*$ vanishes, as $\delta_1^*$ never exceeds $\delta_1^{\rm s}(T^{\rm s})$.
\end{proof}

\subsubsection{Derivations for Example \ref{ex: CRA}}
\label{sec: appendix CRA example}

We provide a derivation for the case with two risk-neutral bidders and one risk-averse bidder ($k=2$, $n=3$).
For the case $k=1$, the derivation is analogous and is omitted.

By MR equalization in \eqref{eq: CRA system}, we have
\begin{equation}
\label{eq: CRA ex MR}
2e^{-\delta^{\rm s}_1}-1=e^{-\delta^{\rm s}_n}-(1-e^{-\delta^{\rm s}_n})2e^{\delta^{\rm s}_n-t}.
\end{equation}
As $2\delta_1^{\rm s}+\delta_n^{\rm s}= t$, we have $\delta_1^{\rm s}=(t-\delta_n^{\rm s})/2$.
Plugging this in \eqref{eq: CRA ex MR}, we obtain
\begin{equation}
\label{eq: CRA ex delta n}
e^{-\delta_n^{\rm s}}-(1-e^{-\delta_n^{\rm s}})2e^{\delta_n^{\rm s}-t}=2e^{(\delta_n^{\rm s}-t)/2}-1.
\end{equation}
This is equivalent to a fourth-degree polynomial equation in $e^{-\delta_n^{\rm s}/2}$.
However, instead of solving for $\delta_n^{\rm s}$, we can solve directly for $x_n^{\rm s}$, which is the object we are ultimately interested in.
Indeed, $e^{-\delta_n^{\rm s}(t)}=\psi_n(e^{-t})$, so $e^{-t}=\psi_n^{-1}(e^{-\delta_n^{\rm s}(t)})$.
But notice that $\psi_i^{-1}(u) = ux_i(u)$.
Thus, solving \eqref{eq: CRA ex delta n} for $e^{-t}$ after writing $e^{-\delta_n^{\rm s}}=u$, we will recover $ux_n^{\rm s}(u)$.
Thus, $x_n^{\rm s}(u)$ is the solution to
\begin{equation}
\label{eq: CRA ex x n solve}
u-(1-u)2x_n^{\rm s}=2\sqrt{x_n^{\rm s}}-1\Leftrightarrow 2x_n^{\rm s}(1-u)+2\sqrt{x_n^{\rm s}}-u-1=0,
\end{equation}
so
\begin{equation}
\label{eq: CRA ex x n}
x_n^{\rm s}(u)=\frac{2-u^2-\sqrt{3-2u^2}}{2(1-u)^2}.
\end{equation}
By Theorem \ref{thm: optimal concave}, the expression in \eqref{eq: CRA ex x n} is optimal only when $\delta^*_n(t)=\delta_n^{\rm s}(t)$, that is, $t\leq T^{\rm s}$.
To find $T^{\rm s}$, we equate both sides of \eqref{eq: CRA ex delta n} to $0$, which yields $\delta_1^{\rm s}(T^{\rm s})=\ln 2$, $\delta_n^{\rm s}(T^{\rm s})-T^{\rm s}=-2\ln 2$, so $e^{-\delta^{\rm s}_n(T^{\rm s})}-(1-e^{-\delta^{\rm s}_n(T^{\rm s})})2(1/4)=0$, $\delta_n^{\rm s}(T^{\rm s})=\ln 3$.
Thus, $T^{\rm s}=2\delta_1^{\rm s}(T^{\rm s})+\delta_n^{\rm s}(T^{\rm s})=\ln 12$.
The expression in \eqref{eq: CRA ex x n} is valid for $u\geq e^{-\delta^{\rm s}_n(T^{\rm s})}=1/3$.

To find $x_1^{\rm s}(u)$, we can do the same by writing \eqref{eq: CRA ex MR} in terms of $\delta_1^{\rm s}$ only.
We get
\[
2e^{-\delta^{\rm s}_1}-1=e^{2\delta^{\rm s}_1-t}-(1-e^{2\delta^{\rm s}_1-t})2e^{-2\delta^{\rm s}_1}.
\]
Writing $e^{-\delta^{\rm s}_1}=u$ and solving for $e^{-t}=ux_1^{\rm s}$, we get
\[
x_1^{\rm s}(u)=\max\left\{\frac{2u+2u^2-1}{2u+1/u},0\right\}.
\]
This is optimal for $u\geq e^{-\delta_1^{\rm s}(T^{\rm s})}=1/2$.

For lower $u$, the optimal allocations are the unconstrained optima in \eqref{eq: optimal x concave}, that is, \[\max\mathrm{argmax}_{x_i \in [0,1]}H_i(x_i,u).\]
For the risk-averse bidder, $H_n(x,u) = ux-(1-u)x^2$ is strictly concave in $x$ with $\frac{\partial}{\partial x_n}H_n(x,u) = u-(1-u)2x$, so, for $u<1/3$, the unconstrained optimum is $x_n^*(u)=\frac{u}{2(1-u)}$.
Thus, by Theorem \ref{thm: optimal concave},
\[
x_n^*(u)=
\begin{cases}
\frac{u}{2(1-u)}, & u<1/3;\\
\frac{2-u^2-\sqrt{3-2u^2}}{2(1-u)^2}, & u\geq 1/3,
\end{cases}
\]
and $x_n^*(u)$ is continuous.

For a risk-neutral bidder, $H_1(x,u) = (2u-1)x$ is strictly decreasing in $x$ for $u<1/2$, so the unconstrained optimum is $x_1^*(u) = 0$ there.
Summing up, the optimal allocation to a risk-neutral bidder is
\[
x_1^*(u)=x_2^*(u)=
\begin{cases}
0, & u<1/2;\\
\frac{2u+2u^2-1}{2u+1/u}, & u\geq 1/2.
\end{cases}
\]
At $u=1/2$, the interim allocation experiences a jump from $0$ to $1/6$.

\subsection{Regularity}
\label{sec: appendix regularity}

This section unpacks the regularity conditions of Sections \ref{sec: main result} and \ref{sec: optimal convex} in terms of primitives.
Recall from \eqref{eq: R derivative} that bidder $i$'s MR along the parametrization underlying the $\delta$-transform is $\frac{\partial}{\partial\delta_i}R_i(\delta_i,t) = \frac{\partial}{\partial x_i}H_i(e^{\delta_i-t},e^{-\delta_i})$.
We track its curvature through two semi-elasticities
\[
\xi_i^x = \frac{\partial}{\partial\ln x_i}\left(\frac{\partial H_i}{\partial x_i}\right) = x_i\frac{\partial^2 H_i}{\partial x_i^2},
\qquad
\xi_i^u = \frac{\partial}{\partial\ln u_i}\left(\frac{\partial H_i}{\partial x_i}\right) = u_i\frac{\partial^2 H_i}{\partial x_i\partial u_i},
\]
evaluated at $(x_i,u_i) = (e^{\delta_i-t},e^{-\delta_i})$; differentiating $R_i$ twice then gives
\begin{align}
\label{eq: R second derivatives}
\frac{\partial^2 R_i}{\partial\delta_i^2}= \xi_i^x - \xi_i^u,
\qquad
\frac{\partial^2 R_i}{\partial\delta_i\partial t} = -\xi_i^x.
\end{align}

\subsubsection{Condition (A)}
By \eqref{eq: R second derivatives}, Condition (A) says exactly that each $R_i$ is strictly concave in its first argument.
This holds automatically when $H_i$ is strictly supermodular and concave in $x_i$, since then $\xi_i^x \leq 0 < \xi_i^u$.
More generally, Condition (A) requires that supermodularity is sufficiently strong relative to convexity.
In Myerson's linear model, $\xi_i^x \equiv 0$, and (A) is equivalent to bidders' MVVs being strictly increasing.

\subsubsection{Condition (B)}
We now assume Condition (A) and examine the added restrictions imposed by Condition (B).
Under strict concavity of $R_i$, the equi-MR path is well-defined.
Recall that bidder $i$ is active at time $t$ if $\delta_i^{\rm s}(t) > 0$.
When (B) holds, a bidder that becomes active stays active forever, so the set of active bidders changes at most $n$ times as time progresses.
By the implicit function theorem, on a time interval on which the set of active bidders is constant, say equal to $I$, we have
\[
(\xi_i^x - \xi_i^u)(\delta_i^{\rm s})' - \xi_i^x = (p^{\rm s})' \;\;\forall i \in I,
\qquad
\sum_{i \in I}(\delta_i^{\rm s})' = 1,
\]
where the semi-elasticities are evaluated along the equi-MR path, whereas $(\delta_i^{\rm s})' = 0$ for $i \notin I$.
Solving this linear system yields
\begin{align}
\label{eq: regularity derivatives}
\Bigg[\underbrace{\sum_{j \in I}\frac{1}{\xi_j^x - \xi_j^u}}_{<0}\Bigg]\cdot (p^{\rm s})' = 1 - \sum_{j \in I}\frac{\xi_j^x}{\xi_j^x - \xi_j^u},
\;\;
\Bigg[\underbrace{\sum_{j \in I}\frac{\xi_i^x - \xi_i^u}{\xi_j^x - \xi_j^u}}_{>0}\Bigg]\cdot (\delta_i^{\rm s})' = 1 - \sum_{j \in I}\frac{\xi_j^x - \xi_i^x}{\xi_j^x - \xi_j^u} \;\;\forall i \in I.
\end{align}

Under Condition (A), the first square bracket in \eqref{eq: regularity derivatives} is negative, whereas the second one is positive.
Hence, the positivity of the right-hand sides in \eqref{eq: regularity derivatives} is sufficient for Condition (B).
This can be expressed as the following requirement:
\begin{align}
\label{eq: regularity elasticity}
\sum_{j \in I}\frac{\xi_j^x - \max\{0,\max_{i \in I}\xi_i^x\}}{\xi_j^x - \xi_j^u} < 1,
\end{align}
which trivially holds in Myerson's linear model, as the $x$-semi-elasticity is identically zero.

More generally, \eqref{eq: regularity elasticity} demands that the semi-elasticities be not too strong and not differ across bidders too much.
For example, consider the following multiplicatively separable class in the spirit of our illustrative example and Section \ref{sec: EV applications}:
\[
H_i(x_i,u_i) = u_i^{1/\beta_i}h_i(x_i)
\]
with $\beta_i > 0$ and strictly increasing $h_i$.
Along the equi-MR path, active bidders share $u_i^{1/\beta_i}h_i'(x_i)= p^{\rm s}$, so $\xi_i^u = p^{\rm s}/\beta_i$ and $\xi_i^x = \left(\frac{x_ih''_i}{h_i'}\right)p^{\rm s}$.
Thus, (A) is $\left(\frac{x_ih''_i}{h_i'}\right) < 1/\beta_i$, and, for convex $h_i$, \eqref{eq: regularity elasticity} becomes
\[
\sum_{i \in I}\frac{\max_{j \in I}\left(\frac{x_jh''_j}{h_j'}\right) - \left(\frac{x_ih''_i}{h_i'}\right)}{1/\beta_i - \left(\frac{x_ih''_i}{h_i'}\right)} < 1.
\]
Here, $\frac{x_ih''_i}{h_i'}$ is the elasticity of bidder $i$'s MR with respect to own winning probability.
Condition (B) adds nothing to Condition (A) when those elasticities coincide across bidders, as happens in the illustrative example.

When each $H_i$ is strictly supermodular and concave in $x_i$, so that (A) holds automatically, all $\xi_j^x$ are nonpositive, and \eqref{eq: regularity elasticity} reads
\[
\sum_{j \in I}\frac{\xi_j^x}{\xi_j^x - \xi_j^u} < 1.
\]
If, in addition, the environment is symmetric, then the equi-MR path is symmetric, that is, $\delta_i^{\rm s}(t) = \sfrac{t}{n}$ and $x_i^{\rm s}(u) = u^{n-1}$, and the condition transforms into
\[
n\frac{\xi^x}{\xi^x - \xi^u} < 1
\]
evaluated at $(x_i,u_i) = (u^{n-1},u)$, nesting the regularity condition in Theorem 1 of \citet{gershkov2022optimal} for the symmetric CRA specification $H(x,u) = v(u)x - (1-u)v'(u)g(x)$.

\subsubsection{Condition (B$^*$)}
We keep assuming Condition (A), so that the $\underline{\bm u}$-capped equi-MR path is well-defined for each profile of exclusion types, $\underline{\bm u}$.
Condition (B$^*$) extends Condition (B) to such capped paths.
Following the same argument as in the previous section, we can obtain that it is verified if, along every $\underline{\bm u}$-capped equi-MR path,
\[
\sum_{j \in I}\frac{\xi_j^x - \max\left\{0,\max_{i \in I}\xi_i^x,\max_{i: \delta_i^{\rm s}\big(t|\underline{\bm u}\big) = -\ln\underline u_i}\xi_i^x\right\}}{\xi_j^x - \xi_j^u} < 1,
\]
that is, \eqref{eq: regularity elasticity} with $I$ the set of active bidders and the inner maximum extended over the bidders at their exclusion levels.
This trivially holds in Myerson's linear model, as the $x$-semi-elasticity is identically zero.
More generally, the added terms require that the common MR decline at least as fast as the MR of each bidder at own exclusion level.

\subsection{The condition $p^{\rm s} \geq 0$ in primitives}
\label{sec: appendix positive MR}

\begin{lemma}
\label{lm: positive common MR}
In a regular environment, the common MR $p^{\rm s}$ is nonnegative if and only if
\begin{align}
\label{eq: positive common MR}
\prod_{i=1}^n\inf\left\{u_i \in [e^{-t},1]: \frac{\partial}{\partial x_i}H_i\left(\frac{e^{-t}}{u_i},u_i\right) \geq 0\right\} \leq e^{-t},
\end{align}
with the convention $\inf\emptyset = 1$.
\end{lemma}

In words, the $i$th factor in \eqref{eq: positive common MR} is the lowest type at which bidder $i$'s MR is still nonnegative on the level set $u_ix_i = e^{-t}$ of the canonical score.
The condition says that the surface $\prod_{i=1}^nu_i = e^{-t}$, on which the principal curve sits at time $t$, contains a point at which every bidder's MR is nonnegative.

\begin{proof}
Fix each $t > 0$, and, for $p \in \mathbf{R}$, let
\[
\delta_i(p) := \sup\left\{\delta_i \in [0,t]: \frac{\partial}{\partial x_i}H_i\left(e^{\delta_i-t},e^{-\delta_i}\right) \geq p\right\},
\]
so that, with $u_i = e^{-\delta_i}$, \eqref{eq: positive common MR} reads $\sum_{i=1}^n\delta_i(0) \geq t$.
By Condition (A), bidder $i$'s MR is continuous and strictly decreasing in $\delta_i$, so $\delta_i(p)$ is weakly decreasing in $p$, and \eqref{eq: static system} gives $\delta_i^{\rm s}(t) = \delta_i(p^{\rm s}(t))$ for all $i$.
If $p^{\rm s}(t) \geq 0$, then
\[
\sum_{i=1}^n\delta_i(0) \geq \sum_{i=1}^n\delta_i(p^{\rm s}(t)) = \sum_{i=1}^n\delta_i^{\rm s}(t) = t.
\]
Conversely, if $p^{\rm s}(t) < 0$, then $\delta_j(0) \leq \delta_j(p^{\rm s}(t)) = \delta_j^{\rm s}(t)$ for every bidder $j$.
For an active bidder $i$ (which exists, since $\sum_{j=1}^n\delta_j^{\rm s}(t) = t > 0$), the inequality is strict.
Hence,
\[
\sum_{j=1}^n\delta_j(0) < \sum_{j=1}^n\delta_j^{\rm s}(t) = t,
\]
so \eqref{eq: positive common MR} fails.
\end{proof}

\end{document}